\documentclass[aps,prresearch,reprint,superscriptaddress]{revtex4-2}
\usepackage{amsmath,amsfonts,amssymb}
\usepackage{mathtools}
\usepackage{multirow}
\usepackage{ulem}
\usepackage{algorithmic}
\usepackage[ruled, vlined, linesnumbered]{algorithm2e}

\usepackage[utf8]{inputenc}

\newcommand{\update}[1]{\textcolor{black}{#1}}
\usepackage{bm}
\usepackage{times}

\usepackage{mathtools}
\usepackage{amsmath}
\usepackage[shortlabels]{enumitem}
\usepackage{graphicx,epic,eepic,epsfig,amsmath,latexsym,amssymb,verbatim,color}
 
\usepackage{amsfonts}       % blackboard math symbols
\usepackage{nicefrac}       % compact symbols for 1/2, etc.

\usepackage{amsmath}
\usepackage{bbm}

\usepackage{float}
\usepackage{tikz}
\usetikzlibrary{chains}
\usetikzlibrary{fit}
\usetikzlibrary{arrows}
\usepackage{epsfig}
\usetikzlibrary{shapes.symbols,patterns} % for source symbols
\usepackage{pgfplots}
\pgfplotsset{compat=1.18}

\usepackage[strict]{changepage}
\usepackage{hyperref}
\hypersetup{colorlinks=true,citecolor=blue,linkcolor=blue,filecolor=blue,urlcolor=blue,breaklinks=true}

\usepackage[marginal]{footmisc}
\usepackage{url}
\usepackage{theorem}

\newtheorem{definition}{Definition}
\newtheorem{proposition}{Proposition}
\newtheorem{lemma}[proposition]{Lemma}

\newtheorem{theorem}[proposition]{Theorem}

\newtheorem{corollary}[proposition]{Corollary}

\def\squareforqed{\hbox{\rlap{$\sqcap$}$\sqcup$}}
\def\qed{\ifmmode\squareforqed\else{\unskip\nobreak\hfil
\penalty50\hskip1em\null\nobreak\hfil\squareforqed
\parfillskip=0pt\finalhyphendemerits=0\endgraf}\fi}
\def\endenv{\ifmmode\;\else{\unskip\nobreak\hfil
\penalty50\hskip1em\null\nobreak\hfil\;
\parfillskip=0pt\finalhyphendemerits=0\endgraf}\fi}
\newenvironment{proof}{\noindent \textbf{{Proof~} }}{\hfill $\blacksquare$}

\newcounter{remark}

\newcounter{example}

\mathchardef\ordinarycolon\mathcode`\:
\mathcode`\:=\string"8000
\def\vcentcolon{\mathrel{\mathop\ordinarycolon}}
\begingroup \catcode`\:=\active
  \lowercase{\endgroup
  \let :\vcentcolon
  }

\usepackage{cleveref}
\usepackage{graphicx}
\usepackage{xcolor}

\RequirePackage[framemethod=default]{mdframed}
\newmdenv[skipabove=7pt,
skipbelow=7pt,
backgroundcolor=darkblue!15,
innerleftmargin=5pt,
innerrightmargin=5pt,
innertopmargin=5pt,
leftmargin=0cm,
rightmargin=0cm,
innerbottommargin=5pt,
linewidth=1pt]{tBox}

\newmdenv[skipabove=7pt,
skipbelow=7pt,
backgroundcolor=blue2!25,
innerleftmargin=5pt,
innerrightmargin=5pt,
innertopmargin=5pt,
leftmargin=0cm,
rightmargin=0cm,
innerbottommargin=5pt,
linewidth=1pt]{dBox}
\newmdenv[skipabove=7pt,
skipbelow=7pt,
backgroundcolor=darkkblue!15,
innerleftmargin=5pt,
innerrightmargin=5pt,
innertopmargin=5pt,
leftmargin=0cm,
rightmargin=0cm,
innerbottommargin=5pt,
linewidth=1pt]{sBox}
\definecolor{darkblue}{RGB}{0,76,156}
\definecolor{darkkblue}{RGB}{0,0,153}
\definecolor{blue2}{RGB}{102,178,255}
\definecolor{darkred}{RGB}{195,0,0}
\newcommand{\nc}{\newcommand}
\nc{\rnc}{\renewcommand}
\nc{\lbar}[1]{\overline{#1}}
\nc{\bra}[1]{\langle#1|}
\nc{\ket}[1]{|#1\rangle}
\nc{\ketbra}[2]{|#1\rangle\!\langle#2|}
\nc{\braket}[2]{\langle#1|#2\rangle}

\nc{\proj}[1]{| #1\rangle\!\langle #1 |}
\nc{\avg}[1]{\langle#1\rangle}
\nc{\rank}{\operatorname{Rank}}
\nc{\smfrac}[2]{\mbox{$\frac{#1}{#2}$}}
\nc{\tr}{\operatorname{Tr}}
\nc{\ox}{\otimes}
\nc{\dg}{\dagger}
\nc{\dn}{\downarrow}
\nc{\cA}{{\cal A}}
\nc{\cB}{{\cal B}}
\nc{\cC}{{\cal C}}
\nc{\cD}{{\cal D}}
\nc{\cE}{{\cal E}}
\nc{\cF}{{\cal F}}
\nc{\cG}{{\cal G}}
\nc{\cH}{{\cal H}}
\nc{\cI}{{\cal I}}
\nc{\cJ}{{\cal J}}
\nc{\cK}{{\cal K}}
\nc{\cL}{{\cal L}}
\nc{\cM}{{\cal M}}
\nc{\cN}{{\cal N}}
\nc{\cO}{{\cal O}}
\nc{\cP}{{\cal P}}
\nc{\cQ}{{\cal Q}}
\nc{\cR}{{\cal R}}
\nc{\cS}{{\cal S}}
\nc{\cT}{{\cal T}}
\nc{\cU}{{\cal U}}
\nc{\cV}{{\cal V}}
\nc{\cX}{{\cal X}}
\nc{\cY}{{\cal Y}}
\nc{\cZ}{{\cal Z}}
\nc{\cW}{{\cal W}}
\nc{\csupp}{{\operatorname{csupp}}}
\nc{\qsupp}{{\operatorname{qsupp}}}
\nc{\var}{{\operatorname{var}}}
\nc{\rar}{\rightarrow}
\nc{\lrar}{\longrightarrow}
\nc{\polylog}{{\operatorname{polylog}}}
\nc{\wt}{{\operatorname{wt}}}
\nc{\av}[1]{{\left\langle {#1} \right\rangle}}
\nc{\supp}{{\operatorname{supp}}}

\nc{\argmin}{{\operatorname{argmin}}}

\def\x{\xi}

\nc{\RR}{{{\mathbb R}}}
\nc{\CC}{{{\mathbb C}}}
\nc{\FF}{{{\mathbb F}}}
\nc{\NN}{{{\mathbb N}}}
\nc{\ZZ}{{{\mathbb Z}}}
\nc{\PP}{{{\mathbb P}}}
\nc{\QQ}{{{\mathbb Q}}}
\nc{\UU}{{{\mathbb U}}}
\nc{\EE}{{{\mathbb E}}}
\nc{\id}{{\operatorname{id}}}

\nc{\CHSH}{{\operatorname{CHSH}}}

\nc{\be}{\begin{equation}}
\nc{\ee}{{\end{equation}}}
\nc{\bea}{\begin{eqnarray}}
\nc{\eea}{\end{eqnarray}}
\nc{\<}{\langle}
\rnc{\>}{\rangle}
\nc{\rU}{\mbox{U}}

\nc{\ob}[1]{#1}

\nc{\SEP}{{\text{\rm SEP}}}
\nc{\NS}{{\text{\rm NS}}}
\nc{\LOCC}{{\text{\rm LOCC}}}
\nc{\PPT}{{\text{\rm PPT}}}
\nc{\EXT}{{\text{\rm EXT}}}
\nc{\Sym}{{\operatorname{Sym}}}

\nc{\ERLO}{{E_{\text{r,LO}}}}
\nc{\ERLOCC}{{E_{\text{r,LOCC}}}}
\nc{\ERPPT}{{E_{\text{r,PPT}}}}
\nc{\ERLOCCinfty}{{E^{\infty}_{\text{r,LOCC}}}}
\nc{\Aram}{{\operatorname{\sf A}}}

\newcommand{\Choi}{Choi-Jamio\l{}kowski }

\usepackage{tikz}
\usepackage{hyperref}
\hypersetup{colorlinks=true,citecolor=blue,linkcolor=blue,filecolor=blue,urlcolor=blue,breaklinks=true}

\makeatletter
\def\grd@save@target#1{%
  \def\grd@target{#1}}
\def\grd@save@start#1{%
  \def\grd@start{#1}}
\tikzset{
  grid with coordinates/.style={
    to path={%
      \pgfextra{%
        \edef\grd@@target{(\tikztotarget)}%
        \tikz@scan@one@point\grd@save@target\grd@@target\relax
        \edef\grd@@start{(\tikztostart)}%
        \tikz@scan@one@point\grd@save@start\grd@@start\relax
        \draw[minor help lines,magenta] (\tikztostart) grid (\tikztotarget);
        \draw[major help lines] (\tikztostart) grid (\tikztotarget);
        \grd@start
        \pgfmathsetmacro{\grd@xa}{\the\pgf@x/1cm}
        \pgfmathsetmacro{\grd@ya}{\the\pgf@y/1cm}
        \grd@target
        \pgfmathsetmacro{\grd@xb}{\the\pgf@x/1cm}
        \pgfmathsetmacro{\grd@yb}{\the\pgf@y/1cm}
        \pgfmathsetmacro{\grd@xc}{\grd@xa + \pgfkeysvalueof{/tikz/grid with coordinates/major step}}
        \pgfmathsetmacro{\grd@yc}{\grd@ya + \pgfkeysvalueof{/tikz/grid with coordinates/major step}}
        \foreach \x in {\grd@xa,\grd@xc,...,\grd@xb}
        \node[anchor=north] at (\x,\grd@ya) {\pgfmathprintnumber{\x}};
        \foreach \y in {\grd@ya,\grd@yc,...,\grd@yb}
        \node[anchor=east] at (\grd@xa,\y) {\pgfmathprintnumber{\y}};
      }
    }
  },
  minor help lines/.style={
    help lines,
    step=\pgfkeysvalueof{/tikz/grid with coordinates/minor step}
  },
  major help lines/.style={
    help lines,
    line width=\pgfkeysvalueof{/tikz/grid with coordinates/major line width},
    step=\pgfkeysvalueof{/tikz/grid with coordinates/major step}
  },
  grid with coordinates/.cd,
  minor step/.initial=.2,
  major step/.initial=1,
  major line width/.initial=2pt,
}
\makeatother

\usepackage{thmtools}
\usepackage{thm-restate}
\usepackage{etoolbox}
\makeatletter
\def\problem@s{}
\newcounter{problems@cnt}

\newcommand{\allproblems}{\problem@s}
\makeatother

\definecolor{beamer}{rgb}{0.2,0.2,0.7}
\definecolor{colorone}{rgb}{1,0.36,0.03}
\definecolor{colortwo}{rgb}{0.4,0.77,0.17}
\definecolor{colorthree}{rgb}{0.01,0.51,0.93}
\definecolor{colorfour}{rgb}{0.47,0.26,0.58}
\definecolor{colorfive}{rgb}{0.12,0.55,0.16}
\usepackage{tcolorbox}
\usepackage{relsize}
\usepackage{graphicx}
\usepackage{booktabs}
\usepackage{amsmath}
\usepackage{tikz}
\usetikzlibrary{quantikz}
\nc{\st}{\text{subject to} \ }
\nc{\supre}{\text{supremum} \ }
\nc{\sdp}{\text{sdp}}
\usepackage{array}

\allowdisplaybreaks
\nc{\ith}[1]{{#1}^\mathrm{th}}

\begin{document}
\title{Power and limitations of distributed quantum state purification}
 
\author{Benchi Zhao}
\affiliation{QICI Quantum Information and Computation Initiative, School of Computing and Data Science, The University of Hong Kong, Pokfulam Road, Hong Kong}

\author{Yu-Ao Chen}
\affiliation{Thrust of Artificial Intelligence, Information Hub, The Hong Kong University of Science and Technology (Guangzhou), Guangdong 511453, China}

\author{Xuanqiang Zhao}
\affiliation{QICI Quantum Information and Computation Initiative, School of Computing and Data Science, The University of Hong Kong, Pokfulam Road, Hong Kong}

\author{Chengkai Zhu}
\affiliation{Thrust of Artificial Intelligence, Information Hub, The Hong Kong University of Science and Technology (Guangzhou), Guangdong 511453, China}

\author{Giulio Chiribella}
\email{giulio@cs.hku.hk}
\affiliation{QICI Quantum Information and Computation Initiative, School of Computing and Data Science, The University of Hong Kong, Pokfulam Road, Hong Kong}
\affiliation{Quantum Group, Department of Computer Science, University of Oxford, Wolfson Building, Parks Road, Oxford, OX1 3QD, United Kingdom}
\affiliation{Perimeter Institute for Theoretical Physics, 31 Caroline Street North, Waterloo, Ontario, Canada}

\author{Xin Wang}
\email{felixxinwang@hkust-gz.edu.cn}
\affiliation{Thrust of Artificial Intelligence, Information Hub, The Hong Kong University of Science and Technology (Guangzhou), Guangdong 511453, China}

\begin{abstract}
\update{Quantum state purification protocols, which mitigate noise by converting multiple copies of noisy quantum states into fewer copies with a lower noise level,   have applications in quantum communication and computation with imperfect devices. Here, we systematically study the task of state purification in distributed quantum systems, demanding that purification be achieved by local operations and classical communication (LOCC). 
We prove that, in the presence of depolarizing noise, no LOCC purification protocol starting from two copies can work blindly for all the states in three important sets: the set of all pure two-qubit states, the set of all two-qubit maximally entangled states, and the Bell basis. 
In stark contrast, we show that a targeted, single-state purification is always achievable in the presence of depolarizing noise, and we provide an explicit analytical LOCC protocol for every given two-qubit state.  For arbitrary finite sets of pure states and arbitrary noise profiles, we develop an optimization-based algorithm that systematically designs LOCC purification protocols, and we demonstrate it through concrete examples. Overall, our results identify both fundamental limitations and practical noise reduction strategies for distributed quantum information processing.}
\end{abstract}

% \date{\today}
\maketitle

\textit{Introduction.}---
Quantum computers promise major speedups in certain computational problems, and are expected to impact a broad range of areas, including cryptography~\cite{shor1999polynomial, rivest1978method}, machine learning~\cite{biamonte2017quantum, lloyd2014quantum}, and quantum chemistry~\cite{feynman2018simulating, cao2019quantum}.  However, the number of qubits in existing hardware implementations is still limited, thereby preventing the solution of practically relevant problems. For example, the largest superconducting processors have several hundreds of physical qubits~\cite{google2025quantum}, while solving real-world problems such as factoring RSA numbers requires millions of qubits~\cite{gidney2025factor}. 

Distributed quantum computing (DQC)~\cite{caleffi2024distributed,peng2020simulating,mitarai2021constructing,piveteau2023circuit} offers a promising route to scale up the size of quantum computations. The idea is to combine the power of multiple quantum processors, by letting them perform joint computations via local operations and classical communication (LOCC)~\cite{chitambar2014everything}, which is essential for entanglement distillation~\cite{bennett1996purification,Devetak2003a,Fang2019c,zhao2021practical} and distribution~\cite{Briegel1998,Dur1999,Azuma2015a,Chen2024d} in quantum networks. While individual processors have limited capacity, the total power of the overall network can facilitate the implementation of advanced quantum algorithms~\cite{fujii2022deep,zhou2023qaoa}.

A remaining challenge, however, is the presence of noise and decoherence, which limits the performance of all realistic quantum computing architectures, including DQC. In practice, quantum systems inevitably interact with their environment, leading to deviations from the ideal scenario of pure states and unitary dynamics. Such deviations impair the computation and limit the accuracy.
To address this problem, several approaches have been developed, including quantum error correction~\cite{knill1997theory,calderbank1996good,knill2000theory,chao2018quantum} and error mitigation~\cite{temme2017error,cai2023quantum,endo2018practical,zhao2024retrieving,zhao2023information}.

These approaches are designed to work in scenarios where a single copy of the noisy quantum state is available. When multiple copies are available, instead, it is sometimes possible to remove part of the noise through quantum state purification~\cite{fiuravsek2004optimal, yao2025protocols, cirac1999optimal, barenco1997stabilization, childs2025streaming}. Formally, a quantum state purification scenario is modelled as follows: a noisy process,  described by a completely positive trace preserving (CPTP) map $\mathcal{N}$ (also known as a {\em quantum channel}),  degrades a pure state $\psi = \proj{\psi}$ into a mixed state $\mathcal{N}(\psi)$~\cite{nielsen2010quantum}.
Then, the purification protocol is applied to approximately restore the initial pure state; mathematically, the protocol is described by a completely positive trace non-increasing (CPTN) map $\mathcal{E}$ (also known as a {\em quantum operation}) that probabilistically yields a higher-fidelity state $\sigma_\psi = \mathcal{E}(\mathcal{N}(\psi)^{\otimes n})/\tr[\mathcal{E}(\mathcal{N}(\psi)^{\otimes n})]$ from $n$ noisy copies of the noisy state $\mathcal{N}(\psi)$.  The condition for a purification protocol to be useful is that the final state $\sigma_\psi$ has a higher fidelity to the initial state $\psi$ than the noisy state $
\mathcal{N} (\psi)$; in formula,  $F(\sigma_\psi, \psi) \geq F(\mathcal{N}(\psi), \psi)$, where $F(\rho,\sigma): =  \|\sqrt \rho  \sqrt \sigma \|_1^2$ denotes quantum fidelity. \update{If the target state $\psi$ is known, the purification is called \textit{targeted state purification}, e.g., entanglement distillation~\cite{bennett1996purification,Devetak2003a,Fang2019c,zhao2021practical} and magic state distillation~\cite{itogawa2025efficient,bravyi2012magic,sales2025experimental}. For a more general setting, if the target state $\psi$ is unknown and is drawn from a state set $\cS$, the purification is called \textit{blind state purification}.}

Most existing purification protocols have considered a global setting, in which arbitrary joint operations are allowed on the $n$-copy state~\cite{fiuravsek2004optimal, yao2025protocols, cirac1999optimal, barenco1997stabilization, childs2025streaming} (see the Supplemental Material~\ref{app:global_purification_protocols} for more discussion).  
In DQC, it is important to consider a scenario where the state $\psi$ and its noisy version $\mathcal{N}(\psi)$ are distributed to multiple processors. As illustrated in FIG.~\ref{fig:setting}, these processors have access only to the assigned subsystems and are restricted to performing LOCC.
Until now, however, this scenario has remained largely unexplored. Is it possible to find distributed purification protocols that purify quantum states using only LOCC? Is quantum entanglement an essential resource for purifying noisy quantum states?

\begin{figure}[h]
\centering
\includegraphics[width=\linewidth]{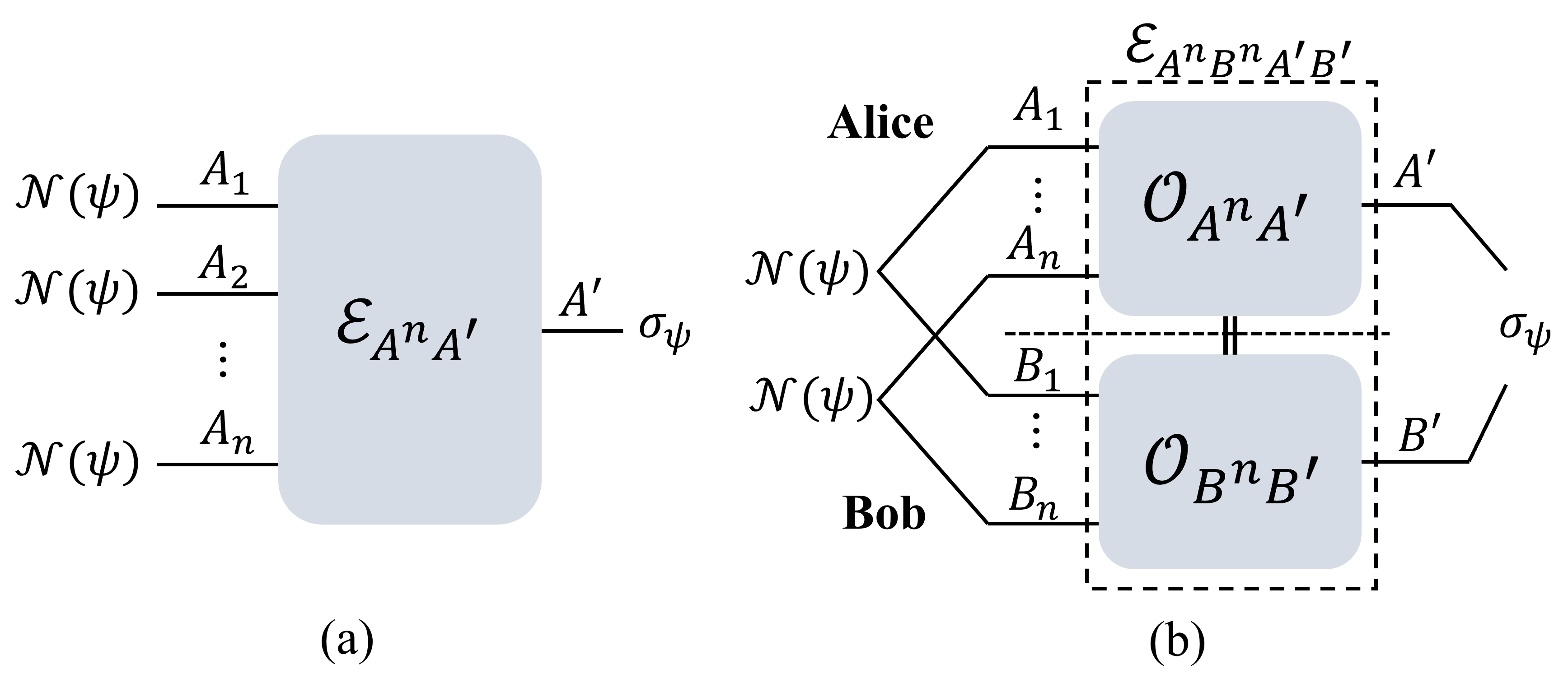}
    \caption{Diagram of $n\rightarrow1$ purification protocol. (a) Global purification protocol with $\cE_{A^nA'}\in {\rm CPTN}$. (b) Distributed purification protocol with $\cE_{A^nB^nA'B'}\in {\rm LOCC}\cap {\rm CPTN}$.}
    \label{fig:setting}
\end{figure}

In this work, we investigate the feasibility of distributed \update{blind state purification} under depolarizing noise, which is a canonical quantum noise process defined as
$\cN^{\gamma}(\psi) = (1-\gamma)\psi + \gamma\frac{I}{d}$,
where $\gamma$ denotes the noise level, $I$ is the identity matrix, and $d$ is the Hilbert space dimension.
\update{Specifically, we consider the scenario where an unknown state $\psi$ is drawn from the set $\cS_P$ of all pure states or from the set $\cS_B$ of four Bell states, and the noisy state $\mathcal{N}^\gamma(\psi)$ is shared between remote parties. We show that no nontrivial $2 \rightarrow 1$ LOCC purification protocol exists in this setting for either set.}
Based on the no-go theorem for the Bell states set, similarly, we arrive at the no-go theorem for the maximally entangled states set $\cS_{MES}$.
Importantly, these $2 \rightarrow 1$ no-go results do not extend to the general $n \rightarrow 1$ setting \update{since, given many copies of the noisy state $\cN^\gamma(\psi)$, the remote parties could implement tomography to determine the state}.
Then, the problem reduces to purifying a known quantum state with LOCC, \update{i.e., targeted state purification. Notably, when the target state is the Bell state, the task reduces to entanglement distillation~\cite{deutsch1996quantum, bennett1996purification, zang2025no, zhao2021practical}. For an arbitrary given state, we construct an analytical $2\rightarrow 1$ LOCC purification protocol}. For the state sets with finite size, we further introduce an optimization-based algorithm for designing LOCC purification protocols. These results establish fundamental limitations and constructive methods for distributed state purification, offering a new avenue for reducing noise in distributed quantum architectures.

%%%%%%%%%%%%%%%%%%%%%%%%%%%%%%%%%%%%%%%%%%%%%%%%%%%%%%%%%%%%%%%%%%%%%%%%
\textit{No-go theorem for distributed purification protocol.}--- 
We start with some notations. Denote a pure state as $\ket{\psi}$, and the corresponding density matrix is denoted as $\psi=\proj{\psi}$. The Bell states are defined as
$\ket{\Phi}^\pm=\frac{1}{\sqrt{2}}(\ket{00}\pm\ket{11})$, and $\ket{\Psi}^\pm=\frac{1}{\sqrt{2}}(\ket{01}\pm\ket{10})$, with corresponding Bell density matrices $\Phi^\pm, \Psi^\pm$, respectively. The distributed $n\rightarrow 1$ purification protocol is denoted as $\cE_{A^nB^nA'B'}$, where $A^nB^n=A_1B_1A_2B_2\cdots A_nB_n$ is the input system, and $A'B'$ is the output system. In this work, we only consider the case of  2-qubit systems. 
% and the dimension for each system is $d=2$.

Suppose we have a noisy state $\cN(\psi)$, where $\psi$ is some pure state and $\cN$ is a given noise channel. Let $\cE$ be an $n\rightarrow 1$ CPTN-LOCC map. Then the output state is
\begin{equation}
    \sigma_\psi\coloneqq\frac{\Tilde{\sigma}_\psi}{p_\psi} = \frac{\cE(\cN(\psi)^{\ox n})}{\tr[\cE(\cN(\psi)^{\ox n})]},
\end{equation}
where $\Tilde{\sigma}_\psi=\cE(\cN(\psi)^{\ox n})$ stands for the unnormalized output state, and $p_\psi=\tr[\Tilde{\sigma}_\psi]$ is the success probability of getting the state $\sigma_\psi$. If the output state $\sigma_\psi$ is closer to the target state $\psi$, i.e., $F(\sigma_\psi, \psi)\ge F(\cN(\psi), \psi)$, then the CPTN-LOCC map $\cE$ is a \textit{distributed purification protocol} for the pure state $\psi$ over noise $\cN$. If the state $\psi$ is sampled from a set of bipartite pure states $\cS$, and the map $\cE$ can increase the fidelity for every state $\psi$ in the set $\cS$, then we say the CPTN-LOCC map $\cE$ is a \textit{distributed purification protocol} for the set $\cS$ over noise $\cN$. The formal definition is given in the following. 

\begin{definition}{\rm \textbf{(Distributed purification protocol)}} \label{def:universal_purification_protocol}
    An $n\rightarrow 1$ CPTN map $\cE$ is a distributed purification protocol for the state set $\cS$ over a noisy channel $\cN$ if the protocol $\cE$ can be realized by LOCC and increases the fidelity for every state in the set $\cS$, i.e.,
    \begin{align}\label{eq:my_purification_definition}
        F(\psi, \sigma_\psi) \ge F(\psi, \cN(\psi)),\; \forall \,\psi\in \cS.
    \end{align}
\end{definition}
There is a special case called \textit{trivial distributed purification protocol}, which refers to the protocol that preserves the fidelity for every state in the set, i.e., $F(\sigma_\psi, \psi) = F(\cN(\psi), \psi),\; \forall \,\psi\in \cS$. For example, the partial trace channel is a trivial purification protocol. On the contrary, the \textit{nontrivial purification protocol} is defined as a purification protocol that can strictly increase the fidelity over some states, which is $F(\sigma_\psi, \psi) > F(\cN(\psi), \psi),\; \exists \,\psi\in \cS$. Obviously, for a given state set $\cS$ and a noisy channel $\cN$, the $n\rightarrow 1$ distributed purification protocol $\cE$ is not unique, and different protocols have different performances. To fairly compare the purification protocols, we define the \textit{average purification fidelity} to quantify the performance of different purification protocols, which only takes the successful cases into account. Specifically, sample quantum states from the state set uniformly, apply the purification protocol, rule out all the failure cases, and the average fidelity of the successful cases can be computed by $\bar{F}(n;\cS,\cN,\cE) = \frac{1}{|\cS|}\sum_{\psi\in\cS} F(\psi, \frac{\Tilde{\sigma}_\psi}{\bar{P}})$, where $\Tilde{\sigma}_\psi$ is the unnormalized purified state, and $\bar{P}$ is the average successful probability. Clearly, the performance of a purification protocol depends on the average successful probability $\bar P$. Thus, with a given average success probability $\bar P$, we define the \textit{optimal distributed purification protocol} $\cE^*$ as a protocol that achieves the \textit{optimal average purification fidelity}, which is
\begin{align}
    \bar{F}_{\rm LOCC}^*(n,\bar{P};\cS,\cN) &= \bar{F}(n;\cS,\cN, \cE^*) \nonumber\\
    = \max_{\cE}   
    \big\{\bar{F}(n;\cS,&\cN,\cE)\,\big| \,\cE\in{\rm LOCC} \cap {\rm CPTN}, \nonumber\\
    &\frac
    {1}{|S|} \sum_{\psi\in\cS} \tr[\cE(\cN(\psi)^{\ox n})]=\bar{P}\big\}.
\end{align}
The definition of average purification fidelity is also valid for the continuous state set by replacing the summation with the integral.
Details can be found in Supplemental Material~\ref{app:formal_definition}.

From the definition, the performance of distributed purification protocols depends on the state set. \update{We firstly consider the universal set, which contains all 2-qubit pure states $\cS_P$}. Surprisingly, our first main result reveals that there does not exist a $2\rightarrow 1$ nontrivial distributed purification protocol.

\begin{theorem}\label{theo:no-go_universal}{\rm \textbf{(No-go theorem for the pure state set)}}
    For the depolarizing noise channel $\cN_{AB}^\gamma$ with noise level $\gamma\in(0,1)$, there is no nontrivial $2\rightarrow 1$ LOCC purification protocol for the set $\cS_P$ that contains all 2-qubit pure states.
\end{theorem}

To prove Theorem~\ref{theo:no-go_universal}, we first relax LOCC to positive partial transpose (PPT) operations.
Then we assume there exists a non-trivial purification protocol $\cE$, \update{which must satisfy the linear constraints shown in Eq.~(C3) in the Supplemental Material~\ref{app:proof_universal_no-go}.}
By utilizing the symmetry property of the state set $\cS_P$ and Schur's lemma, we reform and simplify the linear constraints, and the problem becomes an SDP problem. We then prove that the optimal solution to the SDP corresponds to the trivial purification protocol, implying there is no nontrivial distributed purification protocol. The proof is shown in the Supplemental Material~\ref{app:proof_universal_no-go}.

Here, we would like to emphasize that the $2\rightarrow 1$ no-go theorem for the pure state set $\cS_P$ does not imply the $n\rightarrow 1$ no-go theorem. For a straightforward strategy, if two parties share infinitely many copies of the noisy state $\cN^\gamma(\psi)$, where $\psi$ is an arbitrary pure state, they could implement tomography to obtain the density matrix of the noisy state. Since the noisy channel $\cN^\gamma$ is global depolarizing noise, $\cN^\gamma(\psi) = (1-\gamma)\psi + \gamma\frac{I}{4}$, the two parties can deduce the pure state $\psi$ from the density matrix of the noisy state. Now the problem reduces to purifying a known quantum state with a distributed purification protocol. Such a problem will be demonstrated in the next section.

The proof of Theorem~\ref{theo:no-go_universal} is highly dependent on the symmetry of the pure state set $\cS_P$. It is natural to ask whether a nontrivial distributed protocol exists if we consider a less symmetric state set. Here we take the Bell state set $\cS_{B}$ into consideration, i.e., $\cS_B = \{\Phi^+, \Phi^-, \Psi^+,\Psi^-\}$. In DQC, operations are applied to the quantum systems locally, and thus it is natural to assume the noise is introduced locally, \update{e.g., a bi-local depolarizing noise $\cN_{AB}^{\gamma_1, \gamma_2} = \cN_{A}^{\gamma_1} \ox \cN_{B}^{\gamma_2}$, which is equivalent to a global noise $\cN_{AB}^{\gamma'}$ with parameter $\gamma' = 1-(1-\gamma_1)(1-\gamma_2)$ shown in Supplemental Material~\ref{app:equ_of_global_local}. Notably, for $\cS_B$ and $\cN_{AB}^{\gamma_1,\gamma_2}$, we show that a non-trivial purification protocol still does not exist.}

\begin{theorem} {\rm \textbf{(No-go theorem for four Bell states)}}\label{theo:no-go_4_bell}
    For the bi-local depolarizing noise channel $\cN_{AB}^{\gamma_1, \gamma_2}$ with noise levels $\gamma_1$ and $\gamma_2$, there is no nontrivial $2\rightarrow 1$ LOCC purification protocol for the set $\cS_B$ that contains 4 Bell states.
\end{theorem}

The detailed proof is shown in Supplemental Material~\ref{app:proof_of_no-go_4_bell}. \update{By Theorem~\ref{theo:no-go_4_bell}, there does not exist a $2\rightarrow 1$ distributed blind purification protocol $\cE$ that can increase the fidelity for four Bell states simultaneously. However, if we only care about one Bell state rather than four, then the problem reduces to entanglement distillation, and it is possible to increase the fidelity with a purification protocol.}

If we apply the local unitaries $W_{AB}
=U_A\ox V_B$ to the Bell states set, we could obtain a new set $\cS_B^{W_{AB}}$. Since the local unitary transformed set is equivalent to the Bell states set, we could arrive at the following corollary. 
\begin{corollary}\label{coro:no-go_for_4_orthogonal}
    For the bi-local depolarizing noise channel $\cN_{AB}^{\gamma_1,\gamma_2}$ with noise levels $\gamma_1$ and $\gamma_2$, there is no nontrivial $2\rightarrow 1$ LOCC purification protocol for the set $\cS_B^{W_{AB}}$ that contains 4 arbitrary orthogonal maximally entangled states.
\end{corollary}

Combining Theorem~\ref{theo:no-go_4_bell} and Corollary~\ref{coro:no-go_for_4_orthogonal}, it is straightforward to arrive at a statement that there does not exist a nontrivial $2\rightarrow 1$ purification protocol for the union set $\cS_U =\cS_B \cup \cS_B^{W_{AB}^1}\cup \cdots\cS_B^{W_{AB}^k}\cdots \cup\cS_B^{W_{AB}^n}$, where $W_{AB}^k$ refers to local unitary operation.
If there exists a nontrivial $2\rightarrow 1$ purification protocol for the union set $\cS_U$, it means the purified fidelity strictly increases for some states, i.e., $F(\psi, \sigma_\psi)>F(\psi, \cN(\psi), \exists\,\psi\in\cS_U$. Then, by removing some states from the union set, we obtain $F(\psi, \sigma_\psi)>F(\psi, \cN(\psi), \exists\,\psi\in\cS_B^{W_{AB}^k}$, which contradicts the no-go theorem shown in Corollary~\ref{coro:no-go_for_4_orthogonal}. Therefore, there must be no nontrivial purification protocol for the union set $\cS_U$.
If all local unitary operations are considered, the union set becomes the maximally entangled state set
\begin{align}
    \cS&_{MES} =\cS_B \cup \cS_B^{W_{AB}^1}\cup \cdots\cS_B^{W_{AB}^k}\cdots \cup\cS_B^{W_{AB}^\infty}\nonumber\\
    &= \{\proj{\Phi}\;|\;\ket{\Phi} = U_A \ox V_B \ket{\Phi^+}, \forall\; U_A,V_B\in SU(2) \}\nonumber,
\end{align}
and we arrive at the no-go theorem for maximally entangled states set $\cS_{MES}$.

\begin{corollary}{\rm \textbf{(No-go theorem for maximally entangled states set)}}\label{theo:no-go_max_entangled_set}
    For the local depolarizing noise channel $\cN_{AB}^{\gamma_1, \gamma_2}$ with noise levels $\gamma_1$ and $\gamma_2$, there is no nontrivial $2\rightarrow 1$ LOCC purification protocol for the set $\cS_{MES}$ that contains all maximally entangled pure states.
\end{corollary}
The above no-go theorems highlight the essential role of entanglement in quantum state purification. Without entanglement, LOCC alone is not sufficiently powerful to purify these representative sets of quantum states.

%%%%%%%%%%%%%%%%%%%%%%%%%%%%%%%%%%%%%%%%%%%%%%%%%%%%%%%%%%%%%%%%%%%%%%%%%%%%%
\textit{Purifying a known state with LOCC protocol.}--- 
We have seen different no-go theorems for the distributed state purification. Then what can we do if we know more information about the noisy states? The simplest case is that the state set $\cS$ only has one pure state $|\cS|=1$ and the state is known. This case is trivial in the global purification setting, because the global purification protocol can directly prepare the target state $\psi$. However, such a problem is nontrivial in the distributed computing paradigm, as it is impossible to prepare the state $\psi$ via an LOCC protocol as long as $\psi$ is entangled.
In Theorem~\ref{theo:go_theorem_for_arbitrary_state}, we show that a nontrivial distributed purification protocol exists if the target state $\psi$ is known.

\begin{figure}[h]
\centering
\includegraphics[width=\linewidth]{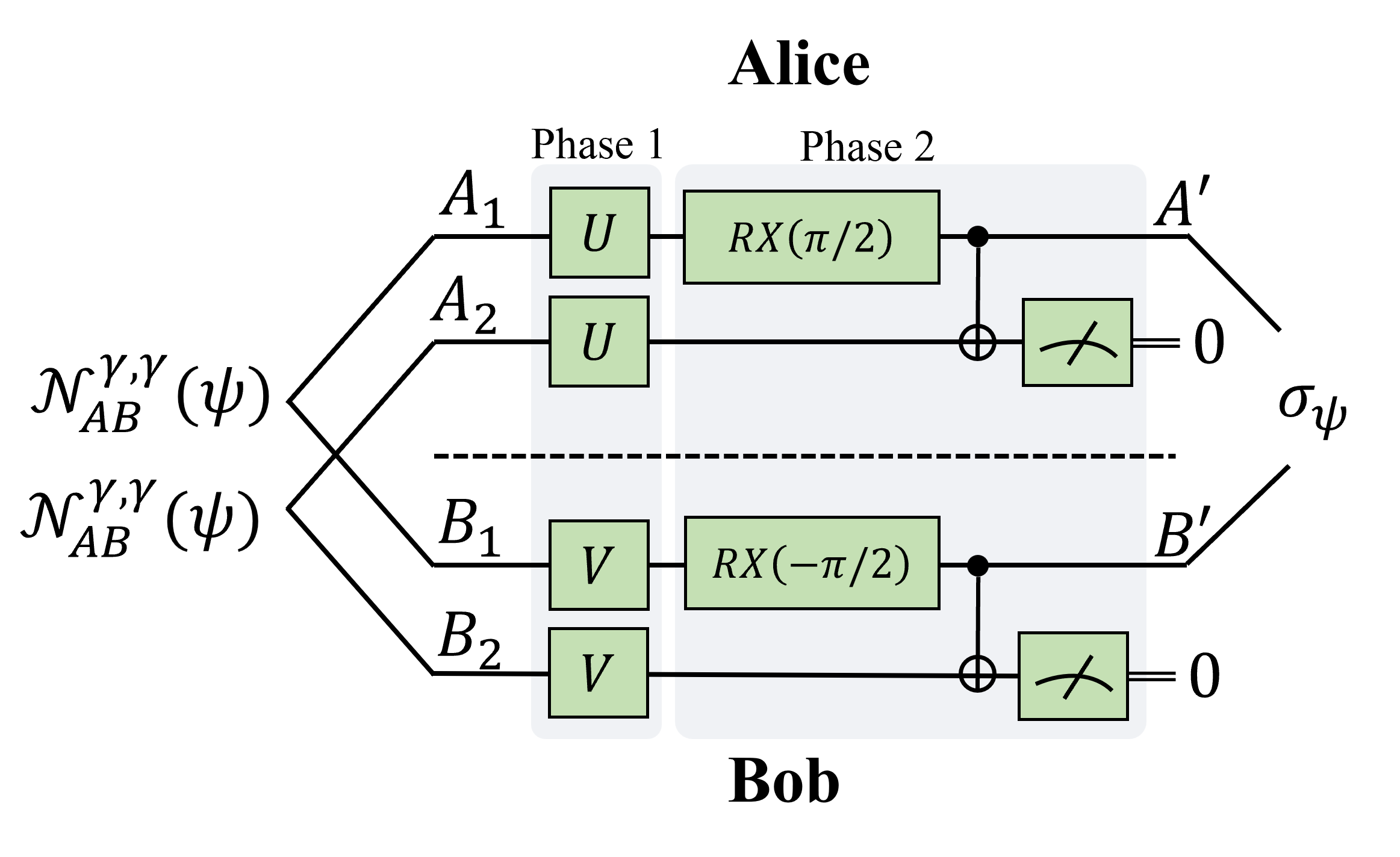}
    \caption{Diagram for the $2\rightarrow1$ analytical distributed purification protocol for noisy state $\cN_{AB}^{\gamma,\gamma}(\psi)$. This purification protocol has two phases. In phase 1, the local unitaries map the noisy state $\cN_{AB}^{\gamma,\gamma}(\psi)$ to $\cN_{AB}^{\gamma,\gamma}(\psi_\alpha)$. Phase 2 can purify the noisy state $\cN_{AB}^{\gamma,\gamma}(\psi_\alpha)$ for arbitrary $\alpha\in[0,1]$.}
    \label{fig:protocol_for_single_state}
\end{figure}

\begin{theorem} {\rm \textbf{(Distributed purification protocol for known states})}\label{theo:go_theorem_for_arbitrary_state}
    For a given \update{pure} state $\psi$, which is corrupted by the bi-local depolarizing noise channel $\cN_{AB}^{\gamma, \gamma}$ with noise level $\gamma\in[0,0.4]$, there always exists a nontrivial LOCC purification protocol, such that the purified state $\sigma_\psi$ has a strictly higher fidelity, i.e., $F(\sigma_\psi, \psi) > F(\cN_{AB}^{\gamma,\gamma}(\psi), \psi)$.
    \end{theorem}

To prove the theorem, we propose an analytical $2\rightarrow 1$ distributed purification protocol to purify the state corrupted by bi-local depolarizing noise, which is shown in FIG.~\ref{fig:protocol_for_single_state}. When two parties, Alice and Bob, possess two copies of shared noisy states $\cN_{AB}^{\gamma, \gamma}(\psi)$, they apply some local unitaries $U$ and $V$ on their own systems respectively (as shown in Phase 1), mapping the noisy state $\cN_{AB}^{\gamma,\gamma}(\psi)$ to $\cN_{AB}^{\gamma,\gamma}(\psi_\alpha)$, where $\ket{\psi_\alpha}=\alpha\ket{00} + \sqrt{1-\alpha^2}\ket{11}$ with $\alpha\in[0,1]$.
Then, for the noisy state $\cN_{AB}^{\gamma,\gamma}(\psi_\alpha)$ with arbitrary $\alpha\in[0,1]$, the circuit shown in Phase 2 can non-trivially purify it.
Specifically, Alice and Bob apply local rotation gates and CNOT gates on their own parts, and 
make measurements on systems $A_2,B_2$. Denote the measurement results as $m_A$ and $m_B$ respectively, and Alice and Bob exchange the results via classical communication. If $m_A=m_B=0$, then the purification is successful and outputs the purified state $\sigma_\psi$; otherwise, the purification fails and the process described above needs to be repeated. To verify the effectiveness of the proposed protocol, we derive the fidelity of the purified state and the target state $F(\sigma_\psi, \psi)$ and show that $F(\sigma_\psi, \psi) - F(\cN_{AB}^\gamma(\psi), \psi)\ge 0$ holds for arbitrary states with $\gamma\in[0,0.4]$. \update{Here we would like to highlight that the protocol shown in FIG.~\ref{fig:protocol_for_single_state} not only can purify a single state, but also works for a set of states with the same Schmidt basis.} More details are shown in Supplemental Material~\ref{app:proof_of_go_single_state}.

The conventional distributed purification protocols, such as DEJMPS~\cite{deutsch1996quantum} and BBPSSW\cite{bennett1996purification}, can only purify the noisy maximally entangled state. While the proposed protocol, as shown in FIG.~\ref{fig:protocol_for_single_state}, can not only purify maximally entangled states, but also works for non-maximally entangled states. Note that non-maximally entangled states are also important resources in distributed quantum information processing tasks such as entanglement catalyst~\cite{jonathan1999entanglement} and the verification of Hardy's paradox~\cite{hardy1992quantum, hardy1993nonlocality}.

%%%%%%%%%%%%%%%%%%%%%%%%%%%%%%%%%%%%%%%%%%%%%%%%%%%%%%%%%%%%%%%%%%%%%%%%

\textit{Purification protocol via optimization.}---
In a more general case, where the state set contains more than one pure state, i.e., $|\cS| \ge 2$, the proposed protocol shown in FIG.~\ref{fig:protocol_for_single_state} no longer works, and we have to propose an LOCC purification protocol for every new state set. This motivates us to introduce an automatic algorithm to design the distributed purification protocol. Specifically, Alice and Bob share $n$ copies of noisy states $\cN_{AB}(\psi)$, and the state $\psi$ is sampled from a set of pure states $\cS$. Alice and Bob apply local unitaries $U$ and $V$ on their own systems, respectively. After measuring all the systems except the first pair, the two parties check the measurement results via classical communication. If the measurement results are all 0, the purification process is successful, and the first pair state is the purified state $\sigma_\psi$; otherwise, the purification fails. 

In this algorithm, the optimal local unitaries $U$ and $V$, that satisfy the constraints, are unknown. To solve this problem, we utilize the parametrized quantum circuits (PQCs)~\cite{benedetti2019parameterized} to represent the local unitaries $U(\boldsymbol{\theta})$ and $V(\boldsymbol{\zeta})$, and apply gradient-based optimization method~\cite{cerezo2021variational} to search the optimal parameters $\boldsymbol{\theta}, \boldsymbol{\zeta}$ in the PQCs $U(\boldsymbol{\theta})$ and $V(\boldsymbol{\zeta})$. The optimization method has been successfully applied to various tasks, such as quantum chemistry~\cite{kandala2017hardware, peruzzo2014variational, mcclean2017hybrid, nakanishi2019subspace}, quantum information analysis~\cite{chen2023near, chen2021variational, zhao2021practical,zhao2025variational,liu2025dynamic}, and quantum data compression~\cite{romero2017quantum, cao2021noise}. \update{Note that the optimized protocols here no longer require the noise level constraint $\gamma \le 0.4$, a condition needed in Theorem~\ref{theo:go_theorem_for_arbitrary_state} for nontrivial purification.}

To optimize the parameters in the PQCs, we have to design a cost function that describes our requirements. In the case of a distributed purification protocol, there are two requirements. For the first one, the average of purified fidelity should be as large as possible, i.e., $\max_{U,V} \sum_{\psi\in\cS} F(\psi, \sigma_{\psi})$. For the second one, the purification protocol should increase the fidelity for every purified state in the set, i.e., $F(\psi,\sigma_\psi) \ge F(\psi, \cN(\psi)),\; \forall\, \psi\in\cS.$
Based on the requirements, we have the following cost function

\begin{equation}\label{eq:cost_function}
    C = -\sum_{\psi\in\cS} F(\psi, \sigma_\psi) + \sum_{\psi\in\cS}S(F(\psi, \cN(\psi)) - F(\psi,\sigma_\psi)),
\end{equation}
where $S(x) = \frac{1}{1 + e^{-ax}}$ is sigmoid function with positive real coefficient $a$. The first summation term $\sum_{\psi\in\cS} F(\psi,\sigma_\psi)$ corresponds to the first requirement. The second summation is the punishment term. When $F(\psi, \cN(\psi)) - F(\psi,\sigma_\psi)\ge0$, it implies the purified state has lower fidelity, then $S(F(\psi, \cN(\psi)) - F(\psi,\sigma_\psi)) \ge 1/2$. The more the fidelity drops, the greater the punishment will be. The penalty strength can be modified by tuning the parameter $a$. With the chosen ansatz and the cost function Eq.~\eqref{eq:cost_function}, we can execute the optimization algorithm to train the parameters. More algorithm details please refer to the Supplemental Material~\ref{app:VQA}. 

\begin{figure}[h]
\centering
\includegraphics[width=\linewidth]{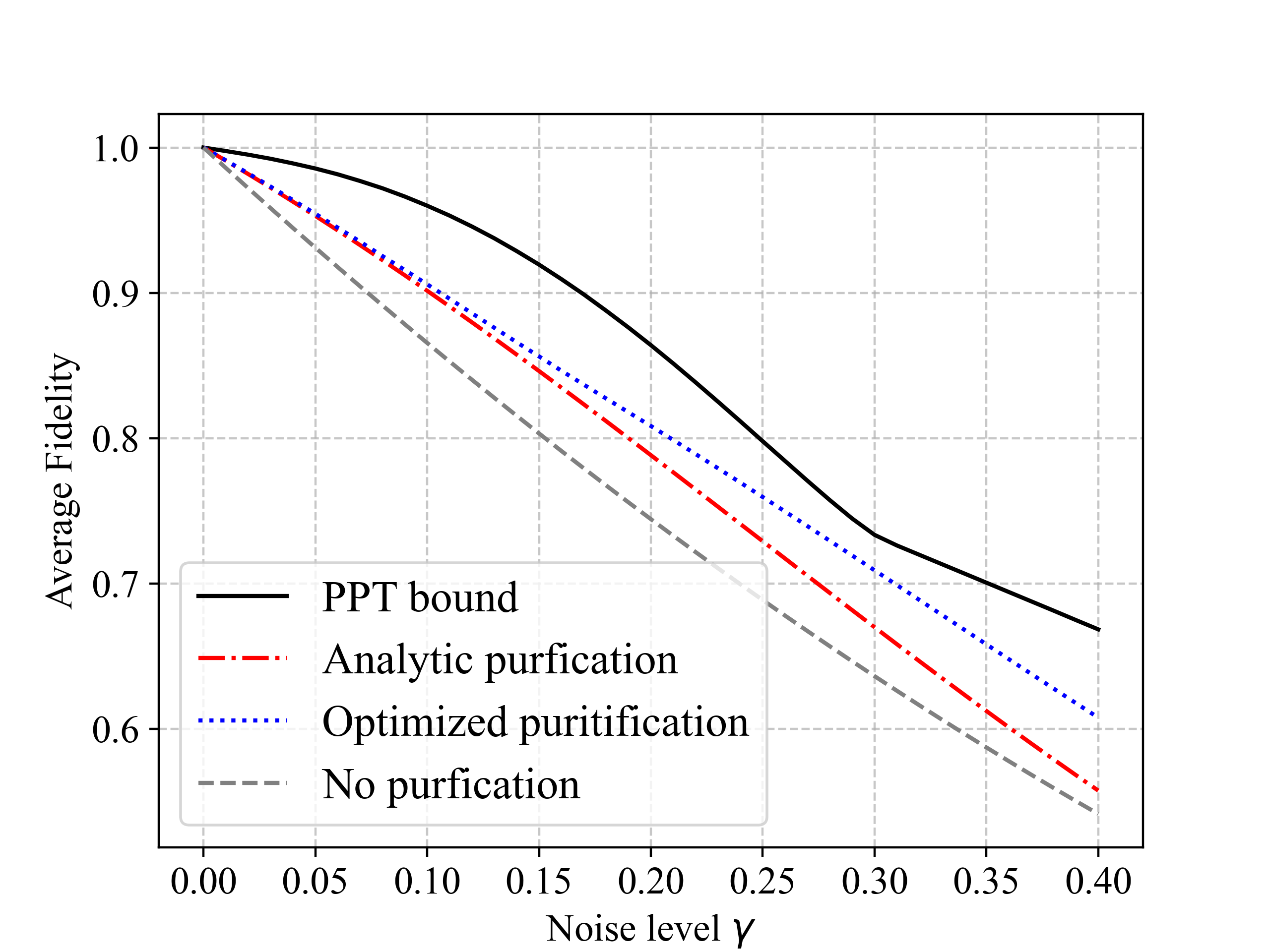}
    \caption{Comparison of the performance of different $2\rightarrow 1$ distributed purification protocols. The black solid line is the PPT upper bound. The dotted blue line is achieved by the optimized purification protocol. The dash-dot red curve corresponds to the Phase 2 circuit shown in FIG.~\ref{fig:protocol_for_single_state}. The dashed grey line refers to the fidelity without purification.}
    \label{fig:purified_fidelity_local_Depo}
\end{figure}
Here we conduct numerical experiments~\cite{codes}, via the QUAIRKIT platform~\cite{quairkit}, to show the effectiveness of the optimized purification protocol. The numerical results are shown in FIG.~\ref{fig:purified_fidelity_local_Depo}. Specifically, we consider the pure state set containing 4 different states $\cS_d=\{\psi_{\frac{1}{\sqrt{2}}}, \psi_{\frac{1}{\sqrt{3}}}, \psi_{\frac{1}{\sqrt{4}}}, \psi_{\frac{1}{\sqrt{5}}}\}$, 
where $\ket{\psi_\alpha}=\alpha\ket{00}+\sqrt{1-\alpha^2}\ket{11}$. The noise model is bi-local depolarizing channel $\cN_{AB}^{\gamma, \gamma}$, with noise level varying $\gamma\in[0,0.4]$. From the proof of Theorem~\ref{theo:go_theorem_for_arbitrary_state}, the Phase 2 circuit shown in FIG.~\ref{fig:protocol_for_single_state} is a nontrivial purification protocol for the state in the form of $\ket{\psi_\alpha}$. Thus, the Phase 2 circuit can also non-trivially purify the set $\cS_d$.
We compare the performance of the optimized protocol and the analytical protocol and find that the optimized protocol achieves higher fidelity. To benchmark the performance of the optimized purification protocol, we compare the fidelity achieved by the optimized purification protocol and the PPT bound, which is calculated by the SDP shown in Supplemental Material~\ref{app:PPT_bound}. Here, we fix the average success probability as 0.1. Note that the fidelity achieved by optimized purification is close to the PPT bound, implying the performance of the optimized purification is remarkable. \update{For larger quantum noise level, the average fidelity becomes more complicated, and we include it in the Supplemental Material~\ref{app:lager_noise_level}.} \update{We also conduct more numerical experiments for better understanding (cf.~ Supplemental Material~\ref{app:more_experiments})}.

%%%%%%%%%%%%%%%%%%%%%%%%%%%%%%%%%%%%%%%%%%%%%%%%%%%%%%%%%
\textit{Discussion.}--- 
In this work, we systematically characterized the fundamental limitations and operational capabilities of distributed state purification. In particular, we proved that, in the presence of depolarizing noise,  no  $2 \rightarrow 1$ LOCC purification protocol can work blindly for all states in three important sets: the set of all pure two-qubit states, the set of all two-qubit maximally entangled states, and the Bell basis.   
If instead the purification protocol is targeted to a single, known pure state,  an explicit purification protocol can be constructed.  \update{Finally, for a finite set of pure states}, we introduced an optimization-based algorithm to design LOCC purification protocols. These results identify both no-go regimes and constructive possibilities, establishing distributed state purification as a viable strategy for enhancing noise resilience in quantum networks.

For future work, it would be interesting to identify the necessary and sufficient conditions for state sets that admit a nontrivial distributed purification protocol, \update{and it would also be important to figure out how the symmetry of noise relates to the no-go theorems}. Another promising direction is to determine the optimal $n \rightarrow 1$ purification protocol under LOCC constraints, as well as to study the power and limitations under the multipartite entanglement regime. 

\textbf{Acknowledgements}--- B.Z. and X.W. would like to thank Jiayi Zhao and Bartosz Regula for insightful discussions.
G.C. and X.Z. acknowledge support from the Chinese Ministry of Science and Technology (MOST) through grant 2023ZD0300600, by the Hong Kong Research Grant Council (RGC) through grant R7035-21F, and by the State Key Laboratory of Quantum Information Technologies and Materials, Chinese University of Hong Kong. X.W. acknowledges the support from the National Natural Science Foundation of China (Grant. No.~12447107), the Guangdong Provincial Quantum Science Strategic Initiative (Grant No.~GDZX2403008, GDZX2503001), 
the Guangdong Natural Science Foundation (Grant No.~2025A1515012834), and the Quantum Science Center of Guangdong-Hong Kong-Macao Greater Bay Area.  

%%%%%%%%%%%%%%%%%%%%%%%%%%%%%%%%%%%%%%%%%%%%%%%%%%%%%%%%%%%%%%%%%%%%%%%%%%%%%%
% Bibliography
%%%%%%%%%%%%%%%%%%%%%%%%%%%%%%%%%%%%%%%%%%%%%%%%%%%%%%%%%%%%%%%%%%%%%%%%%%%%%%
% \bibliographystyle{apsrev4-2}
% \bibliographystyle{plain}
% \bibliographystyle{apsrev4-2-titles}
% \bibliography{apsrev4-2-titles}
% \bibliography{revtex4-2}
\bibliography{ref.bib}

%%%%%%%%%%%%%%%%%%%%%%%%%%%%%%%%
% Appendix
%%%%%%%%%%%%%%%%%%%%%%%%%%%%%%

\clearpage
\vspace{2cm}
\onecolumngrid
\vspace{2cm}
\begin{center}
{\textbf{\Large Supplemental Materials for \\ Power and limitations of distributed quantum state purification}}
\end{center}

\appendix
\tableofcontents
\renewcommand{\appendixname}{Supplemental Material}
\renewcommand{\thesection}{\Alph{section}}
\renewcommand{\thesubsection}{\arabic{subsection}}
\renewcommand{\theequation}{\Alph{section}.\arabic{equation}}

\section{Global purification protocol}\label{app:global_purification_protocols}
The state purification has been well studied, and many practical protocols have been proposed. For example, the \textit{symmetric projection} method~\cite{barenco1997stabilization,childs2025streaming, yao2025protocols} probabilistically maps the $n$-copy of noisy state into the symmetric subspace, and the output states have higher purity. The circuit for the implementation of symmetric projection is shown in FIG.~\ref{fig:symmetric_projection}. The unitary $U_k, k\in[1,n-1]$ maps the initial state $\ket{0}^{\ox k}$ to the W-state, i.e., $U_k(\ket{0}^{\ox k}) = \frac{1}{\sqrt{k+1}}(\ket{00\cdots0} + \ket{00\cdots1} + \cdots + \ket{10\cdots0})$. Then apply control-swap gates and $U_k^\dagger$, then measure the ancilla qubits in the end. The measurement results are denoted as $\mathbf{\omega}$. If all the measurement are 0, i.e., $\mathbf{\omega} = \mathbf{0}$, then the purification is successful and output the state $\sigma(\mathbf{\omega})$. Otherwise, the purification process fails, and we have to repeat the purification protocol $\cE_{sym}$ until it succeeds. The output purified state is 
\begin{equation}\label{eq:purified_state}
    \sigma(\mathbf{0})=\cE(\cN(\psi) ^{\ox n}) = \frac{\sum_{j=1}^n \cN(\psi)^j}{\sum_{j=1}^n\tr[\cN(\psi)^j]}.
\end{equation}

\begin{figure}[h]
\centering
\includegraphics[width=0.8\linewidth]{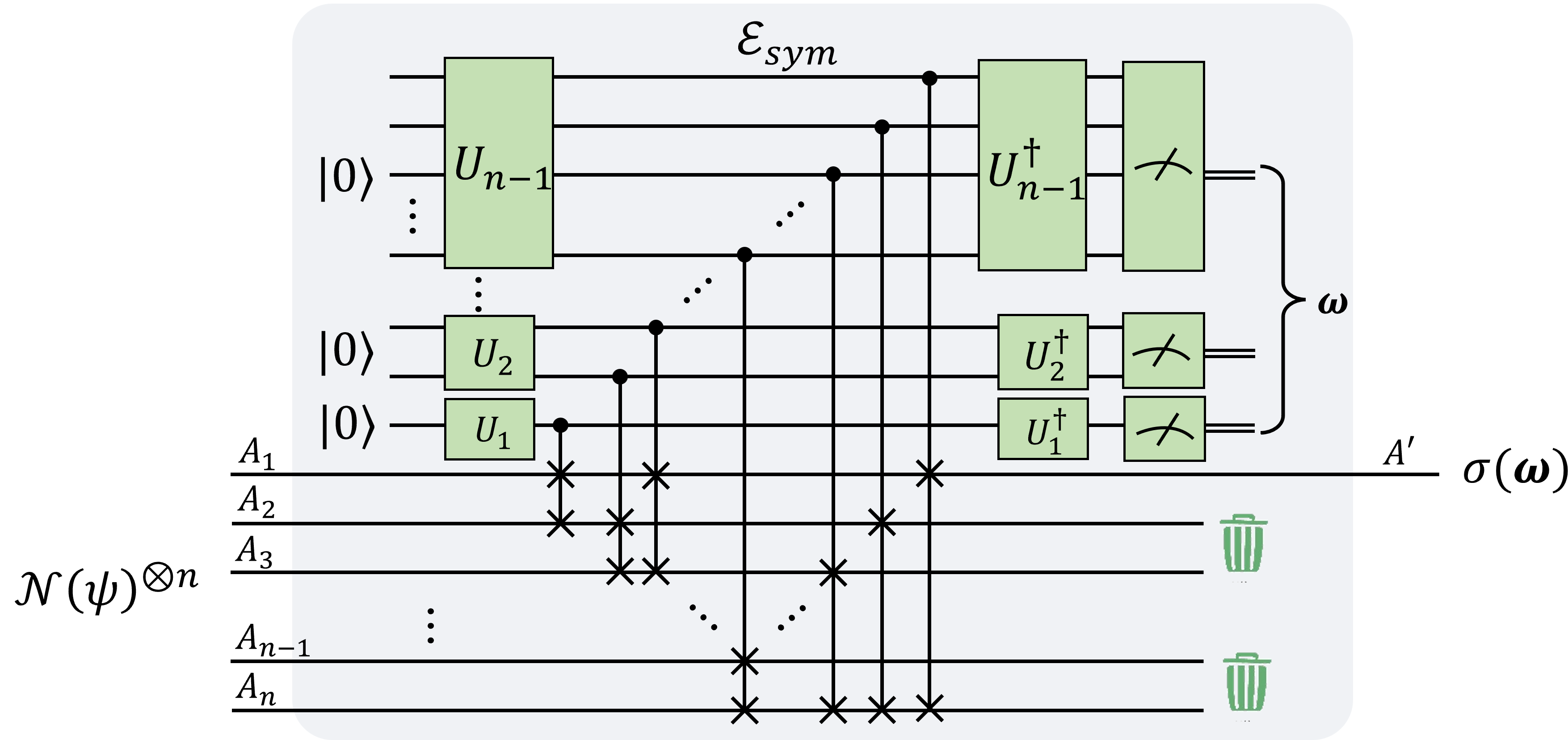}
    \caption{Quantum circuit project quantum state to its symmetry subspace} 
    \label{fig:symmetric_projection}
\end{figure}

\begin{proposition}\label{prop:symmetric_projection_is_universal_protocol} 
    For the noise channels $\cN$ whose dominant term of noise state $\cN(\psi)$ remains $\psi$, the symmetric projection is an $n\rightarrow 1$ purification protocol for the set that contains all pure states.
\end{proposition}
\begin{proof}
    Since the dominant term of the noisy state $\cN(\psi)$ remains $\psi$, we have 
    \begin{equation}
        \cN(\psi) = \lambda_1 \psi + \sum_{i=2}^{d} \lambda_i \psi_i,
    \end{equation}
    with $\lambda_1 > \lambda_k, \,\forall\, k \ge 2$, and $\sum_{i=1}^{d} \psi_i=1$. The $\lambda_i$ are the eigenvalues, and the $\psi_i$ are the noisy terms, which are orthogonal to the ideal pure state $\psi$. The output purified state is given in Eq.~\eqref{eq:purified_state}.
    Note that $\cN(\psi)^j,\,\forall\,j$ share the same eigenvectors. For convenience, we order the eigenvalues in a descending sequence, $\lambda_1> \lambda_2\ge\cdots\ge\lambda_{d}$, then the eigenvalues of the output purified state $\sigma(\mathbf{0})$ are 
    \begin{equation}
        \lambda_i' = \frac{\lambda_i+\lambda_i^2 + \cdots +\lambda_i^n}{1 + \sum_{m=1}^d\lambda_m^2 + \cdots+\sum_{m=1}^d\lambda_m^n}.
    \end{equation}
    In particular,
    \begin{align}
        \lambda_1' - \lambda_1 &= \frac{\lambda_1+\lambda_1^2 + \cdots +\lambda_1^n}{1 + \sum_{m=1}^d\lambda_m^2 + \cdots+\sum_{m=1}^d\lambda_m^n} - \lambda_1\\
        &= \frac{\lambda_1^2 + \lambda_1^3 + \cdots + \lambda_1^n - \lambda_1\sum_{m=1}^d\lambda_m^2 - \cdots-\lambda_1\sum_{m=1}^d\lambda_m^n}{1 + \sum_{m=1}^d\lambda_m^2 + \cdots+\sum_{m=1}^d\lambda_m^n}\\
        &= \frac{\lambda_1^2\sum_{m=1}^d\lambda_m - \lambda_1\sum_{m=1}^d \lambda_m^2 +\cdots \lambda_1^n\sum_{m=1}^d \lambda_m - \lambda_1\sum_{m=1}^d \lambda_m^n}{1 + \sum_{m=1}^d\lambda_m^2 + \cdots+\sum_{m=1}^d\lambda_m^n}\\
        &= \lambda_1 \frac{\sum_{m=1}^d \lambda_m(\lambda_1-\lambda_m)+\cdots + \lambda_m(\lambda_1^{n-1} - \lambda_m^{n-1})}{1 + \sum_{m=1}^d\lambda_m^2 + \cdots+\sum_{m=1}^d\lambda_m^n}\\
        & \ge 0.
    \end{align}
    The equality holds when the noisy state is pure, i.e., no noise at all. Thus, for any input state $\psi$, the purified state has higher fidelity, i.e., $\lambda_1'=F(\psi, \cE(\cN(\psi)^{\ox 2})) \ge F(\psi,\cN(\psi))=\lambda_1$, which satisfies the definition of purification protocol. Here, we complete the proof.
\end{proof}

When we take $n=2$, then the symmetric projection reduces to the \textit{swap test gadget}~\cite{childs2025streaming} as shown in FIG.~\ref{fig:swap_test_gadget}. Since the swap test gadget is a special case of the symmetric projection, we directly have Swap test gadget is a $2\rightarrow 1$ purification protocol for the set that contains all pure states.
    
\begin{figure}[h]
\centering
\includegraphics[width=0.6\linewidth]{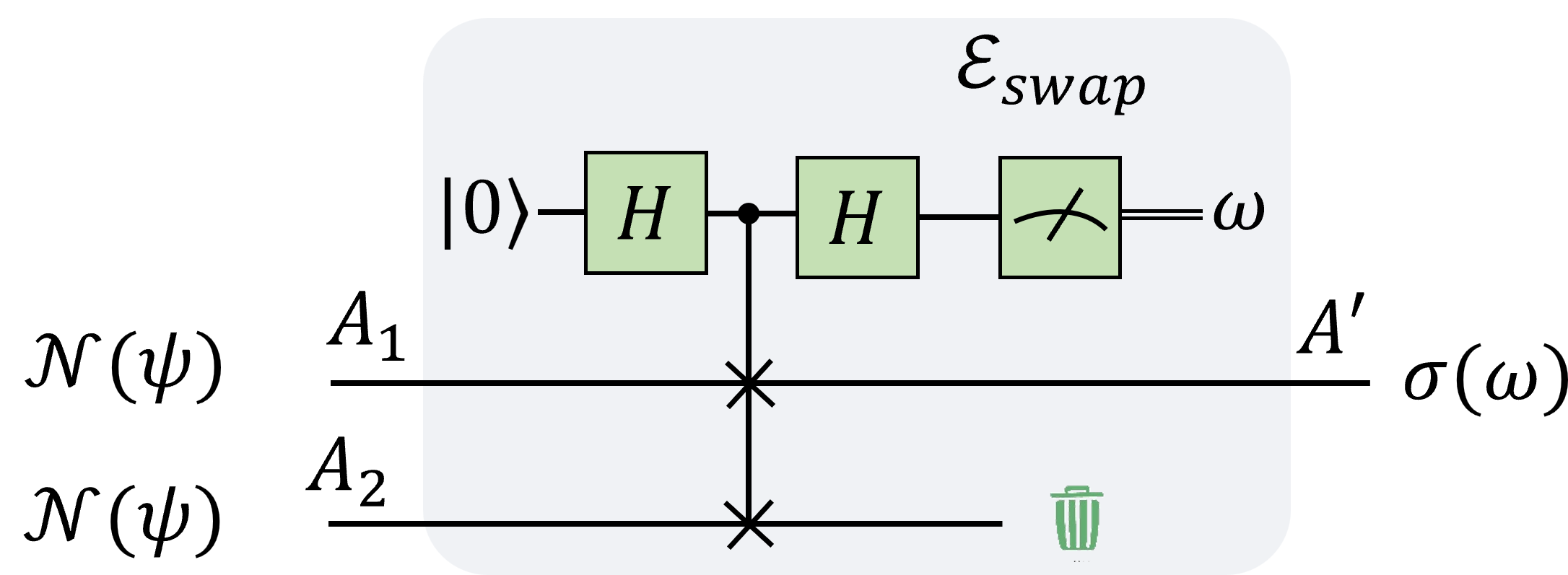}
    \caption{Quantum circuit for $2\rightarrow 1$ swap test gadget on arbitrary input noisy states $\cN(\psi)$. The shaded area is the purification protocol $\cE_{swap}$. The three-qubit gate between Hadamard gates is a controlled swap.} 
    \label{fig:swap_test_gadget}
\end{figure}

%%%%%%%%%%%%%%%%%%%%%%%%%%%%%%%%%%%%%%%%%%%%%%%%%%%%%%%%%%%%%%%%%%%%
\section{Formal definitions of average purification fidelity and optimal distributed purification protocol}\label{app:formal_definition}

To fairly compare the purification protocols, we define the \textit{average purification fidelity} to quantify the performance of different purification protocols.
The average purification fidelity only takes the successful cases into account. Specifically, sample quantum states from the state set uniformly, apply the purification protocol, rule out all the failure cases, and the average fidelity of the successful cases can be computed by $\bar{F}(n;\cS,\cN,\cE)=\sum_{\psi\in\cS} \frac{p_\psi }{\sum_{\psi\in \cS} p_\psi}F(\sigma_\psi, \psi)$. Dividing $|\cS|$ from the numerator and denominator, respectively, we arrive at $\bar{F}(n;\cS,\cN,\cE)=\frac{1}{|\cS|}\sum_{\psi\in\cS} \frac{p_\psi }{\bar{P}}F(\sigma_\psi, \psi)$. Since the output state is $\psi=\frac{\tilde \sigma_\psi}{p_\psi}$, we have $\bar{F}(n;\cS,\cN,\cE) = \frac{1}{|\cS|}\sum_{\psi\in\cS} F\Big(\psi, \frac{\Tilde{\sigma}_\psi}{\bar{P}}\Big)$. Formally, the average purification fidelity is defined in the following.
 
\begin{definition}\label{def:average_purificartion_fid}{\rm \textbf{(Average purification fidelity)}} 
    For a given pure state set $\cS$, noise channel $\cN$ and the $n\rightarrow1$ distributed purification protocol $\cE$, the average purification fidelity is defined as 
    \begin{equation}
        \bar{F}(n;\cS,\cN,\cE) = \frac{1}{|\cS|}\sum_{\psi\in\cS} F\Big(\psi, \frac{\Tilde{\sigma}_\psi}{\bar{P}}\Big),
    \end{equation}
    where $\Tilde{\sigma}_\psi$ is the unnormalized purified state.
\end{definition}

Clearly, the performance of purification protocol depends on the average successful probability $\bar P$, which was also discussed in Ref.~\cite{yao2025protocols}. Thus, with a given average success probability $\bar P$, we define the \textit{optimal purification protocol} as a protocol that achieves \textit{optimal average purification fidelity} $\bar{F}^*$. Since this work focuses on the distributed purification protocol, we have to constrain the operation to LOCC. The \textit{optimal distributed purification protocol} can be defined formally in the following.

\begin{definition}\label{def:opt_average_purificartion_fid}{\rm \textbf{(Optimal distributed purification protocol)}} 
    For a given pure state set $\cS$, noise channel $\cN$ and the average success probability $\bar{P}$, the optimal $n\rightarrow 1$ distributed purification protocol $\cE^*$ is 
    a protocol that achieves the optimal average purification fidelity:
    \begin{align}
        \bar{F}_{\rm LOCC}^*(n,\bar{P};\cS,\cN) &= \bar{F}(n;\cS,\cN, \cE^*) \nonumber\\
        &= \max_{\cE}   
        \big\{\bar{F}(n;\cS,\cN,\cE)\,\big| \,\cE\in{\rm LOCC} \cap {\rm CPTN}, \frac
        {1}{|S|} \sum_{\psi\in\cS} \tr[\cE(\cN(\psi)^{\ox n})]=\bar{P}\big\}.
    \end{align}
\end{definition}

%%%%%%%%%%%%%%%%%%%%%%%%%%%%%%%%%%%%%%%%%%%%%%%%%%%%%%%%%%%%%%%%%%
\section{Proof of Theorem~\ref{theo:no-go_universal}}\label{app:proof_universal_no-go}

\noindent \textbf{Theorem~\ref{theo:no-go_universal} (No-go theorem for the pure state set)} \textit{For the two-qubit depolarizing noise channel $\cN_{AB}^\gamma$ with noise level $\gamma\in(0,1)$, there is no non-trivial $2\rightarrow 1$ LOCC purification protocol for the set that contains all pure states $\cS_P$.}
\vspace{0.5 cm}

\begin{proof}
    At the beginning, we denote $\mu_{P}$ as the Haar measure of set $\cS_P$,
    \begin{align}
        M_\psi&=\cN^{\ox n}(\psi^{\ox n})^{T}_{A^nB^n}\ox(\psi-\tr[\psi \cN(\psi)] I)_{A'B'}\text{ for any }\psi\in\cS_P,
    \end{align}
    $M=\mathbb{E}_{\psi\sim\mu_P} M_\psi$ as the expectation of $M_\psi$ with respect to the Haar measure,
    and have that
    \begin{equation}
        \tr[J_{\cE_{A^nB^nA'B'}}M]=\mathbb E_{\psi\sim\mu_P}{\tr[J_{\cE_{A^nB^nA'B'}} ^{T_{A^nB^n}} \cN^{\ox n}\ox \cI(\psi^{\ox n+1})]}-\tr[\psi \cN(\psi)]{\tr[J_{\cE_{A^nB^nA'B'}} ^{T_{A^nB^n}} \cN^{\ox n}(\psi^{\ox n})\ox I_{A'B'}]}.
    \end{equation}
    {\bf Suppose} that there exists a nontrivial $2\rightarrow 1$ LOCC purification protocol, which is equivalent to having a feasible solution $J_{\cE_{A^nB^nA'B'}}$ such that the following constraints are satisfied:
    \begin{align}\label{app:proof_universal_no_go_constraint}
        \begin{cases}
        J_{\cE_{A^nB^nA'B'}} \ge 0,\\ \tr_{A'B'} [J_{\cE_{A^nB^nA'B'}}] \le I_{A^nB^n},\\
        J_{\cE_{A^nB^nA'B'}}^{T_{A^nA'}} \ge 0,\\
        \forall\,\psi\in\cS_P,\ \tr[J_{\cE_{A^nB^nA'B'}}M_\psi]\ge0,\\
        \tr[J_{\cE_{A^nB^nA'B'}}M]>0.
        \end{cases}
    \end{align}
    Note that for any $W=U\ox V\in \operatorname{SU}(2)^{\ox 2}$, we have
    \begin{align}
        M_{W\psi W^\dag}=&(\cN^{\ox n}(W\psi W^\dag)^{\ox n})^{T}_{A^nB^n}\ox (W\psi W^\dag-\tr[W\psi W^\dag \cN(W\psi W^\dag)] I)_{A'B'}\\
        =&(W^{*\ox n}_{A^nB^n}\ox W_{A'B'})((\cN^{\ox n}(\psi)^{\ox n})^{T}_{A^nB^n}\ox(\psi-\tr[\psi \cN(\psi)] I)_{A'B'})(W^{T\ox n}_{A^nB^n}\ox W^\dag_{A'B'})\\
        =&(W^{*\ox n}_{A^nB^n}\ox W_{A'B'})M_\psi(W^{T\ox n}_{A^nB^n}\ox W^\dag_{A'B'}),\\
        \tr[J_{\cE_{A^nB^nA'B'}}M]
        =&\mathbb E_{\psi\sim\mu_P}\tr[J_{\cE_{A^nB^nA'B'}}M_{W^\dag\psi W}]\\
        =&\mathbb E_{\psi\sim\mu_P}\tr[(W^{*\ox n}_{A^nB^n}\ox W_{A'B'})J_{\cE_{A^nB^nA'B'}}(W^{T\ox n}_{A^nB^n}\ox W^\dag_{A'B'})M_\psi]\\
        =&\tr[(W^{*\ox n}_{A^nB^n}\ox W_{A'B'})J_{\cE_{A^nB^nA'B'}}(W^{T\ox n}_{A^nB^n}\ox W^\dag_{A'B'})M].
    \end{align}
    Moreover, since
    \begin{equation}
        J_{\cE'_{A^nB^nA'B'}}\coloneqq(W^{*\ox n}_{A^nB^n}\ox W_{A'B'})J_{\cE_{A^nB^nA'B'}}(W^{T\ox n}_{A^nB^n}\ox W^\dag_{A'B'})
    \end{equation}
    satisfies
    \begin{align}
        \begin{cases}
        \operatorname{Spec}\left(J_{\cE'_{A^nB^nA'B'}}\right)=\operatorname{Spec}\left(J_{\cE_{A^nB^nA'B'}}\right),\\
        \tr_{A'B'} [J_{\cE'_{A^nB^nA'B'}}]=
        W^{*\ox n}_{A^nB^n}\tr_{A'B'} [J_{\cE_{A^nB^nA'B'}}]W^{T\ox n}_{A^nB^n} \le W^{*\ox n}_{A^nB^n}I_{A^nB^n}W^{T\ox n}_{A^nB^n}=I_{A^nB^n},\\
        \operatorname{Spec}\left(J_{\cE'_{A^nB^nA'B'}}^{T_{A^nA'}}\right)
        =\operatorname{Spec}\left((U^{*\ox n}_{A^n}\ox U_{A'}\ox V^{*\ox n}_{B^n}\ox V_{B'})J_{\cE_{A^nB^nA'B'}}(U^{T\ox n}_{A^n}\ox U^\dag_{A'}\ox V^{T\ox n}_{B^n}\ox V^\dag_{B'})\right)^{T_{A^nA'}},\\
        \qquad=\operatorname{Spec}\left((U^{\ox n}_{A^n}\ox U^*_{A'}\ox V^{*\ox n}_{B^n}\ox V_{B'})J_{\cE_{A^nB^nA'B'}}^{T_{A^nA'}}(U^{\dag\ox n}_{A^n}\ox U^T_{A'}\ox V^{T\ox n}_{B^n}\ox V^\dag_{B'})\right)
        =\operatorname{Spec}\left(J_{\cE_{A^nB^nA'B'}}^{T_{A^nA'}}\right),\\
        \forall\,\psi\in\cS_P,\ \tr[J_{\cE'_{A^nB^nA'B'}}M_\psi]=\tr[J_{\cE'_{A^nB^nA'B'}}M_{W^\dag\psi W}]\ge0,\\
        \tr[J_{\cE'_{A^nB^nA'B'}}M]=\tr[J_{\cE_{A^nB^nA'B'}}M]>0,
        \end{cases}
    \end{align}
    we find that $J_{\cE_{A^nB^nA'B'}}$ is a feasible solution to the constraints in Eq.\eqref{app:proof_universal_no_go_constraint} implies that $J_{\cE'_{A^nB^nA'B'}}$ is also a feasible solution. Consider further, the linear combination of the feasible solutions is also feasible, i.e.,
    \begin{equation}
    J_{\bar{\cE}_{A^nB^nA'B'}}=\mathbb E_{U,V\sim\mu_{\operatorname{SU}(2)}}J_{\cE'_{A^nB^nA'B'}}=\mathbb E_{W\sim\mu_{\operatorname{SU}(2)^{\ox2}}}(W^{*\ox n}_{A^nB^n}\ox W_{A'B'}) J_{\cE_{A^nB^nA'B'}}(W^{T\ox n}_{A^nB^n}\ox W^\dag_{A'B'}),
    \end{equation}
        is also feasible.
    By the irreducible representation decomposition of 
    \begin{equation}
        \operatorname{SU}(2)\rightarrow \operatorname{SU}(8):\ U\mapsto U^{*\ox 2}\ox U=2U\oplus f_{4}(U),
    \end{equation}
    we have
    \begin{equation}
        \exists\, G\in\operatorname{SU}(8),\ s.t.\ 
        \forall\,U\in\operatorname{SU}(2),\ G(U^{*\ox 2}\ox U)G^\dag=\begin{pmatrix}
            I_2\ox U\\&f_{4}(U)
        \end{pmatrix},
    \end{equation}
    where $f_4$ is a $4$-dimensional irreducible representation of $\operatorname{SU}(2)$. Moreover, by the irreducible representation decomposition of 
    \begin{equation}
        \operatorname{SU}(2)^{\ox2}\rightarrow \operatorname{SU}(64):\ U\ox V\mapsto U^{*\ox 2}\ox V^{*\ox 2}\ox U\ox V=4U\ox V\oplus2U\ox f_{4}(V)\oplus 2V\ox f_{4}(U)\ox f_4(U)\ox f_4(V),
    \end{equation}
    we have $\exists\, \tilde G\in\operatorname{SU}(64),\ s.t.\ 
        \forall\,U\ox V\in\operatorname{SU}(2)^{\ox2}$,
    \begin{equation}
        \tilde G(U^{*\ox 2}\ox V^{*\ox 2}\ox U\ox V)\tilde G^\dag=\begin{pmatrix}
            I_4\ox U\ox V\\&I_2\ox U\ox f_{4}(V)\\&&I_2\ox V\ox f_{4}(U)\\&&&f_{4}(U)\ox f_{4}(V)
        \end{pmatrix},
    \end{equation}
    where $f_4$ is a $4$-dimensional irreducible representation of $\operatorname{SU}(2)$. The following computation is based on a specific $\tilde G$ which could be found in Supplemental Material \ref{appdx:G}. Otherwise, using other $\tilde G$ may introduce basis transformations. Since $J_{\bar{\cE}_{A^nB^nA'B'}}$ is commutative with
    \begin{equation}
        U^{*\ox n}_{A^n}\ox V^{*\ox n}_{B^n}\ox U_{A'}\ox V_{B'},\text{ for any }U,V\in\operatorname{SU}(2),
    \end{equation}
    by mixed Schur–Weyl duality\cite{grinko2025mixed}, we have $\tilde GJ_{\bar{\cE}_{A^nB^nA'B'}}\tilde G^\dag$ is always of form
    \begin{equation} %\label{eq:Gp_con}
        \tilde GJ_{\bar{\cE}_{A^nB^nA'B'}}\tilde G^\dag
        =\begin{pmatrix}
            H_0\ox I_4\\
            &H_1\ox I_8\\
            &&H_2\ox I_8\\
            &&&h_3I_{16}
        \end{pmatrix}=\begin{pmatrix}
            H_0\\
            &H_1\ox I_2\\
            &&H_2\ox I_2\\
            &&&h_3I_{4}
        \end{pmatrix}\ox I_4,
    \end{equation}
    where $H_0\in\mathbb C^{4\times4}$, $H_1,H_2\in\mathbb C^{2\times2}$, $h_3\in\mathbb C$. 
    Simultaneously, we find 
    \begin{align}\label{eq:Gp_TB}
        &\tilde GJ_{\bar{\cE}_{A^nB^nA'B'}}^{T_{B^nB'}}\tilde G^\dagger
        =\begin{pmatrix}
            (I_2\ox Z)H_0^{T_B}(I_2\ox Z)\\
            &H_1\ox I_2\\
            &&ZH_2^TZ\ox I_2\\
            &&&h_3I_{4}
        \end{pmatrix}\ox I_4,\\
        &\operatorname{Spec}(\tr_{A'B'} [J_{\bar{\cE}_{A^nB^nA'B'}}])=\left\{\frac49\bra{00}H_0\ket{00}+\frac89\bra{0}H_1\ket{0}+\frac89\bra{0}H_2\ket{0}+\frac{16}{9}h_3,\right.\\
        &\qquad\left.\frac43\bra{01}H_0\ket{01}+\frac83\bra{1}H_2\ket{1},\frac43\bra{10}H_0\ket{10}+\frac83\bra{1}H_1\ket{1},4\bra{11}H_0\ket{11}\right\},
    \end{align}
    where $\bra{jk}H_0^{T_B}\ket{lr}=\bra{jr}H_0\ket{lk}$ for any $j,k,l,r=0,1$.
    Denote $\tr_d$ as the partial trace on the last $d$-dimensional system, the block representations of $\tr_4[\tilde G M_\psi \tilde G^\dag]$ and $\tr_4[\tilde G M \tilde G^\dag]$ as
    \begin{equation}
        \tr_4[\tilde G M_\psi \tilde G^\dag]=\sum_{j,k=0}^3\ketbra{j}{k}\ox M_\psi^{(jk)},\ \tr_4[\tilde G M \tilde G^\dag]=\sum_{j,k=0}^3\ketbra{j}{k}\ox M^{(jk)},
    \end{equation}
    and we have
    \begin{equation}
        \tr[J_{\bar{\cE}_{A^nB^nA'B'}}M]=\tr[H_0 M^{(00)}]+\tr[H_1\tr_2[M^{(11)}]]+\tr[H_2\tr_2[ M^{(22)}]]+h_3\tr[M^{(33)}].
    \end{equation}
    Then we obtain an SDP related to the constraints \eqref{app:proof_universal_no_go_constraint}
    \begin{subequations}\label{eq:sdp_H1234}
        \begin{align}
        \max\; & \tr[H_0M^{(00)}]+\tr[H_1\tr_2[M^{(11)}]]+\tr[H_2\tr_2[ M^{(22)}]]+h_3\tr[M^{(33)}]\\
        \text{s.t.}\; & H_0,H_1,H_2,h_3,H_0^{T_B}\ge 0;\\
        &\frac49\bra{00}H_0\ket{00}+\frac89\bra{0}H_1\ket{0}+\frac89\bra{0}H_2\ket{0}+\frac{16}{9}h_3\le1;\\
        &\frac43\bra{10}H_0\ket{10}+\frac83\bra{1}H_1\ket{1}\le1;\\
        &\frac43\bra{01}H_0\ket{01}+\frac83\bra{1}H_2\ket{1}\le1;\\
        &4\bra{11}H_0\ket{11}\le1;\\
        & \forall\,\psi,\ \tr[H_0M_\psi^{(00)}]+\tr[H_1\tr_2[M_\psi^{(11)}]]+\tr[H_2\tr_2[ M_\psi^{(22)}]]+h_3\tr[M_\psi^{(33)}]\ge0.
        \end{align}
    \end{subequations}
    and its relaxation
        \begin{subequations}\label{eq:proof_universal_no_go_constraint}
        \begin{align}
        \max\; & \tr[H_0M^{(00)}]+\tr[H_1\tr_2[M^{(11)}]]+\tr[H_2\tr_2[ M^{(22)}]]+h_3\tr[M^{(33)}]\\
        \text{s.t.}\; & H_0,H_1,H_2,h_3,H_0^{T_B}\ge 0;\\
        &\frac49\bra{00}H_0\ket{00}+\frac89\bra{0}H_1\ket{0}+\frac89\bra{0}H_2\ket{0}+\frac{16}{9}h_3\le1;\\
        &\frac43\bra{10}H_0\ket{10}+\frac83\bra{1}H_1\ket{1}\le1;\\
        &\frac43\bra{01}H_0\ket{01}+\frac83\bra{1}H_2\ket{1}\le1;\\
        &4\bra{11}H_0\ket{11}\le1;\\
        & \tr[H_0M_{\ket\Phi^+}^{(00)}]+\tr[H_1\tr_2[M_{\ket\Phi^+}^{(11)}]]+\tr[H_2\tr_2[ M_{\ket\Phi^+}^{(22)}]]+h_3\tr[M_{\ket\Phi^+}^{(33)}]\ge0.
        \end{align}
    \end{subequations}
    It is checked that for any $\gamma\in[0,1]$,
    \begin{align}
       &\tr_2[M^{(11)}]=\tr_2[M^{(22)}]=\operatorname{Diag}\left(-\frac{1}{120} (1-\gamma) \left(27 \gamma^2-56 \gamma+80\right),-\frac{1}{40} \gamma(1-\gamma)(32-15\gamma)\right)\le0,\\
        &\tr[M^{(33)}]=-\frac{1}{60} (1-\gamma) \left(27 \gamma^2-56 \gamma+80\right)\le0,\\
        &\tr_2[M_{\ket\Phi^+}^{(11)}]=\tr_2[M_{\ket\Phi^+}^{(22)}]=\operatorname{Diag}\left(-\frac{1}{24} (1-\gamma) \left(3 \gamma^2-6 \gamma+16\right),-\frac{3}{8} \gamma(1-\gamma)(2-\gamma)\right)\le0,\\
        &\tr[M_{\ket\Phi^+}^{(33)}]=-\frac{1}{12} (1-\gamma) \left(3 \gamma^2-6 \gamma+10\right)\le0.
    \end{align}
    Hence for any optimal solution $(H_0,H_1,H_2,h_3)=(\check H_0,\check H_1,\check H_2,\check h_3)$ of SDP \eqref{eq:proof_universal_no_go_constraint}, $(H_0,H_1,H_2,h_3)=(\check H_0,0,0,0)$ must also be a feasible solution because for any $\gamma\in[0,1]$
    \begin{align}
        &\frac49\bra{00}\check H_0\ket{00}\le\frac49\bra{00}\check H_0\ket{00}+\frac89\bra{0}\check H_1\ket{0}+\frac89\bra{0}\check H_2\ket{0}+\frac{16}{9}\check h_3\le1,\\
        &\frac43\bra{10}\check H_0\ket{10}\le\frac43\bra{10}\check H_0\ket{10}+\frac83\bra{1}\check H_1\ket{1}\le1,\\
        &\frac43\bra{01}\check H_0\ket{01}\le\frac43\bra{01}\check H_0\ket{01}+\frac83\bra{1}\check H_2\ket{1}\le1,\\
        & \tr[\check H_0M_{\ket\Phi^+}^{(00)}]\ge\tr[\check H_0 M_{\ket\Phi^+}^{(00)}]+\tr[\check H_1\tr_2[M_{\ket\Phi^+}^{(11)}]]+\tr[\check H_2\tr_2[M_{\ket\Phi^+}^{(22)}]]+\check h_3\tr[M_{\ket\Phi^+}^{(33)}]\ge0.    
    \end{align}
    Moreover, for any $\gamma\in[0,1]$, by
    \begin{equation}
        \tr[\check H_0 M^{(00)}]+\tr[\check H_1\tr_2[M^{(11)}]]+\tr[\check H_2\tr_2[ M^{(22)}]]+\check h_3\tr[M^{(33)}]\le \tr[\check H_0M^{(00)}]
    \end{equation}
    $(H_0,H_1,H_2,h_3)=(\check H_0,0,0,0)$ is also an optimal solution of SDP \eqref{eq:proof_universal_no_go_constraint}.
    Therefore, we could let $H_1=H_2=h_3=0$, or equivalently
    \begin{equation}
        J_{\bar{\cE}_{A^nB^nA'B'}}
        =\tilde G(\ketbra{00}{00}\ox H_0\ox I_4)\tilde G^\dagger,
    \end{equation}
    and find the following SDP has equivalent maximum with \eqref{eq:proof_universal_no_go_constraint}.
    \begin{subequations}\label{eq:sdp_H1PhiP}
    \begin{align}
    \max\; & \tr[H_0 M^{(00)}]\\
    \text{s.t.}\; & H_0,H_0^{T_B} \ge 0;\\
    &\bra{00}H_0\ket{00}\le\frac94,\ 
    \bra{11}H_0\ket{11}\le\frac14;\\
    &\bra{10}H_0\ket{10}\le\frac34,\ 
    \bra{01}H_0\ket{01}\le\frac34;\\
    & \tr[H_0M_{\ket\Phi^+}^{(00)}]\ge0.
    \end{align}
    \end{subequations}
    
    where
    \begin{align}
    &VM^{(00)}V^\dag=-(1-\gamma) \operatorname{Diag}\left(\frac{1}{16} \gamma  (4-3 \gamma),-\frac{3}{80} \gamma  (4-3 \gamma),\frac{1}{80} \gamma  (32-15 \gamma),\frac{1}{240} \left(27 \gamma ^2-56 \gamma +80\right)
    \right),\\
    &VM_{\ket{\Phi}^+}^{(00)}V^\dag=-(1-\gamma)\left(
    \begin{array}{cccc}
         \frac{1}{16} \gamma  (4-3 \gamma) & 0 & 0 & 0 \\
         0 & -\frac{3}{80} \gamma  (4-3 \gamma ) & 0 & -\frac{1}{20} \gamma  (4-3 \gamma) \\
         0 & 0 & \frac{1}{16} \gamma  (8-3 \gamma ) & 0 \\
         0 & -\frac{1}{20} \gamma  (4-3 \gamma ) & 0 & \frac{1}{240} \left(123 \gamma ^2-264 \gamma +200\right) \\
    \end{array}
    \right),\\
    &\text{for }V=\left(
    \begin{array}{cccc}
         0 & \frac{1}{\sqrt{2}} & \frac{1}{\sqrt{2}} & 0 \\
         -\frac{3}{\sqrt{10}} & 0 & 0 & \frac{1}{\sqrt{10}} \\
         0 & -\frac{1}{\sqrt{2}} & \frac{1}{\sqrt{2}} & 0 \\
         \frac{1}{\sqrt{10}} & 0 & 0 & \frac{3}{\sqrt{10}}
    \end{array}\right).
    \end{align}
%%%%
    Denote the set of $d$-dimensional positive semi-definite Hermitian as $\cL_d^+$. By Peres–Horodecki criterion~\cite{PhysRevLett.77.1413, horodecki2001separability},    
    the $4$-dimensional Hermitian $H_0$ satisfies $H_0\ge0$ and $H_0^{T_B}\ge0$ if and only if $H_0$ is separable, i.e. there exist several $H_{0j}^A,H_{0j}^B\in\cL_2^+$ with $H_0=\sum_{j}H_{0j}^A\ox H_{0j}^B$. Thus we have an equivalent programming with \eqref{eq:sdp_H1PhiP}:
    \begin{subequations}\label{eq:sdp_H1PhiP_SEP}
    \begin{align}
    \max\; & \tr[(\sum_{j}H_{0j}^A\ox H_{0j}^B) M^{(00)}]\\
    \text{s.t.}\; & \forall j,\ H_{0j}^A, H_{0j}^B \in\cL_2^+;\\
    &\bra{00}(\sum_{j}H_{0j}^A\ox H_{0j}^B)\ket{00}\le\frac94,\ 
    \bra{11}(\sum_jH_{0j}^A\ox H_{0j}^B)\ket{11}\le\frac14;\\
    &\bra{10}(\sum_{j}H_{0j}^A\ox H_{0j}^B)\ket{10}\le\frac34,\ 
    \bra{01}(\sum_{j}H_{0j}^A\ox H_{0j}^B)\ket{01}\le\frac34;\\
    & \tr[(\sum_{j}H_{0j}^A\ox H_{0j}^B)M_{\ket\Phi^+}^{(00)}]\ge0.
    \end{align}
    \end{subequations}
    We claim that 
    \begin{align}
        &H_{0}^A,H_{0}^B\in\cL_2^+,\ \gamma\in(0,1)\implies \tr[(H_{0}^A\ox H_{0}^B)M_{\ket\Phi^+}^{(00)}]\le0,\label{eq:claim1}\\
        &H_{0}^A,H_{0}^B\in\cL_2^+,\ \gamma\in(0,1),\ \tr[(H_{0}^A\ox H_{0}^B)M_{\ket\Phi^+}^{(00)}]=0\implies \tr[(H_{0}^A\ox H_{0}^B)M^{(00)}]\le0,\label{eq:claim2}
    \end{align}
    which will be proved after this proof. 
    By these two claims, we find $0$ is an upper bound of the maximum of \eqref{eq:sdp_H1PhiP_SEP} because
    \begin{align}
        &\gamma\in(0,1)\bigwedge(\forall j,\ H_{0j}^A,H_{0j}^B\in\cL_2^+)\bigwedge  \tr[(\sum_{j}H_{0j}^A\ox H_{0j}^B)M_{\ket\Phi^+}^{(00)}]\ge0\\
        \xRightarrow{\eqref{eq:claim1}} &\gamma\in(0,1)\bigwedge(\forall j,\ H_{0j}^A,H_{0j}^B\in\cL_2^+,\tr[(H_{0j}^A\ox H_{0j}^B)M_{\ket\Phi^+}^{(00)}]\le0)\bigwedge  \tr[(\sum_{j}H_{0j}^A\ox H_{0j}^B)M_{\ket\Phi^+}^{(00)}]\ge0\\
        \iff &\gamma\in(0,1)\bigwedge(\forall j,\ H_{0j}^A,H_{0j}^B\in\cL_2^+,\tr[(H_{0j}^A\ox H_{0j}^B)M_{\ket\Phi^+}^{(00)}]=0)\\
        \xRightarrow{\eqref{eq:claim2}} & \forall j,\ \tr[(H_{0j}^A\ox H_{0j}^B)M^{(00)}]\le0\\
        \implies & \tr[(\sum_jH_{0j}^A\ox H_{0j}^B)M^{(00)}]\le0.
    \end{align}
    On the other hand, feasible solution with each $H_{0j}^A=H_{0j}^B=0$ suggests $0$ is also an lower bound of the maximum of \eqref{eq:sdp_H1PhiP_SEP}. As a result, the maximal values of programming \eqref{eq:sdp_H1PhiP_SEP} and SDPs \eqref{eq:sdp_H1PhiP} and \eqref{eq:proof_universal_no_go_constraint} are all $0$, and the maximal value of \eqref{eq:sdp_H1234} is at most $0$, which implies that \eqref{app:proof_universal_no_go_constraint} has no feasible solution.
\end{proof}

\begin{proof}\hspace{-0.2em}\textbf{of \eqref{eq:claim1}.}
    Claim \eqref{eq:claim1} could be considered as a corollary of following Theorem~\ref{theo:no-go_4_bell}, which is proved independent of this section. Here we could give another proof based on symbolic computation automatic reasoning using the cylindrical algebraic decomposition algorithm~\cite{collins1976quantifier, collins1998quantifier, caviness2012quantifier, collins1991partial}. After setting 
    \begin{equation}
        H_0^A=\left(\begin{array}{cc}
            a & b+i d \\
            b-i d & c \\
        \end{array}\right),\ 
        H_0^B=\left(\begin{array}{cc}
            e & f+i h \\
            f-i h & g \\
        \end{array}\right),
    \end{equation}
    with real variables $\{a,b,c,d,e,f,g,h\}$, we could use a system of inequalities
    \begin{equation}
        \{a\ge 0,c \ge 0,
        a c \ge b^2 + d^2,
        e\ge 0,g \ge 0,
        e g \ge f^2 + h^2\}
    \end{equation}
    to indicate $H_{0}^A,H_{0}^B\ge0$ equivalently.
    Now we introduce another system of inequalities
    \begin{equation}\label{eq:ineqalities1}
    \begin{aligned}    
        &\{a\ge 0,c \ge 0,
        a c \ge b^2 + d^2,
        e\ge 0,g \ge 0,
        e g \ge f^2 + h^2,0<\gamma<1,\\
        &\ (-3 \gamma ^2+6 \gamma -4) a e + 9\gamma(
    \gamma-2) (a g + c e )-9 \left(3 \gamma ^2-6 \gamma +4\right) c g + 24(\gamma - 1) b f + 24 d h>0\}    
    \end{aligned}
    \end{equation}
    where the left hand of the last inequality equals to $48\tr[(H_{0}^A\ox H_{0}^B)M_{\ket\Phi^+}^{(00)}]/(1-\gamma)$ because
    \begin{equation}
        48M_{\ket\Phi^+}^{(00)}/(1-\gamma)=\left(
\begin{array}{cccc}
 -3 \gamma ^2+6 \gamma -4 & 0 & 0 & 6 (\gamma -2) \\
 0 & 9 (\gamma -2) \gamma  & 6 \gamma  & 0 \\
 0 & 6 \gamma  & 9 (\gamma -2) \gamma  & 0 \\
 6 (\gamma -2) & 0 & 0 & -9 \left(3 \gamma ^2-6 \gamma +4\right)
\end{array}
\right).
    \end{equation}
    Thus the semialgebraic set defined by \eqref{eq:ineqalities1} is exactly the feasible solutions of constraints 
    \begin{equation}
        H_{0}^A,H_{0}^B\in\cL_2^+,\ \gamma\in(0,1),\tr[(H_{0}^A\ox H_{0}^B)M_{\ket\Phi^+}^{(00)}]>0.
    \end{equation}
    However, when we use the cylindrical algebraic decomposition algorithm~\cite{collins1976quantifier, collins1998quantifier, caviness2012quantifier, collins1991partial}
    on \eqref{eq:ineqalities1}, we will always obtain {\bf False}, which derives that claim \eqref{eq:claim1} always holds.
\end{proof}

\begin{proof}\hspace{-0.2em}\textbf{of \eqref{eq:claim2}.}
    Analogous to the proof of Claim \eqref{eq:claim1}, Claim \eqref{eq:claim2} could also be proved by the cylindrical algebraic decomposition algorithm from symbolic computation automatic reasoning. To prove Claim \eqref{eq:claim2}, it is equivalent to prove the following constrains
    \begin{align}\label{eq:claim2_constrains}
    \begin{cases}
        \gamma\in(0,1),\\
        H_0^A, H_0^B\in\cL_2^+,\\
        \tr[(H_0^A\ox H_0^B) M_{\ket\Phi^+}^{(00)}]/(1-\gamma)=0,\\
        \tr[(H_0^A\ox H_0^B) M^{(00)}]/(1-\gamma)>0,
    \end{cases}
    \end{align}
    are not feasible.
    After setting 
    \begin{equation}
        H_0^A=\left(\begin{array}{cc}
            a & b+i d \\
            b-i d & c \\
        \end{array}\right),\ 
        H_0^B=\left(\begin{array}{cc}
            e & f+i h \\
            f-i h & g \\
        \end{array}\right),
    \end{equation}
    with real variables $\{a,b,c,d,e,f,g,h\}$, we have established a system of inequalities equivalent to constrains \eqref{eq:claim2_constrains}:
    \begin{equation}\label{eq:ineqalities2}
    \begin{aligned}    
        &\{a\ge 0,c \ge 0,
        a c \ge b^2 + d^2,
        e\ge 0,g \ge 0,
        e g \ge f^2 + h^2,0<\gamma<1,\\
        &\ (-3 \gamma ^2+6 \gamma -4) a e + 9\gamma(
    \gamma-2) (a g + c e )-9 \left(3 \gamma ^2-6 \gamma +4\right) c g + 24(\gamma - 1) b f + 24 d h=0,\\
    &\ \left(-27 \gamma ^2+38 \gamma -8\right) ae+3  \gamma  (15 \gamma -26) (ag+ce)-9  \left(3 \gamma ^2-6 \gamma +8\right) cg+48  (\gamma -1) bf+24 (\gamma +2) d h>0
    \}    
    \end{aligned}
    \end{equation}
    where the left hand of the last inequality equals to $240\tr[(H_{0}^A\ox H_{0}^B)M^{(00)}]/(1-\gamma)$ because
    \begin{equation}
        240M^{(00)}/(1-\gamma)=\left(\begin{array}{cccc}
             -27 \gamma ^2+38 \gamma -8 & 0 & 0 & 6 (\gamma -4) \\
             0 & 3 \gamma  (15 \gamma -26) & 18 \gamma  & 0 \\
             0 & 18 \gamma  & 3 \gamma  (15 \gamma -26) & 0 \\
             6 (\gamma -4) & 0 & 0 & -9 \left(3 \gamma ^2-6 \gamma +8\right)
        \end{array}\right).
    \end{equation}
    When we use the cylindrical algebraic decomposition algorithm~\cite{collins1976quantifier, collins1998quantifier, caviness2012quantifier, collins1991partial}
    on \eqref{eq:ineqalities2}, we will always obtain {\bf False}, which derives that claim \eqref{eq:claim2} always holds. 
    The time cost of the cylindrical algebraic decomposition algorithm depends on the order of the variables. The time cost would be 2.5 seconds with order $(\gamma,a,e,c,g,b,d,f,h)$, but 8.5 hours with order $(\gamma,a,b,c,d,e,f,g,h)$ when we use the ``CylindricalDecomposition" function~\cite{reference.wolfram_2025_cylindricaldecomposition} 
    (or ``Reduce" function~\cite{reference.wolfram_2025_reduce})
    in Wolfram Mathematica~\cite{Mathematica} 
    on a workstation equipped with an Intel Core i9-14900K CPU and 48 GB RAM (2×24 GB, 8000 MT/s). 
\end{proof}

By Schmidt decomposition~\cite{nielsen2010quantum}` , any pure state of dimension $4$ could be rewritten as
$\ket{\psi}=\sigma_0\ket{\psi_0^A}\ket{\psi_0^B}+\sigma_1\ket{\psi_1^A}\ket{\psi_1^B}$ 
with $\sigma_0,\sigma_1\ge0$, $\braket{\psi_j^A}{\psi_k^A}=\braket{\psi_j^B}{\psi_k^B}=\delta_{jk}$ which is Kronecker delta. Denote 
\begin{equation}
    q=\sigma_0^2\sigma_1^2\in[0,\frac14],
\end{equation}
and we have all of $M_{\ket{\psi}}^{(00)}\tr_2[M_{\ket{\psi}}^{(11)}],\tr_2[M_{\ket{\psi}}^{(22)}],\tr[M_{\ket{\psi}}^{(33)}]$ only depend on $\gamma$ and $q$, and moreover linear on $q$:
\begin{equation}
    {\footnotesize
        VM_{\ket{\psi}}^{(00)}V^\dag=\frac{1-\gamma}{720}\left(\begin{array}{cccc}
        45 \gamma  (3 \gamma -4) & 0 & 0 & 0 \\
        0 & 27 \gamma  (4-3 \gamma) & 0 & 24 \gamma  (4-3 \gamma) (10 q-1) \\
        0 & 0 & 15 \gamma  (9 \gamma -32 q-16) & 0 \\
        0 & 24 \gamma  (4-3 \gamma) (10 q-1) & 0 & 111 \gamma ^2-248 \gamma -1920 \gamma ^2 q+4160 \gamma  q-2400 q
        \end{array}\right)
    }
\end{equation}
\begin{align}
    &\tr_2[M_{\ket{\psi}}^{(11)}]=\tr_2[M_{\ket{\psi}}^{(22)}]=\frac{1-\gamma}{72}\operatorname{Diag}\left(
    -21 \gamma ^2+44 \gamma +48 \gamma ^2 q-104 \gamma  q-48,3 \gamma  (9 \gamma +8 q-20)
    \right)\\
    &\tr[M_{\ket{\psi}}^{(33)}]=\frac{1-\gamma}{36}\left(-21 \gamma ^2+44 \gamma +48 \gamma ^2 q-104 \gamma  q+120 q-60\right),
\end{align}
Furthermore, it is checked that $\tr_2[M_{\ket{\psi}}^{(11)}],\tr_2[M_{\ket{\psi}}^{(22)}]$ are both negative semi-definite, and $\tr[M_{\ket{\psi}}^{(33)}]$ is always non-positive for $\gamma\in(0,1)$ and $q\in[0,\frac14]$. Especially, we find that 
\begin{equation}
    q=\frac{1}{10}\implies (M_{\ket{\psi}}^{(00)}\tr_2[M_{\ket{\psi}}^{(11)}],\tr_2[M_{\ket{\psi}}^{(22)}],\tr[M_{\ket{\psi}}^{(33)}])=(M^{(00)},\tr_2[M^{(11)}],\tr_2[M^{(22)}],\tr[M^{(33)}]),
\end{equation}
which indicates that the maximum point of $\tr[H_0M^{(00)}]$ is the same as the maximum point of $\tr[H_0M_{\ket{\psi}}^{(00)}]$ for any pure state $\ket{\psi}$, which implies that

\begin{corollary}
    [No-go theorem for the maximally entangled state set]\textit{For the two-qubit depolarizing noise channel $\cN_{AB}^\gamma$ with noise level $\gamma$, there is no non-trivial $2\rightarrow 1$ LOCC purification protocol for any set containing all maximally entangled pure states.}
\end{corollary}

\clearpage
\section{The specific $\tilde G$ in Supplemental Material \ref{app:proof_universal_no-go}}\label{appdx:G}
The matrix $\Tilde{G}$ is a sparse matrix, and all the non-zero elements are shown as follows:

\begin{equation}
    \begin{aligned}
    &\{1,13\}\to \frac{1}{6},\{1,25\}\to \frac{1}{6},\{1,30\}\to \frac{1}{3},\{1,37\}\to \frac{1}{6},\{1,47\}\to \frac{1}{3},\{1,49\}\to \frac{1}{6},\{1,54\}\to \frac{1}{3},\{1,59\}\to \frac{1}{3},\{1,64\}\to \frac{2}{3},\nonumber\\
    &\{2,9\}\to -\frac{1}{3},\{2,14\}\to -\frac{1}{6},\{2,26\}\to -\frac{1}{6},\{2,33\}\to -\frac{1}{3},\{2,38\}\to -\frac{1}{6},\{2,43\}\to -\frac{2}{3},\{2,48\}\to -\frac{1}{3},\{2,50\}\to -\frac{1}{6},\nonumber\\
    &\{2,60\}\to -\frac{1}{3},\{3,5\}\to -\frac{1}{3},\{3,15\}\to -\frac{1}{6},\{3,17\}\to -\frac{1}{3},\{3,22\}\to -\frac{2}{3},\{3,27\}\to -\frac{1}{6},\{3,32\}\to -\frac{1}{3},\{3,39\}\to -\frac{1}{6},\nonumber\\
    &\{3,51\}\to -\frac{1}{6},\{3,56\}\to -\frac{1}{3},\{4,1\}\to \frac{2}{3},\{4,6\}\to \frac{1}{3},\{4,11\}\to \frac{1}{3},\{4,16\}\to \frac{1}{6},\{4,18\}\to \frac{1}{3},\{4,28\}\to \frac{1}{6},\{4,35\}\to \frac{1}{3},\nonumber\\
    &\{4,40\}\to \frac{1}{6},\{4,52\}\to \frac{1}{6},\{5,13\}\to \frac{i}{2 \sqrt{3}},\{5,25\}\to -\frac{i}{2 \sqrt{3}},\{5,37\}\to \frac{i}{2 \sqrt{3}},\{5,47\}\to \frac{i}{\sqrt{3}},\{5,49\}\to -\frac{i}{2 \sqrt{3}},\nonumber\\
    &\{5,59\}\to -\frac{i}{\sqrt{3}},\{6,14\}\to \frac{i}{2 \sqrt{3}},\{6,26\}\to -\frac{i}{2 \sqrt{3}},\{6,38\}\to \frac{i}{2 \sqrt{3}},\{6,48\}\to \frac{i}{\sqrt{3}},\{6,50\}\to -\frac{i}{2 \sqrt{3}},\{6,60\}\to -\frac{i}{\sqrt{3}},\nonumber\\
    &\{7,5\}\to -\frac{i}{\sqrt{3}},\{7,15\}\to -\frac{i}{2 \sqrt{3}},\{7,17\}\to \frac{i}{\sqrt{3}},\{7,27\}\to \frac{i}{2 \sqrt{3}},\{7,39\}\to -\frac{i}{2 \sqrt{3}},\{7,51\}\to \frac{i}{2 \sqrt{3}},\{8,6\}\to -\frac{i}{\sqrt{3}},\nonumber\\
    &\{8,16\}\to -\frac{i}{2 \sqrt{3}},\{8,18\}\to \frac{i}{\sqrt{3}},\{8,28\}\to \frac{i}{2 \sqrt{3}},\{8,40\}\to -\frac{i}{2 \sqrt{3}},\{8,52\}\to \frac{i}{2 \sqrt{3}},\{9,13\}\to \frac{i}{2 \sqrt{3}},\{9,25\}\to \frac{i}{2 \sqrt{3}},\nonumber\\
    &\{9,30\}\to \frac{i}{\sqrt{3}},\{9,37\}\to -\frac{i}{2 \sqrt{3}},\{9,49\}\to -\frac{i}{2 \sqrt{3}},\{9,54\}\to -\frac{i}{\sqrt{3}},\{10,9\}\to -\frac{i}{\sqrt{3}},\{10,14\}\to -\frac{i}{2 \sqrt{3}},\{10,26\}\to -\frac{i}{2 \sqrt{3}},\nonumber\\
    &\{10,33\}\to \frac{i}{\sqrt{3}},\{10,38\}\to \frac{i}{2 \sqrt{3}},\{10,50\}\to \frac{i}{2 \sqrt{3}},\{11,15\}\to \frac{i}{2 \sqrt{3}},\{11,27\}\to \frac{i}{2 \sqrt{3}},\{11,32\}\to \frac{i}{\sqrt{3}},\{11,39\}\to -\frac{i}{2 \sqrt{3}},\nonumber\\
    &\{11,51\}\to -\frac{i}{2 \sqrt{3}},\{11,56\}\to -\frac{i}{\sqrt{3}},\{12,11\}\to -\frac{i}{\sqrt{3}},\{12,16\}\to -\frac{i}{2 \sqrt{3}},\{12,28\}\to -\frac{i}{2 \sqrt{3}},\{12,35\}\to \frac{i}{\sqrt{3}},\nonumber\\
    &\{12,40\}\to \frac{i}{2 \sqrt{3}},\{12,52\}\to \frac{i}{2 \sqrt{3}},\{13,13\}\to -\frac{1}{2},\{13,25\}\to \frac{1}{2},\{13,37\}\to \frac{1}{2},\{13,49\}\to -\frac{1}{2},\{14,14\}\to -\frac{1}{2},\nonumber\\
    &\{14,26\}\to \frac{1}{2},\{14,38\}\to \frac{1}{2},\{14,50\}\to -\frac{1}{2},\{15,15\}\to -\frac{1}{2},\{15,27\}\to \frac{1}{2},\{15,39\}\to \frac{1}{2},\{15,51\}\to -\frac{1}{2},\{16,16\}\to -\frac{1}{2},\nonumber\\
    &\{16,28\}\to \frac{1}{2},\{16,40\}\to \frac{1}{2},\{16,52\}\to -\frac{1}{2},\{17,29\}\to \frac{1}{\sqrt{6}},\{17,53\}\to \frac{1}{\sqrt{6}},\{17,63\}\to \sqrt{\frac{2}{3}},\{18,13\}\to -\frac{1}{3 \sqrt{2}},\nonumber\\
    &\{18,25\}\to -\frac{1}{3 \sqrt{2}},\{18,30\}\to \frac{1}{3 \sqrt{2}},\{18,37\}\to -\frac{1}{3 \sqrt{2}},\{18,47\}\to -\frac{\sqrt{2}}{3},\{18,49\}\to -\frac{1}{3 \sqrt{2}},\{18,54\}\to \frac{1}{3 \sqrt{2}},\nonumber\\
    &\{18,59\}\to -\frac{\sqrt{2}}{3},\{18,64\}\to \frac{\sqrt{2}}{3},\{19,21\}\to -\sqrt{\frac{2}{3}},\{19,31\}\to -\frac{1}{\sqrt{6}},\{19,55\}\to -\frac{1}{\sqrt{6}},\{20,5\}\to \frac{\sqrt{2}}{3},\{20,15\}\to \frac{1}{3 \sqrt{2}},\nonumber\\
    &\{20,17\}\to \frac{\sqrt{2}}{3},\{20,22\}\to -\frac{\sqrt{2}}{3},\{20,27\}\to \frac{1}{3 \sqrt{2}},\{20,32\}\to -\frac{1}{3 \sqrt{2}},\{20,39\}\to \frac{1}{3 \sqrt{2}},\{20,51\}\to \frac{1}{3 \sqrt{2}},\nonumber\\
    &\{20,56\}\to -\frac{1}{3 \sqrt{2}},\{21,9\}\to \frac{1}{3 \sqrt{2}},\{21,14\}\to -\frac{1}{3 \sqrt{2}},\{21,26\}\to -\frac{1}{3 \sqrt{2}},\{21,33\}\to \frac{1}{3 \sqrt{2}},\{21,38\}\to -\frac{1}{3 \sqrt{2}},\nonumber\\
    &\{21,43\}\to \frac{\sqrt{2}}{3},\{21,48\}\to -\frac{\sqrt{2}}{3},\{21,50\}\to -\frac{1}{3 \sqrt{2}},\{21,60\}\to -\frac{\sqrt{2}}{3},\{22,10\}\to \frac{1}{\sqrt{6}},\{22,34\}\to \frac{1}{\sqrt{6}},\{22,44\}\to \sqrt{\frac{2}{3}},\nonumber\\
    &\{23,1\}\to -\frac{\sqrt{2}}{3},\{23,6\}\to \frac{\sqrt{2}}{3},\{23,11\}\to -\frac{1}{3 \sqrt{2}},\{23,16\}\to \frac{1}{3 \sqrt{2}},\{23,18\}\to \frac{\sqrt{2}}{3},\{23,28\}\to \frac{1}{3 \sqrt{2}},\{23,35\}\to -\frac{1}{3 \sqrt{2}},\nonumber\\
    &\{23,40\}\to \frac{1}{3 \sqrt{2}},\{23,52\}\to \frac{1}{3 \sqrt{2}},\{24,2\}\to -\sqrt{\frac{2}{3}},\{24,12\}\to -\frac{1}{\sqrt{6}},\{24,36\}\to -\frac{1}{\sqrt{6}},\{25,29\}\to \frac{i}{\sqrt{2}},\{25,53\}\to -\frac{i}{\sqrt{2}},\nonumber\\
    &\{26,13\}\to -\frac{i}{\sqrt{6}},\{26,25\}\to -\frac{i}{\sqrt{6}},\{26,30\}\to \frac{i}{\sqrt{6}},\{26,37\}\to \frac{i}{\sqrt{6}},\{26,49\}\to \frac{i}{\sqrt{6}},\{26,54\}\to -\frac{i}{\sqrt{6}},\{27,31\}\to \frac{i}{\sqrt{2}},\nonumber\\
    &\{27,55\}\to -\frac{i}{\sqrt{2}},\{28,15\}\to -\frac{i}{\sqrt{6}},\{28,27\}\to -\frac{i}{\sqrt{6}},\{28,32\}\to \frac{i}{\sqrt{6}},\{28,39\}\to \frac{i}{\sqrt{6}},\{28,51\}\to \frac{i}{\sqrt{6}},\{28,56\}\to -\frac{i}{\sqrt{6}},\nonumber\\
\end{aligned}
\end{equation}

\begin{equation}
\begin{aligned}
    &\{29,9\}\to \frac{i}{\sqrt{6}},\{29,14\}\to -\frac{i}{\sqrt{6}},\{29,26\}\to -\frac{i}{\sqrt{6}},\{29,33\}\to -\frac{i}{\sqrt{6}},\{29,38\}\to \frac{i}{\sqrt{6}},\{29,50\}\to \frac{i}{\sqrt{6}},\{30,10\}\to \frac{i}{\sqrt{2}}, \nonumber\\
    &\{30,34\}\to -\frac{i}{\sqrt{2}},\{31,11\}\to \frac{i}{\sqrt{6}},\{31,16\}\to -\frac{i}{\sqrt{6}},\{31,28\}\to -\frac{i}{\sqrt{6}},\{31,35\}\to -\frac{i}{\sqrt{6}},\{31,40\}\to \frac{i}{\sqrt{6}},\{31,52\}\to \frac{i}{\sqrt{6}},\nonumber\\
    &\{32,12\}\to \frac{i}{\sqrt{2}},\{32,36\}\to -\frac{i}{\sqrt{2}},\{33,45\}\to \frac{1}{\sqrt{6}},\{33,57\}\to \frac{1}{\sqrt{6}},\{33,62\}\to \sqrt{\frac{2}{3}},\{34,41\}\to -\sqrt{\frac{2}{3}},\{34,46\}\to -\frac{1}{\sqrt{6}},\nonumber\\
    &\{34,58\}\to -\frac{1}{\sqrt{6}},\{35,13\}\to -\frac{1}{3 \sqrt{2}},\{35,25\}\to -\frac{1}{3 \sqrt{2}},\{35,30\}\to -\frac{\sqrt{2}}{3},\{35,37\}\to -\frac{1}{3 \sqrt{2}},\{35,47\}\to \frac{1}{3 \sqrt{2}},\nonumber\\
    &\{35,49\}\to -\frac{1}{3 \sqrt{2}},\{35,54\}\to -\frac{\sqrt{2}}{3},\{35,59\}\to \frac{1}{3 \sqrt{2}},\{35,64\}\to \frac{\sqrt{2}}{3},\{36,9\}\to \frac{\sqrt{2}}{3},\{36,14\}\to \frac{1}{3 \sqrt{2}},\{36,26\}\to \frac{1}{3 \sqrt{2}},\nonumber\\
    &\{36,33\}\to \frac{\sqrt{2}}{3},\{36,38\}\to \frac{1}{3 \sqrt{2}},\{36,43\}\to -\frac{\sqrt{2}}{3},\{36,48\}\to -\frac{1}{3 \sqrt{2}},\{36,50\}\to \frac{1}{3 \sqrt{2}},\{36,60\}\to -\frac{1}{3 \sqrt{2}},\nonumber\\
    &\{37,5\}\to \frac{1}{3 \sqrt{2}},\{37,15\}\to -\frac{1}{3 \sqrt{2}},\{37,17\}\to \frac{1}{3 \sqrt{2}},\{37,22\}\to \frac{\sqrt{2}}{3},\{37,27\}\to -\frac{1}{3 \sqrt{2}},\{37,32\}\to -\frac{\sqrt{2}}{3},\nonumber\\
    &\{37,39\}\to -\frac{1}{3 \sqrt{2}},\{37,51\}\to -\frac{1}{3 \sqrt{2}},\{37,56\}\to -\frac{\sqrt{2}}{3},\{38,1\}\to -\frac{\sqrt{2}}{3},\{38,6\}\to -\frac{1}{3 \sqrt{2}},\{38,11\}\to \frac{\sqrt{2}}{3},\nonumber\\
    &\{38,16\}\to \frac{1}{3 \sqrt{2}},\{38,18\}\to -\frac{1}{3 \sqrt{2}},\{38,28\}\to \frac{1}{3 \sqrt{2}},\{38,35\}\to \frac{\sqrt{2}}{3},\{38,40\}\to \frac{1}{3 \sqrt{2}},\{38,52\}\to \frac{1}{3 \sqrt{2}},\{39,7\}\to \frac{1}{\sqrt{6}},\nonumber\\
    &\{39,19\}\to \frac{1}{\sqrt{6}},\{39,24\}\to \sqrt{\frac{2}{3}},\{40,3\}\to -\sqrt{\frac{2}{3}},\{40,8\}\to -\frac{1}{\sqrt{6}},\{40,20\}\to -\frac{1}{\sqrt{6}},\{41,45\}\to \frac{i}{\sqrt{2}},\{41,57\}\to -\frac{i}{\sqrt{2}},\nonumber\\
    &\{42,46\}\to \frac{i}{\sqrt{2}},\{42,58\}\to -\frac{i}{\sqrt{2}},\{43,13\}\to -\frac{i}{\sqrt{6}},\{43,25\}\to \frac{i}{\sqrt{6}},\{43,37\}\to -\frac{i}{\sqrt{6}},\{43,47\}\to \frac{i}{\sqrt{6}},\{43,49\}\to \frac{i}{\sqrt{6}},\nonumber\\
    &\{43,59\}\to -\frac{i}{\sqrt{6}},\{44,14\}\to -\frac{i}{\sqrt{6}},\{44,26\}\to \frac{i}{\sqrt{6}},\{44,38\}\to -\frac{i}{\sqrt{6}},\{44,48\}\to \frac{i}{\sqrt{6}},\{44,50\}\to \frac{i}{\sqrt{6}},\{44,60\}\to -\frac{i}{\sqrt{6}},\nonumber\\
    &\{45,5\}\to \frac{i}{\sqrt{6}},\{45,15\}\to -\frac{i}{\sqrt{6}},\{45,17\}\to -\frac{i}{\sqrt{6}},\{45,27\}\to \frac{i}{\sqrt{6}},\{45,39\}\to -\frac{i}{\sqrt{6}},\{45,51\}\to \frac{i}{\sqrt{6}},\{46,6\}\to \frac{i}{\sqrt{6}},\nonumber\\
    &\{46,16\}\to -\frac{i}{\sqrt{6}},\{46,18\}\to -\frac{i}{\sqrt{6}},\{46,28\}\to \frac{i}{\sqrt{6}},\{46,40\}\to -\frac{i}{\sqrt{6}},\{46,52\}\to \frac{i}{\sqrt{6}},\{47,7\}\to \frac{i}{\sqrt{2}},\{47,19\}\to -\frac{i}{\sqrt{2}},\nonumber\\
    &\{48,8\}\to \frac{i}{\sqrt{2}},\{48,20\}\to -\frac{i}{\sqrt{2}},\{49,61\}\to 1,\{50,45\}\to -\frac{1}{\sqrt{3}},\{50,57\}\to -\frac{1}{\sqrt{3}},\{50,62\}\to \frac{1}{\sqrt{3}},\{51,29\}\to -\frac{1}{\sqrt{3}},\nonumber\\
    &\{51,53\}\to -\frac{1}{\sqrt{3}},\{51,63\}\to \frac{1}{\sqrt{3}},\{52,13\}\to \frac{1}{3},\{52,25\}\to \frac{1}{3},\{52,30\}\to -\frac{1}{3},\{52,37\}\to \frac{1}{3},\{52,47\}\to -\frac{1}{3},\nonumber\\
    &\{52,49\}\to \frac{1}{3},\{52,54\}\to -\frac{1}{3},\{52,59\}\to -\frac{1}{3},\{52,64\}\to \frac{1}{3},\{53,21\}\to \frac{1}{\sqrt{3}},\{53,31\}\to -\frac{1}{\sqrt{3}},\{53,55\}\to -\frac{1}{\sqrt{3}},\nonumber\\
    &\{54,5\}\to -\frac{1}{3},\{54,15\}\to \frac{1}{3},\{54,17\}\to -\frac{1}{3},\{54,22\}\to \frac{1}{3},\{54,27\}\to \frac{1}{3},\{54,32\}\to -\frac{1}{3},\{54,39\}\to \frac{1}{3},\{54,51\}\to \frac{1}{3},\nonumber\\
    &\{54,56\}\to -\frac{1}{3},\{55,23\}\to 1,\{56,7\}\to -\frac{1}{\sqrt{3}},\{56,19\}\to -\frac{1}{\sqrt{3}},\{56,24\}\to \frac{1}{\sqrt{3}},\{57,41\}\to \frac{1}{\sqrt{3}},\{57,46\}\to -\frac{1}{\sqrt{3}},\nonumber\\
    &\{57,58\}\to -\frac{1}{\sqrt{3}},\{58,42\}\to 1,\{59,9\}\to -\frac{1}{3},\{59,14\}\to \frac{1}{3},\{59,26\}\to \frac{1}{3},\{59,33\}\to -\frac{1}{3},\{59,38\}\to \frac{1}{3},\{59,43\}\to \frac{1}{3},\nonumber\\
    &\{59,48\}\to -\frac{1}{3},\{59,50\}\to \frac{1}{3},\{59,60\}\to -\frac{1}{3},\{60,10\}\to -\frac{1}{\sqrt{3}},\{60,34\}\to -\frac{1}{\sqrt{3}},\{60,44\}\to \frac{1}{\sqrt{3}},\{61,1\}\to \frac{1}{3},\nonumber\\
    &\{61,6\}\to -\frac{1}{3},\{61,11\}\to -\frac{1}{3},\{61,16\}\to \frac{1}{3},\{61,18\}\to -\frac{1}{3},\{61,28\}\to \frac{1}{3},\{61,35\}\to -\frac{1}{3},\{61,40\}\to \frac{1}{3},\{61,52\}\to \frac{1}{3},\nonumber\\
    &\{62,2\}\to \frac{1}{\sqrt{3}},\{62,12\}\to -\frac{1}{\sqrt{3}},\{62,36\}\to -\frac{1}{\sqrt{3}},\{63,3\}\to \frac{1}{\sqrt{3}},\{63,8\}\to -\frac{1}{\sqrt{3}},\{63,20\}\to -\frac{1}{\sqrt{3}},\{64,4\}\to 1.
\end{aligned}
\end{equation}

%%%%%%%%%%%%%%%%%%%%%%%%%%%%%%%%%%%%%%%%%%%%%%%%%%%%%%%%%%%%
\section{Semidefinite program for optimal average PPT purification fidelity}\label{app:PPT_bound}

\begin{definition}
For a given pure state set $\cS$, noise channel $\cN$ and the average success probability $\bar{P}$, the optimal $n\rightarrow 1$ average PPT purification fidelity is defined by 
    \begin{align}
        \bar{F}_{\rm PPT}^*(n,\bar{P};\cS,\cN)
        = \max_{\cE}   \big\{\bar{F}(n;\cS,\cN,\cE) \,\big| \,\cE\in{\rm PPT} \cap {\rm CPTN}, 
        \frac
        {1}{|S|} \sum_{\psi\in\cS} \tr[\cE(\cN(\psi)^{\ox n})]=\bar{P}\big\}.
    \end{align}
\end{definition}

\begin{proposition}\label{prop:SDP_for_discrete_set}
    For a given pure state set $\cS$, noise channel $\cN$, and the average success probability $\bar{P}$, the optimal purification fidelity achieved by the $n\rightarrow 1$ PPT purification protocol $\cE$ can be characterized by the following SDP.
    \begin{subequations}\label{eq:SDP_for_PPT_bound}
        \begin{align}
            \bar{F}_{\rm PPT}^*(n, \bar{P};\cS, \cN) = \max\; & \frac{1}{\bar{P}} \tr[Q_{A^nB^nA'B'} J_{\cE_{A^nB^nA'B'}}^{T_{A^nB^n}}];\\
            \text{s.t.}\; & \tr[R_{A^nB^nA'B'} J_{\cE_{A^nB^nA'B'}} ^{T_{A^nB^n}}] = \bar{P};\label{eq:SDP_optimal_prob}\\
            & J_{\cE_{A^nB^nA'B'}} \ge 0; \; \tr_{A'B'} [J_{\cE_{A^nB^nA'B'}}] \le I_{A^nB^n};\label{eq:SDP_optimal_CPTN}\\
            & J_{\cE_{A^nB^nA'B'}}^{T_{A^nA'}} \ge 0,\label{eq:SDP_optimal_PPT}\\
            & \frac{\tr[J_{\cE_{A^nB^nA'B'}} ^{T_{A^nB^n}} \cN^{\ox n}\ox \cI(\psi_i^{\ox n+1})]}{\tr[J_{\cE_{A^nB^nA'B'}} ^{T_{A^nB^n}} \cN^{\ox n}(\psi_i^{\ox n})\ox I_{A'B'}]} \ge \tr[\psi_i \cN(\psi_i)] \; \forall \; \psi_i\in\cS \label{eq:SDP_optimal_all_state}, 
        \end{align}
    \end{subequations}
    where $J_{\cE_{A^nB^nA'B'}}$ is the \Choi matrix of map $\cE$, $T$ refers to the transpose operation. In addition \begin{align}
        Q_{A^nB^nA'B'}&=\frac{1}{|\cS|}\sum_{\psi_i\in\cS}\cN^{\ox n} \ox \cI \left(\psi_i^{\ox n+1}\right)\\
        R_{A^nB^nA'B'}& = \frac{1}{|\cS|}\sum_{\psi_i\in\cS}\cN^{\ox n}\left(\psi_i^{\ox n}\right) \ox I
    \end{align}
    where $|\cS|$ is the size of the set $\cS$, $\cI$ and $I$ refers to identity channel and identity matrix, respectively.
\end{proposition}

\begin{proof}
    The objective function of the primal SDP is 
    \begin{align}
        \frac{1}{\bar{P}|\cS|}\sum_{\psi_i\in\cS}\tr[\psi_i \Tilde{\rho}_i] &= \frac{1}{\bar{P}|\cS|}\sum_{\psi_i\in\cS}\tr[\psi_i \cE(\cN(\psi_i)^{\ox n})]\\
        &= \frac{1}{\bar{P}|\cS|}\sum_{\psi_i\in\cS}\tr\left[J_{\cE_{A^nB^nA'B'}}^{T_{A^nB^n}} \cN(\psi_i)^{\ox n}\ox \psi_i\right]\\
        &= \frac{1}{\bar{P}}\tr\left[J_{\cE_{A^nB^nA'B'}}^{T_{A^nB^n}} \sum_{\psi_i\in\cS}\frac{1}{|\cS|} \cN^{\ox n}\ox \cI \left(\psi_i^{\ox n + 1}\right)\right]\\
        &= \frac{1}{\bar{P}}\tr\left[Q_{A^nB^nA'B'} J_{\cE_{A^nB^nA'B'}}^{T_{A^nB^n}}\right].
    \end{align}
    Similarly, we arrive at the average success probability constraint as shown in Eq.~\eqref{eq:SDP_optimal_prob}. Eq.~\eqref{eq:SDP_optimal_CPTN} and Eq.~\eqref{eq:SDP_optimal_PPT} ensure the purification protocol $\cE$ satisfies CPTN and PPT constraints, respectively. Eq.~\eqref{eq:SDP_optimal_all_state} guarantees that for every state in the set $\cS$, the purified state by the CPTN map $\cE$ is less noisy.
\end{proof}

\subsection{Dual problem}
According to the primal problem, we have the Lagrange function
\begin{align}
    \cL(&J_{\cE_{A^nB^nA'B'}}, x, K_{A^nB^n}, L_{A^nB^nA'B'}, y_1, y_2, \cdots, y_{|\cS|})\nonumber\\
    = &\frac{1}{\bar{P}} \tr[Q_{A^nB^nA'B'} J_{\cE_{A^nB^nA'B'}}^{T_{A^nB^n}}] + x\langle \tr[R_{A^nB^nA'B'} J_{\cE_{A^nB^nA'B'}} ^{T_{A^nB^n}}] - \bar{P}\rangle
    + \langle K_{A^nB^n}, I_{A^nB^n} - \tr_{A'B'} [J_{\cE_{A^nB^nA'B'}}] \rangle \nonumber\\
    & + \langle L_{A^nB^nA'B'}, J_{\cE_{A^nB^nA'B'}}^{T_{A^nA'}}\rangle + \sum_{\psi_i\in \cS} y_i \left(\tr[J_{\cE_{A^nB^nA'B'}} ^{T_{A^nB^n}} \cN^{\ox n}\ox \cI(\psi_i^{\ox n+1})]\right)\nonumber\\
    & - \sum_{\psi_i\in \cS} y_i \left( \tr[\psi_i \cN(\psi_i)] \tr[J_{\cE_{A^nB^nA'B'}} ^{T_{A^nB^n}} \cN^{\ox n}(\psi_i^{\ox n})\ox I_{A'B'}]\right)\\
    =& -\bar{P}x + \langle K_{A^nB^n}\rangle +\Big\langle J_{\cE_{A^nB^nA'B'}}, \frac{1}{p} Q_{A^nB^nA'B'} ^{T_{A^nB^n}} + x R_{A^nB^nA'B'}^{T_{A^nB^n}} -K_{A^nB^n}\ox I_{A'B'} +L_{A^nB^nA'B'}^{T_{A^nA'}} \nonumber\\
    & + \sum_{\psi_i\in\cS} y_i \cN^{\ox n}\ox \cI (\psi_i^{\ox n+1}) ^{T_{A^nB^n}}- \sum_{\psi_i\in\cS} y_i \tr[\psi_i \cN(\psi_i)] \cN^{\ox n} (\psi_i^{\ox n}) ^{T} \ox I_{A'B'} \Big\rangle,
\end{align}
where $x, K_{A^nB^n}, L_{A^nB^nA'B'}, y_1,\cdots, y_{|\cS|}$ are Lagrange multipliers. The Lagrange dual function is 
$$g(x, K_{A^nB^n}, L_{A^nB^nA'B'}, \mathbf{y}) = \sup_{J_\cE \ge 0} \cL(J_{\cE_{A^nB^nA'B'}}, x, K_{A^nB^n}, L_{A^nB^nA'B'}, \mathbf{y}),$$
where $\mathbf{y}=[y_1, y_2, \cdots, y_{|\cS|}]$. The variables must satisfies $K_{A^nB^n}\ge 0$, $L_{A^nB^nA'B'}\ge 0$, and $\mathbf{y} \geq 0$.
Otherwise, there is no feasible solution to the primal problem. Besides, since $J_{\cE_{A^nB^n}}\ge 0$, it holds that 
\begin{align}
    \frac{1}{\bar{P}} Q_{A^nB^nA'B'} ^{T_{A^nB^n}} +& x R_{A^nB^nA'B'}^{T_{A^nB^n}} -K_{A^nB^n}\ox I_{A'B'} +L_{A^nB^nA'B'}^{T_{A^nA'}} + \sum_{\psi_i\in\cS} y_i \cN^{\ox n}\ox \cI (\psi_i^{\ox n+1}) ^{T_{A^nB^n}} \nonumber\\
    &- \sum_{\psi_i\in\cS} y_i \tr[\psi_i \cN(\psi_i)] \cN^{\ox n} (\psi_i^{\ox n}) ^{T} \ox I_{A'B'} \le 0,
\end{align}
otherwise the dual problem is unbounded. Thus, we arrive at the dual SDP as
\begin{subequations}
    \begin{align}
        \min \; &\langle K_{A^nB^n}\rangle-\bar{P}x \\
        \text{s.t.}\;& K_{A^nB^n}\ge 0;\; L_{A^nB^nA'B'}\ge 0;\; \mathbf{y} \geq 0;\\
        &\frac{1}{\bar{P}} Q_{A^nB^nA'B'} ^{T_{A^nB^n}} + x R_{A^nB^nA'B'}^{T_{A^nB^n}} -K_{A^nB^n}\ox I_{A'B'} +L_{A^nB^nA'B'}^{T_{A^nA'}} + \sum_{\psi_i\in\cS} y_i \cN^{\ox n}\ox \cI (\psi_i^{\ox n+1}) ^{T_{A^nB^n}} \nonumber\\
    \qquad& - \sum_{\psi_i\in\cS} y_i \tr[\psi_i \cN(\psi_i)] \cN^{\ox n} (\psi_i^{\ox n}) ^{T} \ox I_{A'B'} \le 0.
    \end{align}
\end{subequations}
%%%%%%%%%%%%%%%%%%%%%%%%%%%%%%%%%%%%%%%%%%%%%%%%%%%%%%%%%%
\section{Equivalence between global and local depolarizing noises}\label{app:equ_of_global_local}
\begin{lemma}\label{lemma:local_global_depo}
    Applying depolarizing channels locally on maximally entangled states $\cN_A^{\gamma_1}\ox \cN_B^{\gamma_2}$ is equivalent to applying global depolarizing channel $\cN_{AB}^{\gamma'}$ with parameter $\gamma' = 1-(1-\gamma_1)(1-\gamma_2)$.
\end{lemma}
\begin{proof}
    Denote $\Phi_{AB}^{MAX}$ as an arbitrary maximally entangled state. Due to the property of maximally entangled states, the partial traced density matrix is $\rho_A = \rho_B = I/d$, where $d$ is the dimension of the subsystems $A$ and $B$, i.e., $d_A = d_B = d$. When apply the local depolarizing noise to $\Phi_{AB}^{MAX}$, we have
    \begin{align}
        \cN_A^{\gamma_1}\ox \cN_B^{\gamma_2}(\Phi_{AB}^{MAX}) &= (1-\gamma_1)(1-\gamma_2) \Phi_{AB}^{MAX} + \gamma_1(1-\gamma_2)\frac{I_{AB}}{d_{AB}} +  \gamma_2(1-\gamma_1)\frac{I_{AB}}{d_{AB}} +  \gamma_1 \gamma_2\frac{I_{AB}}{d_{AB}}\nonumber\\
        &= (1-\gamma_1)(1-\gamma_2) \Phi_{AB}^{MAX} + (\gamma_1 + \gamma_2 - \gamma_1 \gamma_2) \frac{I_{AB}}{d_{AB}}\nonumber\\
        & = (1-\gamma_1)(1-\gamma_2) \Phi_{AB}^{MAX} + (1-(1-\gamma_1)(1-\gamma_2)) \frac{I_{AB}}{d_{AB}}\nonumber\\
        &= (1-\gamma') \Phi_{AB}^{MAX} + \gamma' \frac{I_{AB}}{d_{AB}} \\
        &= \cN_{AB}^{\gamma'}(\Phi_{AB}^{MAX}) \nonumber,
    \end{align}
    where $\gamma' = 1-(1-\gamma_1)(1-\gamma_2)$. The proof is completed.
\end{proof}

%%%%%%%%%%%%%%%%%%%%%%%%%%%%%%%%%%%%%%%%%%%%%%%%%%%%%%%%%%

\section{Proof of Theorem~\ref{theo:no-go_4_bell}}\label{app:proof_of_no-go_4_bell}

\vspace{0.5cm}
\noindent\textbf{Theorem~\ref{theo:no-go_4_bell}} {\rm \textbf{(No-go theorem for four Bell states)}}
\textit{For the local depolarizing noise channel $\cN_{AB}^{\gamma_1,\gamma_2}$ with noise levels $\gamma_1$ and $\gamma_2$, there is no nontrivial $2\rightarrow 1$ LOCC purification protocol for the set that contains 4 Bell states $\cS_B$.}
\vspace{0.5cm}

\begin{proof}
Note that arbitrary four orthogonal maximally entangled pure states can be achieved by applying local unitaries on the four Bell states. Thus, it is equivalent to proving there is no nontrivial purification protocol for the set that contains four Bell states $\cS_B = \{\Phi^+, \Phi^-, \Psi^+,\Psi^- \}$.

The main idea of this proof is to show that the optimal $2\rightarrow 1$ PPT protocol is trivial, and since LOCC is a subset of PPT operations, we arrive at the no-go theorem for the LOCC protocol. From Proposition~\ref{prop:SDP_for_discrete_set}, the optimal fidelity can be achieved by the following SDP.
\begin{subequations}
    \begin{align}
        \bar{F}_{\rm PPT}^*(n=2, \bar{P};\cS_B, \cN_{AB}^{\gamma_1,\gamma_2}) = \max\; & \frac{1}{\bar{P}} \tr[Q_{A^nB^nA'B'} J_{\cE_{A^nB^nA'B'}}^{T_{A^nB^n}}];\\
        \text{s.t.}\; & \tr[R_{A^nB^nA'B'} J_{\cE_{A^nB^nA'B'}} ^{T_{A^nB^n}}] = \bar{P};\\
        & J_{\cE_{A^nB^nA'B'}} \ge 0; \; \tr_{A'B'} [J_{\cE_{A^nB^nA'B'}}] \le I_{A^nB^n};\\
        & J_{\cE_{A^nB^nA'B'}}^{T_{A^nA'}} \ge 0;\\
        & \frac{\tr[J_{\cE_{A^nB^nA'B'}} ^{T_{A^nB^n}} \cN_{AB}^{\gamma_1,\gamma_2 \ox n}\ox \cI(\psi_i^{\ox n+1})]}{\tr[J_{\cE_{A^nB^nA'B'}} ^{T_{A^nB^n}} \cN_{AB}^{\gamma_1,\gamma_2 \ox n}(\psi_i^{\ox n})\ox I_{A'B'}]} \ge \tr[\psi_i \cN(\psi_i)] \; \forall \; \psi_i\in\cS_B,
    \end{align}
\end{subequations}
where $Q_{A^nB^nA'B'} = \frac{1}{|\cS_B|}\sum_{\psi_i\in\cS_B}\cN_{AB}^{\gamma_1,\gamma_2\ox 2}\ox \cI(\psi_i ^{\ox 3})$ and $R_{A^nB^nA'B'} = \frac{1}{|\cS_B|}\sum_{\psi_i\in\cS_B}\cN_{AB}^{\gamma_1,\gamma_2\ox 2}(\psi_i ^{\ox 2})_{A^nB^n} \ox I_{A'B'}$. The noise channel $\cN$ is local depolarizing channel, i.e., $\cN_{AB}^{\gamma_1,\gamma_2} = \cN_A^{\gamma_1}\ox\cN_B^{\gamma_2}$. Since the input states are all maximally entangled states, according to Lemma~\ref{lemma:local_global_depo}, it is equivalent to replace the local depolarizing noise channel with the global depolarizing channel, which is $\cN_{AB}^{\gamma_1,\gamma_2}=\cN_{AB}^{\gamma'}$ with $\gamma' = 1-(1-\gamma_1)(1-\gamma_2)$. Obviously, the partial trace channel, is one feasible solution to the primal SDP, i.e., 
\begin{equation}
        J_{\cE_{A^nB^nA'B'}} = \bar{P}\cdot\Phi^+_{A_1A'}\ox\Phi^+_{B_1B'}\ox I_{A^nB^n/A_1B_1}.
    \end{equation} 
In this case, we have
\begin{equation}\label{eq:discrete_trivial_direction}
    \bar{F}_{\rm PPT}^*(2, \bar{P};\cS_B, \cN_{AB}^{\gamma'})\ge 1-\frac{3}{4} \gamma',
\end{equation}
which is the same as the average fidelity without purification, implying the partial trace channel is a trivial solution. 

Besides, we also need to prove $\bar{F}_{\rm PPT}^*(n, \bar{P};\cS_B, \cN_{AB}^{\gamma'})\le 1-\frac{3}{4} \gamma'$ by the following dual SDP.
\begin{subequations}
    \begin{align}
        \min \; &\langle K_{A^nB^n}\rangle-\bar{P}x \\
        \text{s.t.}\;& K_{A^nB^n}\ge 0;\; L_{A^nB^nA'B'}\ge 0;\; y_i \geq 0,\; \forall i;\label{eq:proof_SB_dual_1}\\
        &\frac{1}{\bar{P}} Q_{A^nB^nA'B'} ^{T_{A^nB^n}} + x R_{A^nB^nA'B'}^{T_{A^nB^n}} -K_{A^nB^n}\ox I_{A'B'} +L_{A^nB^nA'B'}^{T_{A^nA'}}\nonumber\\
    \qquad& + \sum_{\psi_i\in\cS_B} y_i \cN_{AB}^{\gamma'\ox n}\ox \cI (\psi_i^{\ox n+1}) ^{T_{A^nB^n}} - \sum_{\psi_i\in\cS_B} y_i \tr[\psi_i \cN_{AB}^{\gamma'}(\psi_i)] \cN_{AB}^{\gamma'\ox n} (\psi_i^{\ox n}) ^{T} \ox I_{A'B'} \le 0 \label{eq:proof_SB_dual_2}.
    \end{align}
\end{subequations}
We note that 
\begin{subequations}\label{eq:proof_SB_dual_feasible}
    \begin{align}
    K_{A^nB^n} = 0;\;  x = -\frac{1}{\bar{P}}\left(1-\frac{3}{4}\gamma'\right);\; &y_i=0,\; \forall i;\\
    L_{A^nB^nA'B'}= \frac{\gamma'(1-\gamma')(1-3\gamma'/4 )}{8\bar{P}}&\big[(\Phi^+_{A_1B_1}\Phi^-_{A_2B_2} + \Phi^-_{A_1B_1}\Phi^+_{A_2B_2})(\Psi^+ + \Psi^-)_{A'B'}\nonumber\\
    + &(\Psi^+_{A_1B_1}\Psi^-_{A_2B_2} + \Psi^-_{A_1B_1}\Psi^+_{A_2B_2})(\Phi^+ + \Phi^-)_{A'B'}\nonumber\\
    +&(\Phi^+_{A_1B_1}\Phi^-_{A'B'} + \Phi^-_{A_1B_1}\Phi^+_{A'B'})(\Psi^+ + \Psi^-)_{A_2B_2}\nonumber\\
    +&(\Psi^+_{A_1B_1}\Psi^-_{A'B'} + \Psi^-_{A_1B_1}\Psi^+_{A'B'})(\Phi^+ + \Phi^-)_{A_2B_2}\nonumber\\
    +& (\Phi^+_{A_2B_2}\Phi^-_{A'B'} + \Phi^-_{A_2B_2}\Phi^+_{A'B'})(\Psi^+ + \Psi^-)_{A_1B_1}\nonumber\\
    +& (\Psi^+_{A_2B_2}\Psi^-_{A'B'} + \Psi^-_{A_2B_2}\Psi^+_{A'B'})(\Phi^+ + \Phi^-)_{A_1B_1}
    \big],
\end{align}
\end{subequations}
is a feasible solution to the dual SDP. The constraints in Eq.~\eqref{eq:proof_SB_dual_1} can be easily checked. For the constraint shown in Eq.~\eqref{eq:proof_SB_dual_2}, substitute $L_{A^nB^nA'B'}$ into the left hand side of ~\eqref{eq:proof_SB_dual_2}, and the largest non-zero eigenvalue is $\frac{\gamma'(\gamma'-1)}{16\bar{P}}\le 0$ when $0\le \gamma'\le 1$. Thus, Eq.~\eqref{eq:proof_SB_dual_feasible} is a feasible solution to the dual SDP, and we have $\bar{F}^*_{\rm PPT}(2,\bar{P};\cS_B, \cN_{AB}^{\gamma'})\le 1-\frac{3}{4}\gamma'$. Combing with Eq.~\eqref{eq:discrete_trivial_direction}, we have
\begin{equation}
    \bar{F}_{\rm PPT}^*(2,\bar{P};\cS_B, \cN_{AB}^{\gamma'}) = 1-\frac{3}{4}\gamma'.
\end{equation}
Here, we have that the optimal PPT purification protocol is trivial. Note that the LOCC is a subset of the PPT operations; we directly have the optimal LOCC purification protocol is trivial, implying no nontrivial LOCC protocol can purify four orthogonal maximally entangled states. The proof is completed.
\end{proof}

%%%%%%%%%%%%%%%%%%%%%%%%%%%%%%%%%%%%%%%%%%%%%%%%%%%%%%%%%%%%%%
\section{Proof of Theorem~\ref{theo:go_theorem_for_arbitrary_state}}\label{app:proof_of_go_single_state}

We divide the proof of Theorem~\ref{theo:go_theorem_for_arbitrary_state} into two parts. In the first part, we show that the protocol shown in FIG.~\ref{fig:protocol_for_single_state_2} could purify the quantum states in the form of \update{$\ket{\psi_\alpha} = \alpha\ket{00}+\sqrt{1-\alpha^2}\ket{11}$} corrupted by local depolarizing noise $\cN_{AB}^{\gamma, \gamma} = \cN_{A}^\gamma\ox\cN_{B}^\gamma$, where $\alpha\in[0,1],\; \beta=\sqrt{1-\alpha^2}$ (Lemma~\ref{lemma:nontrivial protocol_for_a00b11}). Specifically, Alice and Bob possess two copies of shared local depolarizing noisy states $\cN_{AB}^{\gamma, \gamma} (\psi_\alpha)$; they apply RX gates and CNOT gates on their own systems. Then, Alice and Bob make measurements on systems $A_2, B_2$ and denote the measurement results as $m_A$ and $m_B$, respectively. Then they exchange the results via classical communication. If $m_A = m_B = 0$, then the purification is successful and outputs the purified state $\sigma_\psi$; Otherwise, the purification fails, and the process needs to be repeated.

\begin{figure}[h]
\centering
\includegraphics[width=0.5\linewidth]{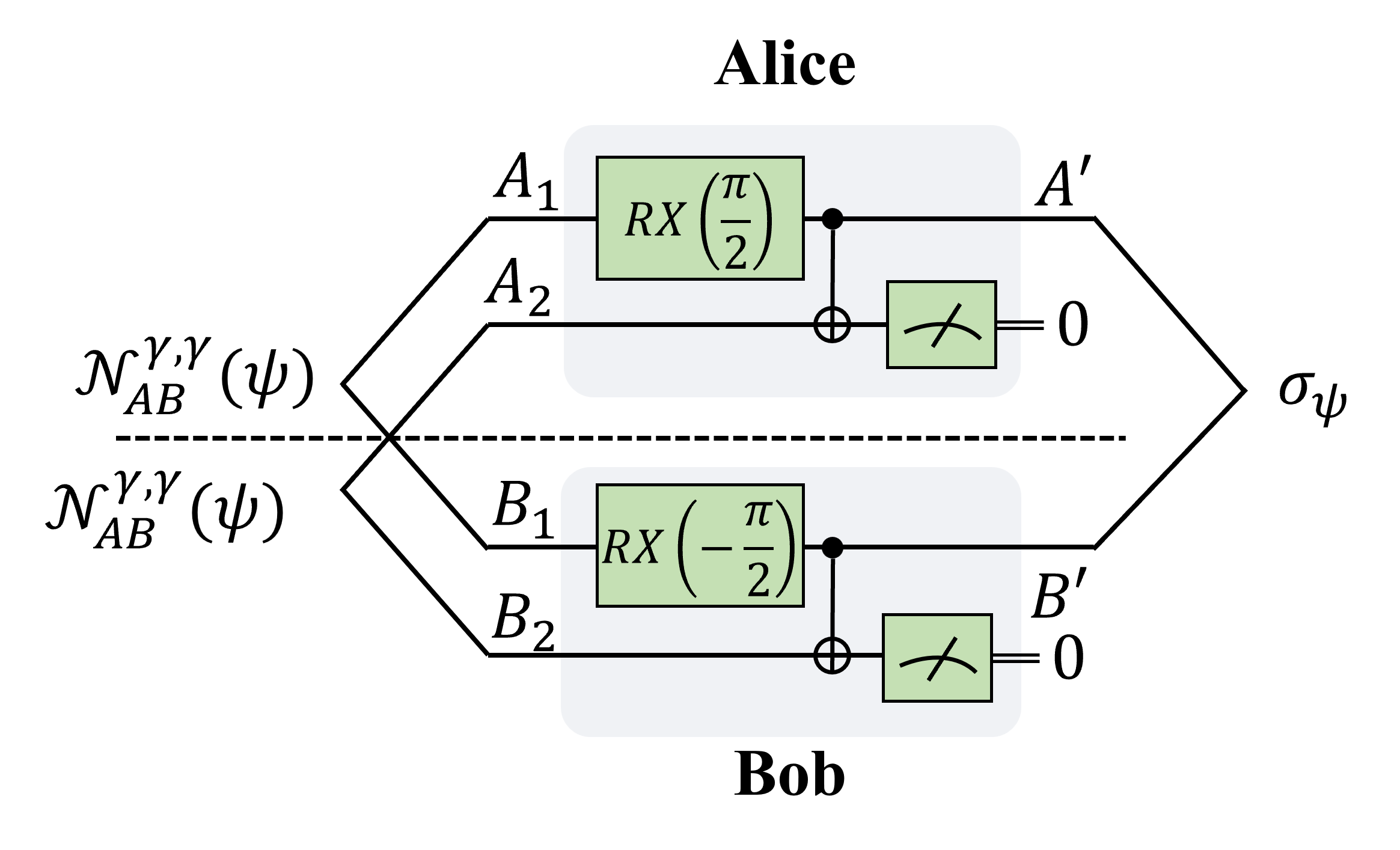}
    \caption{$2\rightarrow1$ purification protocol for noisy state $\cN_{AB}^{\gamma, \gamma}(\psi_\alpha)$, where $\ket{\psi_\alpha} = \alpha\ket{00} + \sqrt{1-\alpha^2}\ket{11}$. Alice and Bob apply $RX$ and CNOT gates on their own part, respectively. Then both parties make measurements on the second registers, and exchange the measurement results. If the measurements are '00', then the purification is successful; Otherwise, the purification fails and needs to repeat the processes.}
    \label{fig:protocol_for_single_state_2}
\end{figure}

\begin{lemma}\label{lemma:nontrivial protocol_for_a00b11}
    If the quantum state in the form $\ket{\psi_\alpha} = \alpha\ket{00} + \sqrt{1-\alpha^2}\ket{11}$ with $\alpha\in[0,1]$, the proposed protocol is a nontrivial $2 \rightarrow 1$ distribute purification protocol $\cE$ such that $F(\psi_\alpha, \cN_{AB}^{\gamma, \gamma}(\psi_\alpha)) \le F(\psi_\alpha, \cE(\cN_{AB}^{\gamma, \gamma}(\psi_\alpha)^{\ox 2}))$, where $\cN_{AB}^{\gamma, \gamma}$ is the local depolarizing noise with noise level $\gamma\in[0,0.4]$.
\end{lemma}

\begin{proof}
    Here we will prove that the fidelity between the target state $\ket{\psi_\alpha} = \alpha\ket{00} + \beta\ket{11}$ and the purified state $\cE(\cN_{AB}^{\gamma, \gamma}(\psi_\alpha)^{\ox2})$, where $\cE$ is the proposed $2\rightarrow 1$ is larger than no-purification fidelity. In other words, the aim is to prove
    \begin{equation}
        F(\psi_\alpha, \cE(\cN_{AB}^{\gamma, \gamma}(\psi_\alpha))^{\ox 2}) - F(\psi_\alpha, \cN_{AB}^{\gamma, \gamma}(\psi_\alpha)) \ge 0.
    \end{equation}
    By applying the protocol $\cE$ proposed in FIG.~\ref{fig:protocol_for_single_state_2}, the increased fidelity can be calculated, thus our aim becomes to prove
    \begin{equation}
        F(\psi_\alpha, \cE(\cN_{AB}^{\gamma, \gamma}(\psi_\alpha))^{\ox 2}) - F(\psi_\alpha, \cN_{AB}^{\gamma, \gamma}(\psi_\alpha))=\frac{f(\alpha,\gamma)g(\alpha,\gamma) }{d(\alpha, \gamma)}\ge0,
    \end{equation}
    where
    \begin{align}
        d(\alpha, \gamma) &= 2(1+2\alpha\sqrt{1-\alpha^2}(\gamma-1)^4),\\
        f(\alpha,\gamma) &= \alpha\gamma(\gamma-2)(\gamma-1)^2,\\
        g(\alpha,\gamma) &= 4\alpha-4\alpha^3-\sqrt{1-\alpha^2}(\gamma-2)^2+8\alpha^2\sqrt{1-\alpha^2}\gamma(\gamma-1)(\alpha^2-1).
    \end{align}
    Since $\alpha\in[0,1]$ and $\gamma\in[0,0.4]$, we can easily arrive at $d(\alpha, \gamma)>0$ and $f(\alpha, \gamma)\le 0$. Now, we need to prove $g(\alpha, \gamma)\le 0$. Fix the variable $\alpha$ and take the derivative over $\gamma$, which is
    \begin{align}
        \frac{\partial g(\alpha, \gamma)}{\partial \gamma} &=     \sqrt{1-\alpha^2}[4-2\gamma+8\alpha^2(\alpha^2-1)(2\gamma-1)].
    \end{align}
    Due to the domain of $\gamma$ and $\alpha$, we directly have $\frac{\partial g(\alpha, \gamma)}{\partial \gamma}\ge 0$, which implies the function $g(\alpha,\gamma)$ increase with respect to the increase of $\gamma$. Thus, the max value of the function $g(\alpha,\gamma)$ is achieved when $\gamma=0.4$.  If the largest value of $g(\alpha,\gamma)$ is no larger than 0, then the proof is finished. Here, the problem becomes proving
    \begin{align}
        g(\alpha, 0.4) = -\frac{4}{25}(25\alpha(\alpha^2-1) + \sqrt{1-\alpha^2}(12\alpha^4-12\alpha^2+16)) \le 0.
    \end{align}
     It is equivalent to proving
     \begin{equation}
         k(\alpha) = 25\alpha(\alpha^2-1) + \sqrt{1-\alpha^2}(12\alpha^4-12\alpha^2+16) \ge 0.
     \end{equation}
     Consider the end points $k(0)=16$, $k(1)=0$. Consider $\alpha\in(0,1)$, we need to prove
     \begin{align}
        25\alpha(\alpha^2-1) + \sqrt{1-\alpha^2}(12\alpha^4-12\alpha^2+16)& \ge 0,\\
        \sqrt{1-\alpha^2}(12\alpha^4-12\alpha^2+16)&\ge 25\alpha(1-\alpha^2).
    \end{align}
    Since $25\alpha(1-\alpha^2)>0$ and $\sqrt{1-\alpha^2}(12\alpha^4-12\alpha^2+16)>0$, we can take square on both side
    \begin{align}
        [\sqrt{1-\alpha^2}(12\alpha^4-12\alpha^2+16)]^2& \ge [25\alpha(1-\alpha^2)]^2.
    \end{align}
    Note that $1-\alpha^2\ge0$, we can reduce it to
    \begin{align}
        (12\alpha^4-12\alpha^2+16)^2& \ge 625\alpha^2(1-\alpha^2).
    \end{align}
     
    Let's $u=\alpha^2\in(0,1)$, and define the a function $p(u)=(12u^2-12u+16)^2- 625u(1-u)$. Now, our aim becomes to prove $p(u)\ge0$ for $u\in(0,1)$. By taking the first and second derivatives of $p(u)$, we found that $u=0.5$ is the only root for $p'(u)$, implying the minimum value for $p(u)$ is $p(0.5)=12.75>0$. Thus, $p(u)\ge0,\; u\in(0,1)$, which implies $g(\alpha, \gamma)\le0$. Together with $d(\alpha, \gamma)>0$ and $f(\alpha, \gamma)\le0$, we arrive at
    \begin{equation}
         F(\psi_\alpha, \cE(\cN_{AB}^{\gamma,\gamma}(\psi_\alpha))^{\ox 2}) - F(\psi_\alpha, \cN_{AB}^{\gamma,\gamma}(\psi_\alpha))=\frac{f(\alpha,\gamma)g(\alpha,\gamma) }{d(\alpha, \gamma)}\ge0.
    \end{equation}
    The proof is complete.
\end{proof}

In the second part, we are going to prove that an arbitrary bipartite pure state $\ket{\psi}$ is equivalent to the $\ket{\psi_\alpha}$ up to the local unitaries, i.e., $\exists\, U,V,\alpha$ such that $\ket{\psi} = U\ox V \ket{\psi_\alpha}$. Then we can transfer the state into the form of $\ket{\psi_\alpha}$ via local unitaries first, and then we can apply the protocol shown in FIG.~\ref{fig:protocol_for_single_state_2} can be used to purify the noisy state.

\vspace{0.5cm}
\noindent\textbf{Theorem~\ref{theo:go_theorem_for_arbitrary_state}} {\rm \textbf{(Go theorem for arbitrary state)}}
    \textit{For the local depolarizing noise channel $\cN_{AB}^{\gamma, \gamma}$ with noise level $\gamma\in[0,0.4]$, the proposed $2\rightarrow 1$ analytical distributed purification protocol is nontrivial for arbitrary pure states.}
\vspace{0.5cm}

\begin{proof}
    For arbitrary given state $\ket{\psi}$, which can represent it in the Schmidt decomposition form $\ket{\psi} = \alpha\ket{u_1}_A \ket{v_1}_B + \beta\ket{u_2}_A \ket{v_2}_B$, the protocol $\cU^\dagger\circ\cE\circ\cU^{\ox 2}(\cdot)$ is a nontrivial distributed purification protocol over the local depolarizing noisy states, where $\cE$ is the protocol shown in FIG.~\ref{fig:protocol_for_single_state_2}, and $\cU(\cdot)=U_A\ox U_B (\cdot) U_A^\dagger\ox U_B^\dagger$ is unitary channel which maps the state to the form of $\alpha\ket{00}+\beta\ket{11}$ with $U_1 = \ketbra{0}{u_1} + \ketbra{1}{u_2}$ and $U_2 = \ketbra{0}{v_1} + \ketbra{1}{v_2}$. The fidelity between the purified state and the target state is 
    \begin{align}
        F(\psi, \cU^\dagger\circ\cE\circ\cU^{\ox 2 }(\cN_{AB}^{\gamma, \gamma}(\psi)^{\ox 2})) &= \tr[\psi \cU^\dagger\circ\cE\cU^{\ox 2}(\cN_{AB}^{\gamma, \gamma}(\psi)^{\ox 2})]\\
        &= \tr[\psi U_A^\dagger \ox U_B^\dagger \cE[U^{\ox 2}_A \ox U^{\ox 2}_B(\cN_{AB}^{\gamma, \gamma}(\psi)^{\ox 2}) U^{\ox 2 \dagger}_A \ox U^{\ox 2 \dagger}_B]U_A \ox U_B]\\
        &= \tr[U_A \ox U_B \psi U_A^\dagger \ox U_B^\dagger \cE[U^{\ox 2}_A \ox U^{\ox 2}_B(\cN_{AB}^{\gamma, \gamma}(\psi)^{\ox 2}) U^{\ox 2 \dagger}_A \ox U^{\ox 2\dagger}_B]].
    \end{align}
    Note that the depolarizing channel commutes with local unitary, 
    \begin{align}
        U\cdot \cN^\gamma(\psi) \cdot U^\dagger &= U\cdot [(1-\gamma)\psi + \gamma\frac{I}{d})]\cdot U^\dagger
        = (1-\gamma)U\psi U^\dagger + \gamma\frac{I}{d}\\
        &= \cN^\gamma(U\psi U^\dagger).
    \end{align}
    Thus, the fidelity becomes
    \begin{align}
        F(\psi, \cU^\dagger\circ\cE\circ\cU^{\ox 2 }(\cN_{AB}^{\gamma, \gamma}(\psi)^{\ox 2}))
        &= \tr[U_A \ox U_B \psi U_A^\dagger \ox U_B^\dagger \cE[U^{\ox 2}_A \ox U^{\ox 2}_B(\cN_{AB}^{\gamma, \gamma}(\psi)^{\ox 2}) U^{\ox 2 \dagger}_A \ox U^{\ox 2\dagger}_B]]\\
        &=\tr[U_A \ox U_B \psi U_A^\dagger \ox U_B^\dagger \cE[(\cN_{AB}^{\gamma, \gamma}(U_A \ox U_B \psi U_A^\dagger \ox U_B^\dagger )^{\ox 2})]]\\
        &= \tr[\psi' \cE[(\cN_{AB}^{\gamma, \gamma}(\psi')^{\ox 2})]]\\
        &= F(\psi', \cE(\cN_{AB}^{\gamma, \gamma}(\psi')^{\ox 2})),
    \end{align}
    where $\psi'= U_A \ox U_B \psi U_A^\dagger \ox U_B^\dagger = \alpha\ket{00} + \beta\ket{11}$. According to Lemma~\ref{lemma:nontrivial protocol_for_a00b11}, we directly arrive at
    \begin{align}
        F(\psi, \cU^\dagger\circ\cE\circ\cU^{\ox 2 }(\cN_{AB}^{\gamma, \gamma}(\psi)^{\ox 2})) &= F(\psi', \cE(\cN_{AB}^{\gamma, \gamma}(\psi')^{\ox 2})) \\
        &\ge F(\psi', \cN_{AB}^{\gamma, \gamma}(\psi')) = F(\psi, \cN_{AB}^{\gamma, \gamma}(\psi)).
    \end{align}
The proof is complete.
\end{proof}
%%%%%%%%%%%%%%%%%%%%%%%%%%%%%%%%%%%%%%%%%%%%%%%%%%%%%%%%%%%%%%%
\section{Optimized purification protocol}\label{app:VQA}

\begin{figure*}[h]
\centering
\includegraphics[width=0.7\linewidth]{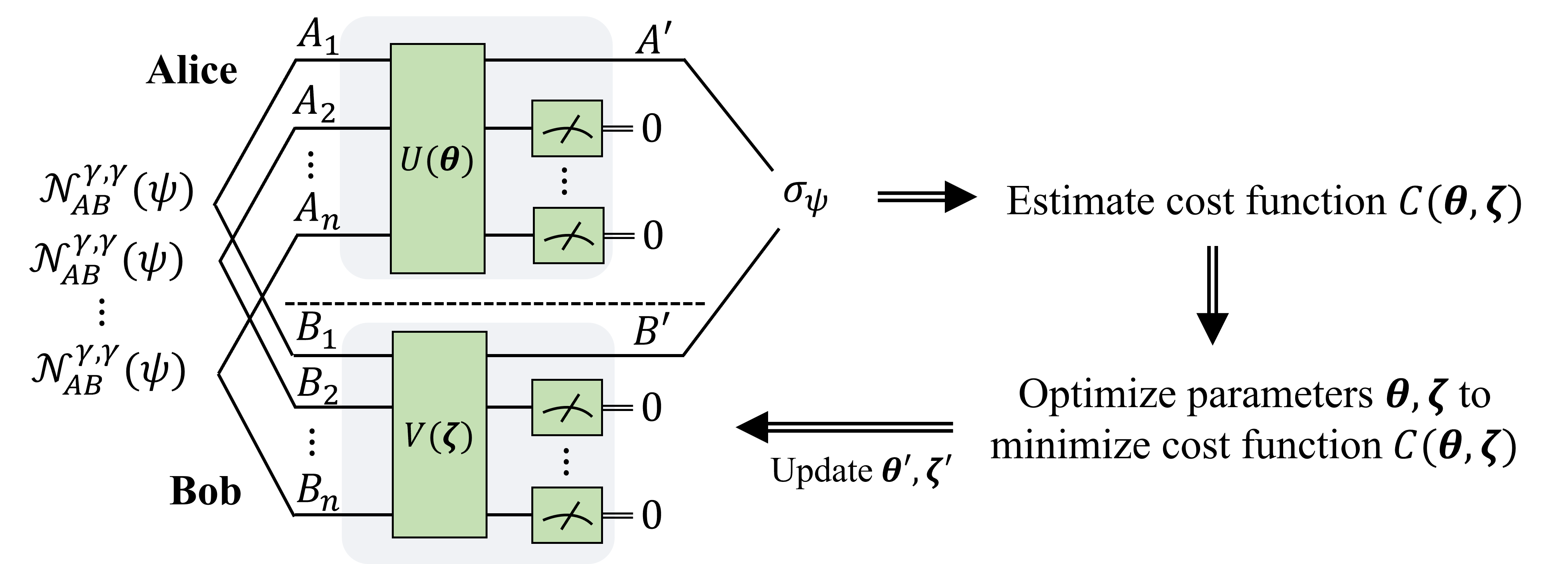}
    \caption{Framework of the optimization purification algorithm. Alice and Bob apply local unitaries $U(\boldsymbol{\theta})$ and $V(\boldsymbol{\zeta})$ on their systems, respectively. Then they make measurements on the systems and check the measurement results via classical communication. If the results are all 0, then output the purified state $\sigma_\psi$; otherwise, repeat the process until it succeeds. With the purified state $\sigma_\psi$, we can estimate the cost function $C(\boldsymbol{\theta}, \boldsymbol{\zeta})$. The cost function can be minimized by utilizing a classical optimizer, updating the parameters to the PQCs, and repeating the iterations until the cost function converges to its minimum.}
    \label{fig:VQA}
\end{figure*}

\begin{algorithm}[t]
    \renewcommand{\algorithmicrequire}{\textbf{Input:}}
    \renewcommand{\algorithmicensure}{\textbf{Output:}}
    \begin{algorithmic}[1]
        \REQUIRE Noisy state $\cN(\psi)$ with $\psi$ is sampled from the state $\cS$, two parameterized quantum circuits (PQC) $U(\boldsymbol{\theta})$ and $V(\boldsymbol{\zeta})$, number of iterations ITR, 
        \ENSURE Quantum circuit for LOCC purification protocol $U(\boldsymbol{\theta}^*) \ox V(\boldsymbol{\zeta}^*)$
        \STATE Randomly initialize parameters $\boldsymbol{\theta},\boldsymbol{\zeta} $, 
        and prepare the circuit ansatz $U(\boldsymbol{\theta}),V(\boldsymbol{\zeta})$ 
        \FOR {itr = 1, $\dots$, ITR}
            \FOR {state $\psi\in \cS$}
                \STATE Apply $U(\boldsymbol{\theta})$ and $V(\boldsymbol{\zeta})$ to the noisy states $\cN(\psi)^{\ox n}$
                \STATE Make measurements on the ancillary qubit
                \IF{Measurement results are all 0}
                \STATE The purification is successful and outputs the purified state $\sigma_\psi$
                \ELSE
                \STATE Repeat the process until it succeeds.
                \ENDIF
                \STATE Compute the fidelity $F(\psi, \sigma_\psi)$
            \ENDFOR
            \STATE Estimate the cost function $C(\boldsymbol{\theta, \zeta})$ as shown in Eq.~\eqref{eq:cost_function}
            \STATE Minimize the cost function $C$ and update parameter $\boldsymbol{\theta', \zeta'}$ via classical computers
        \ENDFOR
        \STATE Output the optimized circuit for distributed purification protocol $U(\boldsymbol{\theta}^*) \ox V(\boldsymbol{\zeta}^*)$
    \end{algorithmic}
\caption{Optimization algorithm for designing a distributed purification protocol}
\label{algo:VQA}
\end{algorithm}

Note that the success case can be defined flexibly. For example, one can define the success case as measurements of ancillary systems are all 1 or all 0. In our case, we only define the measurement as all 0 as the successful case. The algorithm may achieve higher fidelity by \update{introducing} ancilla qubits.

The choice of the unitary ansatz should be made carefully, for the trade-off between trainability and expressibility. If the ansatz we choose is very complex and the number of parameters is large, it may express more quantum states. However, at the same time, it will be harder to train, preventing us from achieving the optimal result. On the contrary, if the ansatz is simple and the number of parameters is small, then it is very easy to train. But the global optimal results may not be expressed by the ansatz. Meanwhile, the barren plateaus, referring to a phenomenon in which the gradient decreases exponentially with respect to the increase of quantum system size, is another important problem we have to take into consideration. Choosing a suitable quantum ansatz and setting the starting parameters carefully can also circumvent barren plateaus to some extent~\cite{liu2024mitigating, jing2025quantum}. As we focus on the $2\rightarrow 1$ purification protocol, we choose the ansatz as \textit{universal 2-qubit gate}~\cite{vidal2004universal, shende2004minimal}, which can represent an arbitrary 2-qubit unitary with 15 parameters. 
%%%%%%%%%%%%%%%%%%%%%%%%%%%%%%%%%%%%%%%%%%%%%%%%%%%%%%%%%%%%%%
\section{Distributed purification protocol for the set $\cS_d$ with extended noise levels}\label{app:lager_noise_level}

\update{
In the numerical experiments shown in FIG.~\ref{fig:purified_fidelity_local_Depo}, we consider the pure state set containing 4 different states $\cS_d=\{\psi_{\frac{1}{\sqrt{2}}}, \psi_{\frac{1}{\sqrt{3}}}, \psi_{\frac{1}{\sqrt{4}}}, \psi_{\frac{1}{\sqrt{5}}}\}$, 
where $\ket{\psi_\alpha}=\alpha\ket{00}+\sqrt{1-\alpha^2}\ket{11}$. The noise model is bi-local depolarizing channel $\cN_{AB}^{\gamma, \gamma}$, with noise level varying $\gamma\in[0,0.4]$. However, if we extend the noise level $\gamma$ to 0.5, the behavior of the PPT bound becomes complicated. The numerical experiments are conducted, and the results are shown in Fig.~\ref{fig:0-0.5}. One can observe that the PPT bound jumps around $\gamma\approx0.43$, and then the average fidelity remains at 0.71 as the noise level increases. It is natural to ask the following two questions:
}

\begin{itemize}
    \item \update{\textbf{Question:} Why there is a jump point around 0.43?}
    
    \update{\textbf{Answer:} 
    When $\gamma \ge 0.43$, if the purification protocol is successful, the corresponding operation is the replacement channel, which maps any input state to the a fixed state, i.e., $\frac{\cE(\cN(\rho)\ox\cN(\rho))}{\tr[\cE(\cN(\rho)\ox\cN(\rho))]} = \ketbra{11}{11}$ (We call it replacement protocol for convenience in the following). Set the noise level as $\gamma=0.43$. Before purification, the fidelities $F(\cN(\psi_\alpha), \psi_\alpha)$ of the state set $\cS_d$ are $\{0.4937, 0.5073,0.5243,0.5378\}$, respectively, while the purified state fidelity $F(\sigma_{\psi_\alpha}, \psi_\alpha)$ becomes $\{0.5, 0.6667,0.75,0.8\}$, where $\sigma_{\psi_\sigma}$ refers to the purified state. By the definition of purification protocol, for any state in the set, the purified state should have a larger fidelity, i.e., 
    \begin{equation}
        F(\sigma_\psi,\psi) \ge F(\cN(\psi), \psi), \; \forall\, \psi\in \cS_d.
    \end{equation}
    Therefore, the replacement protocol is a valid purification protocol. The optimal purified average fidelity is calculated by 
    \begin{equation}
        \bar{F} = \sum_\alpha p_\alpha F(\sigma_{\psi_\alpha}, \psi_\alpha),
    \end{equation}
    with $\sum_\alpha p_\alpha/4= \bar{P}=0.1$. 
    By SDP, the success probabilities $p_\alpha$ are given by $\{0.0569, 0.0941, 0.1169, 0.1321\}$, and the optimal average fidelity is calculated to be $\bar{F}=0.7113$.}
    
    \update{However, when $\gamma\le 0.42$, the replacement protocol shown above is no longer feasible. For example, let's take $\gamma=0.42$, the fidelities before purification are $\{0.5023, 0.5158, 0.5327, 0.5461\}$, respectively. Clearly, if we apply the PPT purification protocol directly, the fidelity of the first output state drops, i.e., 
    \begin{equation}
        F(\sigma_{\psi_{\frac{1}{\sqrt{2}}}}, \psi_{\frac{1}{\sqrt{2}}}) = 0.5 < 0.5023 = F(\cN({\psi_{\frac{1}{\sqrt{2}}}}), \psi_{\frac{1}{\sqrt{2}}}),
    \end{equation}
    which contradicts the definition of the purification protocol, i.e., the replacement protocol is not feasible for $\gamma=0.42$. That's why there is a jump around $\gamma=0.43$.}

    \item \update{\textbf{Question:} Why cannot the optimization method converge to the replacement protocol as PPT does?}
    
    \update{\textbf{Answer:} The replacement channel is basically a measure-and-prepare operation. The ansatz we used in the optimization algorithm can only optimize LOCC protocols consisting of local unitaries; the measure-and-prepare operation is outside the scope of the chosen ansatz. If we apply other kinds of ansatz, e.g., with ancilla systems and partial trace operation, it is possible for the optimized protocol to converge to the replacement protocol.}
\end{itemize}

\begin{figure}[h]
\centering
\includegraphics[width=0.5\linewidth]{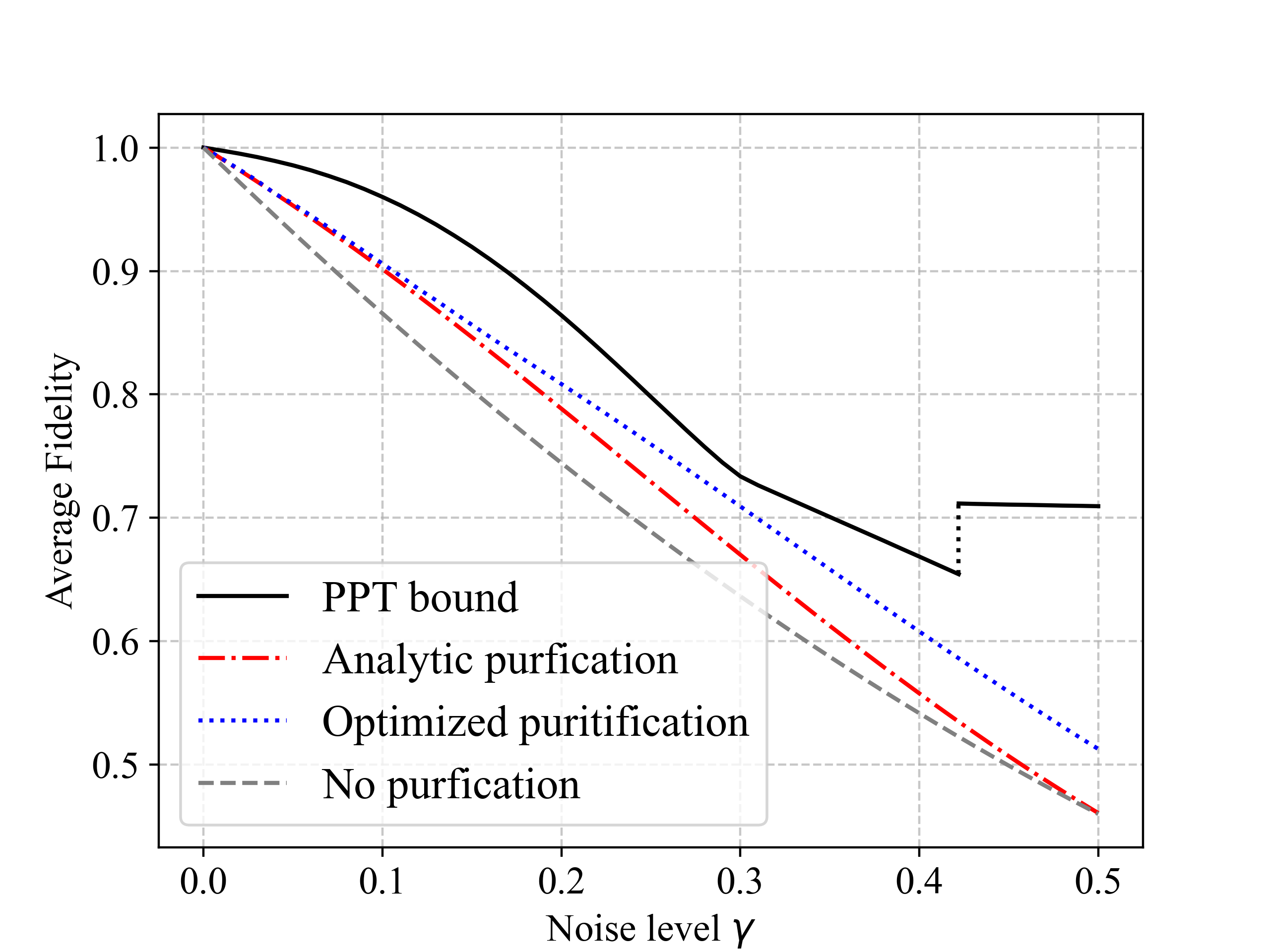}
    \caption{Performance of $2\rightarrow1$ purification with different strategies for the set $\cS_d$ over bi-local depolarizing noise with noise level $\gamma\in[0,0.5]$.}
    \label{fig:0-0.5}
\end{figure}

\update{Note that a similar fidelity jump is also observed in the purification of the set $\cS_d$ degraded by the mixed Pauli noise channel $\cN_{\rm Pauli}^\gamma\ox\cN_{\rm Pauli}^\gamma$, where
\begin{equation}
    \cN_{\rm Pauli}^\gamma(\rho) = (1-\gamma) \rho + \frac{\gamma}{2} X\rho X + \frac{\gamma}{4} Y\rho Y + \frac{\gamma}{4} Z\rho Z.
\end{equation}
The PPT bound for the $2\rightarrow1$ protocol is shown in FIG.~\ref{fig:Pauli_noise}
}.

\begin{figure}[h]
\centering
\includegraphics[width=0.5\linewidth]{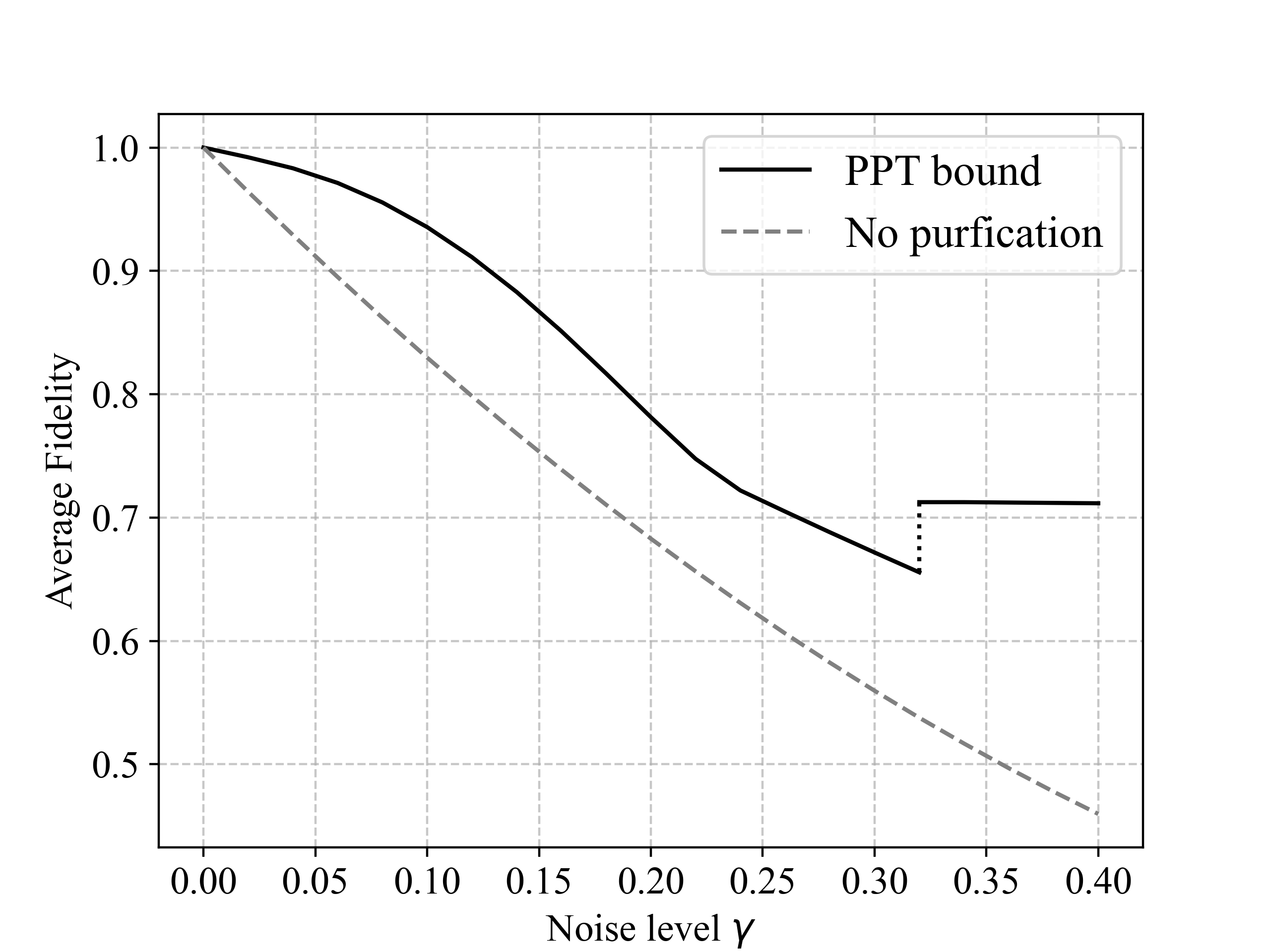}
    \caption{Performance of $2\rightarrow1$ purification with different strategies for the set $\cS_d$ over mixed Pauli noise $\cN_{\rm Pauli}^\gamma\ox\cN_{\rm Pauli}^\gamma$.}
    \label{fig:Pauli_noise}
\end{figure}

\update{However, the jump point does not always exist. It is highly dependent on the state set and the quantum noise. For example, if the set is the four Bell states, i.e., $\cS_B$, and the quantum noise is amplitude damping noise, then there is no such jump point, as illustrated in FIG.~\ref{fig:AD noise for Bell}.}

%%%%%%%%%%%%%%%%%%%%%%%%%%%%%%%%%%%%%%%%%%%%%%%%%%%%%%%%%%%%%%
\section{Additional numerical results}\label{app:more_experiments}
\update{In the main text, we only consider purifying the depolarizing noise corrupted states with $2\rightarrow1$ LOCC purification protocols. However, the tools proposed in this work, such as the SDP for the bound of purification and the protocol design with an optimization-based method, can be applied to a more general scenario. Here, we conduct more numerical experiments to demonstrate the power of tools utilized in this work.}

\subsection{$2\rightarrow 1$ LOCC purification protocol for sets with Bell states}

\update{Here, we consider 4 different state sets which contain 4, 3, 2, and 1 Bell states respectively, i.e., 
$\cS_{B_4}=\{\Phi^+, \Phi^-, \Psi^+,\Psi^-\}$,
$\cS_{B_3}=\{\Phi^+, \Phi^-, \Psi^+\}$,
$\cS_{B_2}=\{\Phi^+, \Psi^+\}$,
and $\cS_{B_1}=\{\Phi^+\}$. The noise we consider is a bi-local depolarizing noise model. The numerical results are shown in TABLE~\ref{tab:Depo_Bell}. In this table, ``No Pure" refers to the averaged fidelity between the noised state and target state, i.e., $\frac{1}{|\cS|}\sum_{\psi\in\cS} F(\cN(\psi), \psi)$. The ``PPT Bound" is computed via the SDP given in Eq.~\eqref{eq:SDP_for_PPT_bound} with the average successful probability $\bar{P}=0.1$. The ``Optimized" refers to the averaged fidelity computed by the protocols, which are designed by the proposed optimization-based method as shown in Algorithm~\ref{algo:VQA}.}

\update{
The numerical results demonstrate the following:
\begin{itemize}
    \item For bi-local depolarizing noise $\cN_{DE}^\gamma\ox\cN_{DE}^\gamma$ with noise level $\gamma\in[0,0.4]$, there is no non-trivial $2\rightarrow 1$ LOCC purification protocol for the set with three or four Bell states, i.e., $\cS_{B_3}$ and $\cS_{B_4}$. 
    \item  For bi-local depolarizing noise $\cN_{DE}^\gamma\ox\cN_{DE}^\gamma$ with noise level $\gamma\in[0,0.4]$, there exists a non-trivial $2\rightarrow 1$ LOCC purification protocol for the set with two or one Bell states, i.e., $\cS_{B_2}$ and $\cS_{B_1}$.
    \item For the set with two Bell states, i.e., $\cS_{B_2}$, the protocol designed by the optimization-based method is near-optimal as the fidelity achieved by the optimized protocol is very close to the PPT bound. 
\end{itemize}
These conclusions shown here are inferred from numerical results. Rigorous proofs are necessary to make them solid.
}

\begin{table}[!htbp]
\centering
\begin{tabular}{|c|c|c|c|c|c|c|c|c|c|c|}
\hline
 & \multicolumn{3}{|c|}{$\cS_{B_1}$}& \multicolumn{3}{|c|}{$\cS_{B_2}$} & \multicolumn{2}{|c|}{$\cS_{B_3}$} & \multicolumn{2}{|c|}{$\cS_{B_4}$}\\ 
\hline
Noise Level $\gamma$ & No Pure & PPT Bound& Optimized & No Pure & PPT Bound & Optimized & No Pure & PPT Bound & No Pure & PPT Bound\\
\hline
0 & 1.0 & 1.0 & 1.0 & 1.0 & 1.0 & 1.0 & 1.0 & 1.0 & 1.0 & 1.0\\
\hline
0.1 & 0.8575 & 0.9731 & 0.8907 & 0.8575 & 0.8907 & 0.8907 & 0.8575 & 0.8575 & 0.8575 & 0.8575\\
\hline
0.2 & 0.7300 & 0.8797 & 0.76759 & 0.7300 & 0.7676 & 0.76759 & 0.7300 & 0.7300 & 0.7300 & 0.7300\\
\hline
0.3 & 0.6175 & 0.7227 & 0.64118 & 0.6175 & 0.6412 & 0.64117 & 0.6175 & 0.6175 & 0.6175 & 0.6175\\
\hline
0.4 & 0.5200 & 0.5399 & 0.5241 & 0.5200 & 0.5241 & 0.5241 & 0.5200 & 0.5200 & 0.5200 & 0.5200 \\
\hline
\end{tabular}
\caption{State purification for the sets with Bell states.}
\label{tab:Depo_Bell}
\end{table}

\subsection{$2\rightarrow 1$ LOCC purification protocol for Bell set $\cS_{B}$ over other noise models}
\update{Here, we conduct another numerical experiment to show the performance of state purification for the Bell set $\cS_B=\{\Phi^\pm, \Psi^\pm\}$ over the bi-local amplitude damping noise $\cN_{AD}^\gamma\ox\cN_{AD}^\gamma$. The Amplitude damping noise $\cN^\gamma_{AD}$ is particularly relevant for the loss of energy or the dissipation of excited states, which is characterized by the damping rate $\gamma$, and has two Kraus operators: }
\begin{equation}
    A_0^\gamma = \ketbra{0}{0} + \sqrt{1-\gamma} \ketbra{1}{1},\; A_1^\gamma=\sqrt{\gamma}\ketbra{0}{1}
\end{equation}
\update{The numerical results are shown in FIG.~\ref{fig:AD noise for Bell}. The PPT bound is computed via the SDP given in Eq.~\eqref{eq:SDP_for_PPT_bound} with the average successful probability $\bar{P}=0.1$. The optimized purification refers to the averaged fidelity computed by the protocols, which are designed by the proposed optimization-based method as shown in Algorithm~\ref{algo:VQA}. From the numerical results, we observe that for bi-local amplitude damping noise, unlike depolarizing noise, it is possible to purify the set with four Bell states via state purification.}

\begin{figure}[h]
\centering
\includegraphics[width=0.5\linewidth]{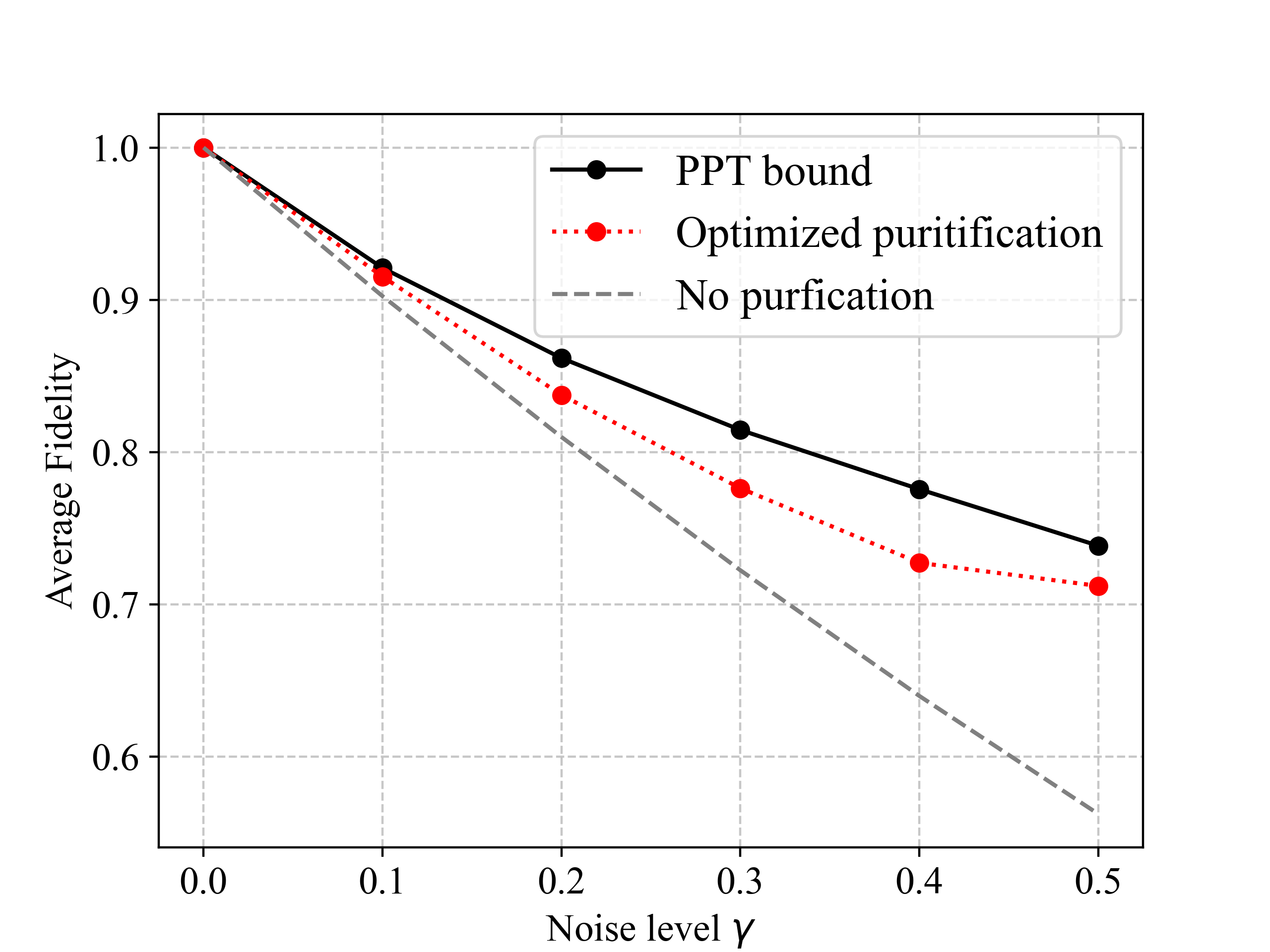}
    \caption{State purification for bi-local amplitude damping corrupted Bell set $\cS_B$.}
    \label{fig:AD noise for Bell}
\end{figure}
    
\update{We also take bi-local dephasing noise $\cN_{DP}^\gamma\ox\cN_{DP}^\gamma$ into consideration, which probabilistically applies the Pauli-Z operation on the quantum state, i.e., 
\begin{equation}
    \cN_{DP}^\gamma(\rho) = (1-\gamma)\rho + \gamma Z\rho Z.
\end{equation}
From the SDP numerical results, the optimal blind purification protocol cannot increase the fidelity at all, as shown in Table~\ref{tab:Dephasing noise}. Therefore, for the bi-local dephasing noise, there is no non-trivial LOCC (even PPT) $2\rightarrow 1$ state purification protocol for the set with four Bell states $\cS_B$.
\begin{table}[!htbp]
    \centering
    \begin{tabular}{|c|c|c|c|c|c|c|}
    \hline  
    Noise level $\gamma$ & 0 & 0.1 & 0.2 & 0.3 & 0.4& 0.5\\
    \hline  
    No purification & 
    1.000 & 0.905  &  0.820  &  0.745 &   0.680  & 0.625 \\
    \hline  
    PPT bound & 1.000 & 0.905  &  0.820  &  0.745 &   0.680  & 0.625 \\
    \hline 
    \end{tabular}
\caption{The PPT bound of purification for Bell set $\cS_B$ over the bi-local dephasing noise.}
\label{tab:Dephasing noise}
\end{table}
}

\subsection{Comparison between optimized protocol and DEJMPS protocol}

\update{To benchmark the performance of the optimized protocol, it is natural to make a comparison with conventional methods, such as DEJMPS~\cite{deutsch1996quantum} and BBPSSW~\cite{bennett1996purification}. Note that DEJMPS works on a broader set of states than BBPSSW. Therefore, we conduct numerical experiments with different settings to compare our protocol with DEJMPS.}

\begin{itemize}
    \item \update{First, we consider the purification of the canonical maximally entangled state $\Phi^+$, degraded by bi-depolarizing noise $\cN_{DE}^\gamma\ox\cN_{DE}^\gamma$ into a Bell-diagonal state. Here, we also conduct a numerical experiment to compare the protocol designed by the optimization-based method and the well-known DEJMPS protocol~\cite{deutsch1996quantum}. The results are shown in TABLE~\ref{tab:DEJMPS_vs_VQA}. The optimized protocol achieves the same fidelity as the DEJMPS protocol, implying the power of the proposed optimization-based method.}
    \begin{table}[!htbp]
        \centering
        \begin{tabular}{|c|c|c|c|c|}
        \hline  
        Noise level & 0.1 & 0.2 & 0.3 & 0.4\\
        \hline  
        No purification & 0.8575 & 0.7300 & 0.6175 & 0.5200 \\
        \hline  
        DEJMPS protocol & 0.8907 & 0.7676 & 0.6412 & 0.5241 \\
        \hline 
        Optimized protocol & 0.8907 & 0.7676 & 0.6412 & 0.5241 \\
        \hline
        \end{tabular}
    \caption{Comparison between the $2\rightarrow 1$ optimized purification protocol and the DEJMPS protocol over bi-depolarizing noise.}
    \label{tab:DEJMPS_vs_VQA}
    \end{table}
    
    \item  \update{Second, we consider the purification of the maximally entangled state $\Phi^+$ degraded by amplitude damping noise applied on both sides, i.e., $\cN_{AD}^\gamma\ox\cN_{AD}^\gamma$. In this case, from the numerical results shown in Fig.~\ref{fig:compare DEJMPS} (a), our optimized protocol outperforms DEJMPS. This result could be expected, because DEJMPS is tailored to Bell diagonal states, whereas the output of the amplitude damping channels is not Bell diagonal anymore.}
    \begin{figure}[h]
    \centering
    \includegraphics[width=\linewidth]{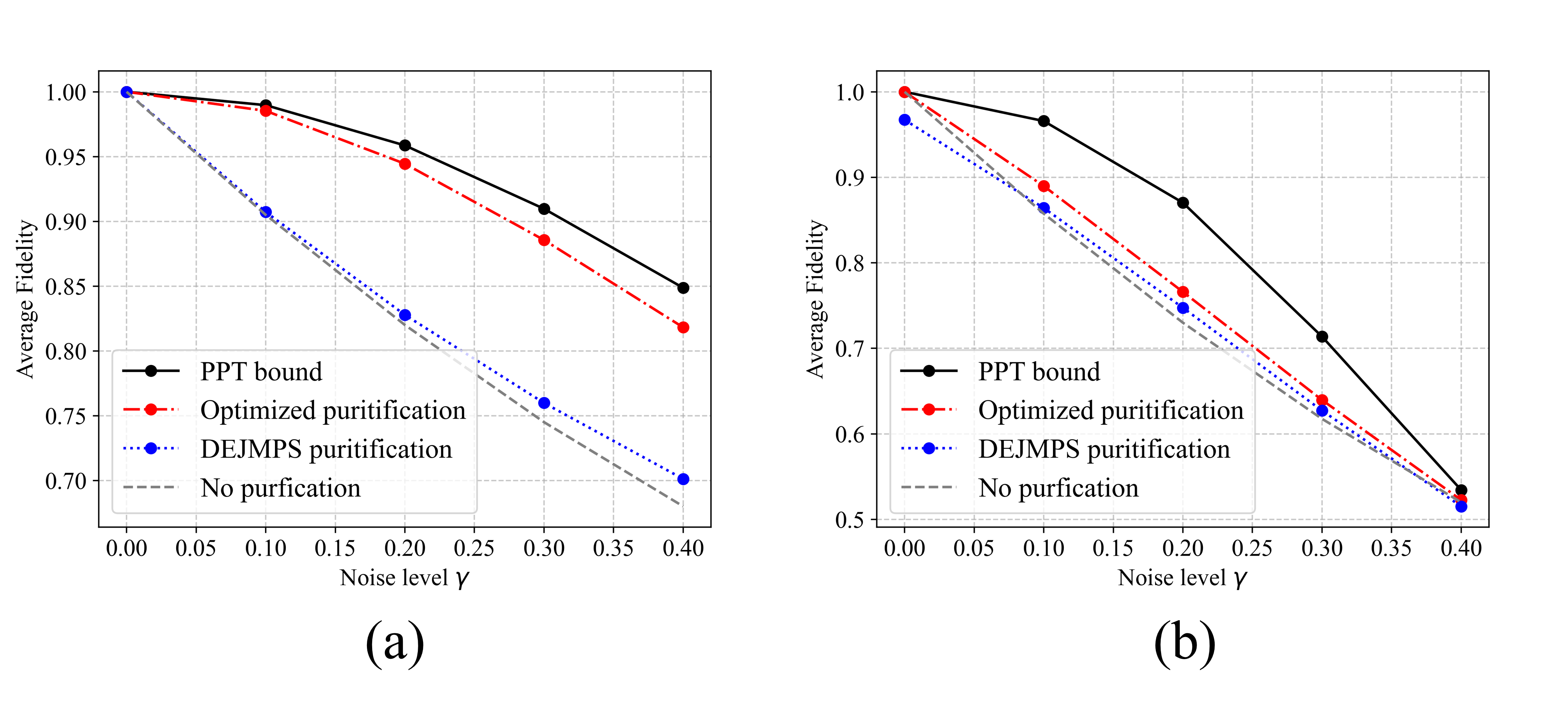}
        \caption{The comparison between the $2\rightarrow 1$ optimized purification protocol and the DEJMPS protocol. (a) The input state is Bell states corrupted by amplitude damping noise $\cN^{\gamma}_{AD}\ox \cN^{\gamma}_{AD}(\Phi^+)$. (b) The input states are sampled from the depolarized maximally entangled state set $\Phi^+, U_1\ox I (\Phi^+)U_1^\dagger\ox I, U_2\ox I (\Phi^+)U_2^\dagger\ox I\}$, where $U_1, U_2$ are rotation-X operations with rotation angle 0.2 and -0.2, respectively.}
        \label{fig:compare DEJMPS}
    \end{figure}

    \item \update{Finally, we consider an example involving the purification of more than a single maximally entangled state. Specifically, we choose the three maximally entangled states $\{\Phi^+, U_1\ox I (\Phi^+)U_1^\dagger\ox I, U_2\ox I (\Phi^+)U_2^\dagger\ox I\}$, where $U_1$ and $U_2$ are $X$-rotation gates with rotation angles 0.2 and -0.2, respectively, i.e., $U_1 = RX(0.2)$ and $U_1 = RX(-0.2)$. These three states are degraded by depolarizing noise applied to both systems. The results are shown in Fig.~\ref{fig:compare DEJMPS} (b). In this case, our protocol also outperforms DEJMPS. Again, this finding is no surprise, because the task is to purify a set of maximally entangled states simultaneously, rather than just purify the canonical maximally entangled state, and therefore is outside the scope for which DEJMPS was designed. }
\end{itemize}

\subsection{$3\rightarrow 1$ LOCC purification protocol for the set $\cS_d$ over depolarizing noise}

\update{The purification protocol via optimization proposed in Algorithm~\ref{algo:VQA} is actually valid for the $n\rightarrow 1$ case. In the main text, we only show the performance of $2\rightarrow 1$ protocols. Here, conduct another numerical experiment to show the performance of $3\rightarrow 1$ LOCC purification protocol. The state set we choose is $\cS_d=\{\psi_{\frac{1}{\sqrt{2}}}, \psi_{\frac{1}{\sqrt{3}}}, \psi_{\frac{1}{\sqrt{4}}}, \psi_{\frac{1}{\sqrt{5}}}\}$, and noise is bi-local depolarizing noise $\cN_{DE}^\gamma\ox\cN_{DE}^\gamma$.}

\begin{figure}[h]
\centering
\includegraphics[width=0.5\linewidth]{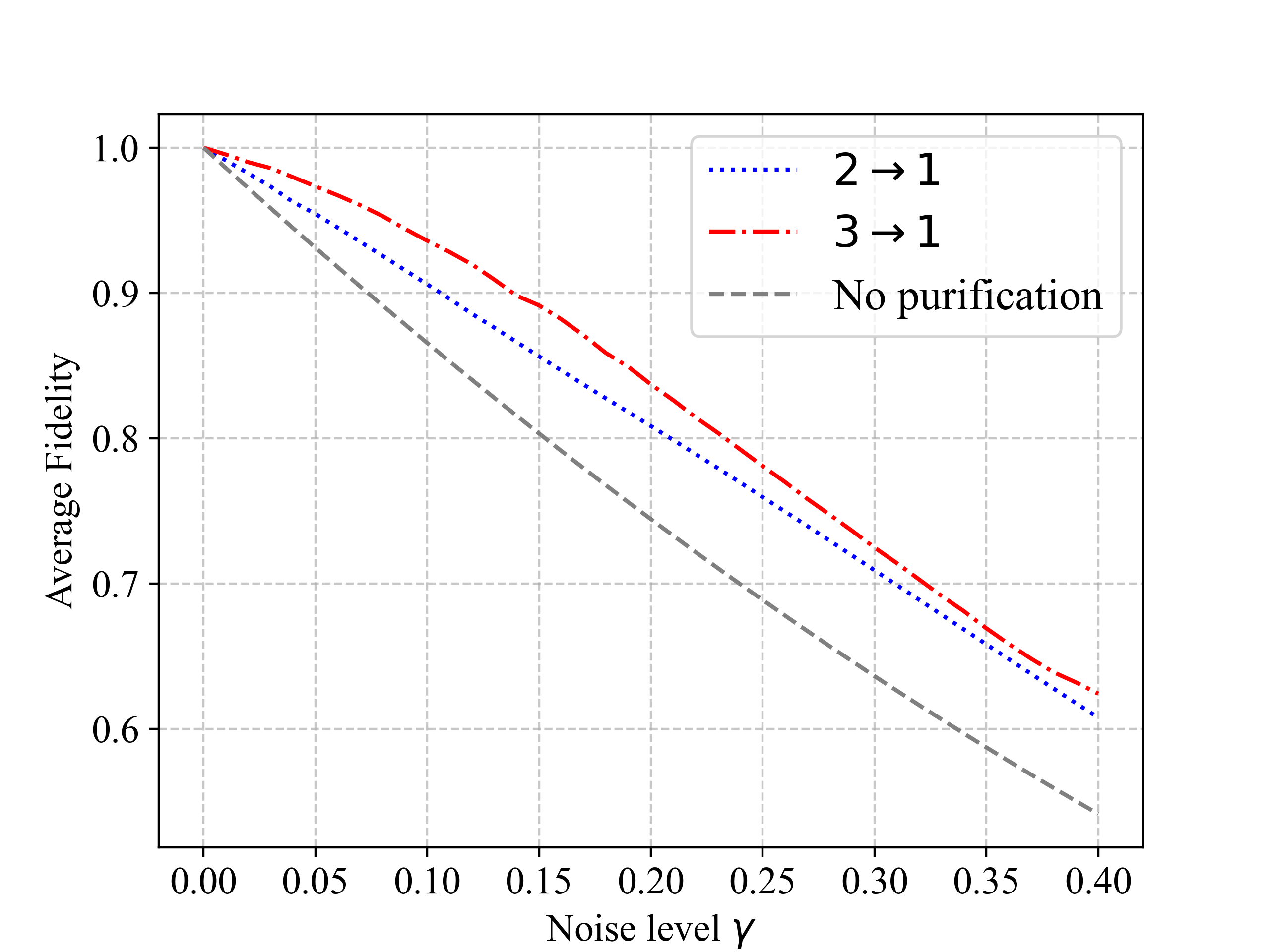}
    \caption{he performance comparison between $2\rightarrow1$ and $3\rightarrow 1$ purification protocols.}
    \label{fig:3-1}
\end{figure}

\update{From the numerical results shown in FIG.~\ref{fig:3-1}, the $3\rightarrow1$ protocol outperforms the $2\rightarrow1$ protocol, which is an intuitive result as $3\rightarrow1$ protocol consumes more resources. Theoretically, the more copies that are consumed in the purification protocol, the higher the fidelity that will be achieved. Of course, meanwhile, it is becoming harder to train the protocol via the optimization-based method.}

\end{document}